\documentclass[acmsmall]{acmart}
\usepackage{libertine}
\usepackage{stackengine}
\usepackage{listings}
\usepackage{color}
\usepackage{mathpartir}
\usepackage{mathtools}
\usepackage{stmaryrd}
\usepackage{wrapfig}
\usepackage{pgffor}
\usepackage{definitions}
\usepackage{xspace}
\usepackage{pifont} 
\usepackage{colortbl}
\usepackage{pgf}
\usepackage{tikz}
\usepackage{soul}
\usepackage{siunitx}

\usetikzlibrary{matrix}
\usetikzlibrary{arrows,automata,quotes,positioning}
\usepackage{longtable}
\usepackage{paralist}

\usepackage{bm}
\usepackage{xcolor}
\usepackage[capitalise]{cleveref}
\usepackage{pftools}
\usepackage{marvosym}
\usepackage{subcaption}

\newcommand\expGMevtrans{1.9}
\newcommand\expGMcoarmochi{1.6}
\newcommand\expGMrealmochi{6.7}
\newcommand\expGMrethfl{5.1}
\newcommand\expGMdirect{0.9}
\newcommand\expSpeedupEVoverDrift{6.3}
\newcommand\expSpeedupEVoverRcaml{5.3}
\newcommand\expSpeedupEVoverRethfl{6.4}
\newcommand\expSpeedupEVoverRealMochi{16.8}
\newcommand\expNewOverCoarMochi{15}
\newcommand\expNewOverRealMochi{10}
\newcommand\expNewOverRethfl{3}
\newcommand\expNewOverTrans{9}
\newcommand\expBenchCount{23}
\newcommand\expDriftVerified{12}
\newcommand\expEDriftVerified{21}
\newcommand\expRealMochiVerified{11}
\newcommand\expRcamlVerified{6}
\newcommand\expRethflVerified{18}

\newcommand\expBestOverviewTrans{1.8}

\newcommand\expBestAuctionEVDrift{2.7}
\newcommand\expBestAuctionTrans{55.1}
\newcommand\expBestAuctionRethfl{18.6}

\newcommand\expTPGMoffDrift{11.0}
\newcommand\expTPGMonDrift{17.3}
\newcommand\expTPGMoffevDrift{1.6}
\newcommand\expTPGMonevDrift{1.9}
\newcommand\expTPSpeedupDrift{0.6}
\newcommand\expTPSpeedupevDrift{0.8}

\newcommand\evdrift{{\bf ev}{\sc Drift}}
\newcommand\drift{{\sc Drift}}
\newcommand\rcaml{{\sc RCaml/Spacer}}
\newcommand\rcamlONLY{{\sc RCaml}}
\newcommand\mochi{{\sc MoCHi}}
\newcommand\coar{{\sc CoAR}}

\newcommand\rethfl{{\sc ReTHFL}}

\lstdefinelanguage{Lambda}{%
  numbers=none,
  morekeywords={if,then,else, ev},
  literate={
    {->}{{$\to$}}{2}
    {lambda}{{$\lambda$}}{1}
  },
  basicstyle=\small\ttfamily,
  keepspaces,
  mathescape 
}[keywords,comments,strings]%

\lstdefinelanguage{Alg}{%
  numbers=none,
  morekeywords={if,then,else, ev},
  literate={
    {->}{{$\to$}}{2}
    {lambda}{{$\lambda$}}{1}
  },
  basicstyle=\bfseries,
  keepspaces,
  mathescape 
}[keywords,comments,strings]%

\crefformat{section}{\S#2#1#3} 
\crefmultiformat{section}{\S#2#1#3}{ and~\S#2#1#3}{, \S#2#1#3}{, and~\S#2#1#3} 
\crefrangeformat{section}{\S#3#1#4 to~\S#5#2#6} 
\crefrangeformat{line}{Lines #3#1#4--#5#2#6}

\title[Abstract Interpretation of Temporal Safety Effects of Higher Order Programs]
{Abstract Interpretation of Temporal Safety Effects of Higher Order Programs}

\author{Mihai Nicola}
\orcid{0000-0003-0204-1626}
\affiliation{%
  \institution{Stevens Institute of Technology}
  \city{Hoboken}
  \country{USA}
}
\email{lnicola@stevens.edu}

\author{Chaitanya Agarwal}
\orcid{0009-0005-0921-697X}
\affiliation{%
  \institution{New York University}
  \city{New York, NY}
  \country{USA}
}
\email{ca2719@nyu.edu}

\author{Eric Koskinen}
\orcid{0000-0001-7363-634X}
\affiliation{%
  \institution{Stevens Institute of Technology}
  \city{Hoboken}
  \country{USA}
}
\email{eric.koskinen@stevens.edu}

\author{Thomas Wies}
\orcid{0000-0001-7363-634X}
\affiliation{%
  \institution{New York University}
  \city{New York, NY}
  \country{USA}
}
\email{wies@cs.nyu.edu}

\ccsdesc[500]{Theory of computation~Automated reasoning}
\ccsdesc[500]{Software and its engineering~Software verification}

\keywords{temporal verification, type-and-effect systems, dependent temporal effects}

\renewcommand{\shortauthors}{Nicola, Agarwal, Koskinen, and Wies}

\newcommand{\smartparagraph}[1]{\smallskip\noindent\emph{#1.}}

\setcopyright{cc}
\setcctype{by}
\acmDOI{10.1145/3763140}
\acmYear{2025}
\acmJournal{PACMPL}
\acmVolume{9}
\acmNumber{OOPSLA2}
\acmArticle{362}
\acmMonth{10}
\received{2025-03-25}
\received[accepted]{2025-08-12}

\begin{document}
\begin{abstract}

This paper describes a new \emph{abstract interpretation}-based approach to verify temporal safety properties of recursive, higher-order programs.
While prior works have provided theoretical impact and some automation, they have had limited scalability.
We begin with a new automata-based ``abstract effect domain'' for summarizing context-sensitive dependent effects, capable of abstracting relations between the program environment and the automaton control state.
Our analysis includes a new transformer for abstracting event prefixes to automatically computed context-sensitive effect summaries, and is instantiated in a type-and-effect system \added{grounded in abstract interpretation}.
Since the analysis is parametric on the automaton, we next instantiate it to a broader class of history/register (or ``accumulator'') automata, beyond finite state automata  to express some context-free properties, input-dependency, event summation, resource usage, cost, equal event magnitude, etc. 

We implemented a prototype \evdrift{} that computes dependent effect summaries (and validates assertions) for OCaml-like recursive higher-order programs. 
As a basis of comparison, we describe reductions to assertion checking for higher-order but effect-free programs, and demonstrate that our approach outperforms prior tools \drift{}, \rcaml{}, \mochi{}\added{, and \rethfl{}.
Overall, across a set of \expBenchCount{} benchmarks, 
\drift{} verified \expDriftVerified{} benchmarks,
\rcaml{} verified \expRcamlVerified{},
\mochi{} verified \expRealMochiVerified{},
\rethfl{} verified \expRethflVerified{}, and
\evdrift{} verified \expEDriftVerified{};
\evdrift{} also achieved a \expSpeedupEVoverDrift{}$\times$, \expSpeedupEVoverRcaml{}$\times$, \expSpeedupEVoverRealMochi{}$\times$, and \expSpeedupEVoverRethfl{}$\times$ speedup over \drift{}, \rcaml{}, \mochi{}, and \rethfl{}, respectively, on those benchmarks that both tools could solve.}

\end{abstract}
\maketitle


\section{Introduction}
\label{sec:intro}

The long tradition of temporal property verification has, in recent years, been also directed at programs written in languages with recursion and higher-order features.
In this direction, a first step was to go beyond simple types to dependent and/or refinement type systems~\cite{DBLP:conf/pldi/FreemanP91,DBLP:conf/popl/XiP99,DBLP:conf/popl/Terauchi10,DBLP:conf/pldi/RondonKJ08,DBLP:journals/pacmpl/PavlinovicSW21}, capable of validating merely (non-temporal) safety assertions.
Subsequently, works focused on verifying termination of higher-order programs, e.g., ~\cite{DBLP:conf/esop/KuwaharaTU014}.

As a next step, researchers focused on \emph{temporal} properties of higher-order programs. In this setting, programs have 
a notion of observable \emph{events} or \emph{effects}, 
typically emitted as a side effect of a program expression such as ``\texttt{ev} $e$,'' where $e$ is first reduced to a program value and then emitted. The semantics of the program is correspondingly augmented to reduce to a pair $(v,\pi)$, where $v$ is the value and $\pi$ is a sequence of events or an ``event trace.''
For such programs a natural question is whether the set of all event traces is included within a given temporal property expressed in Linear Temporal Logic~\cite{DBLP:conf/focs/Pnueli77}, or as an automaton.
Liveness properties apply to programs that may diverge, inducing infinite event traces.
A first approach at automated temporal verification was through the celebrated reduction to fair termination~\cite{DBLP:conf/lics/Vardi87}. \citet{DBLP:conf/popl/MuraseT0SU16} introduced a reduction from higher-order programs and LTL properties to termination of a calling relation.

In a parallel research trend, others have been exploring compositional type-and-effect theories for temporal verification. 
\citet{DBLP:conf/aplas/SkalkaS04} and \citet{DBLP:journals/jfp/SkalkaSH08} described a type-and-effect system to extract a finite abstraction of a program and then perform model-checking on that abstraction.
Later, \citet{DBLP:conf/csl/KoskinenT14} 
and
\citet{DBLP:conf/csl/0001C14}
showed that the effects component in a type-and-effect system
$\Gamma \vdash e : \tau \& \varphi$
could consist of a temporal property $\varphi$, where $\varphi$ holds of the events generated by the reduction of expression
$e$. This was combined with a dependent refinement system by~\citet{DBLP:conf/csl/KoskinenT14} and used with an abstraction of B\"{u}chi automata by~\citet{DBLP:conf/csl/0001C14}.
\citet{DBLP:conf/lics/Nanjo0KT18} then later gave a deductive proof system for verifying such temporal effects, even permitting the temporal effect expressions to depend on program inputs.
In a more distantly related line of research, others consider languages with programmer-provided ``algebraic effects'' and their handlers~\cite{DBLP:journals/pacmpl/SekiyamaU23,DBLP:journals/pacmpl/LagoG24} (see Sec.~\ref{sec:conclusion}).
\subsection{Better Automation Through Abstract Interpretation}

The first step of this paper is a new route to automate 
temporal effect inference and verification of 
recursive higher-order programs
through abstract interpretation.

As a preliminary step, we describe a \added{direct}\removed{na\"{i}ve} approach that reduces verification of such effect-full programs to verifying assertions of effect-less higher-order programs. \removed{We show that this can be done either by (i) translating the program so that every expression reduces to a (value, effect prefix) tuple, where the second component accumulates the effects as a sort of ghost return value, or (ii) translating the program into continuation-passing style, accumulating the effect prefix as an argument.}
We later experimentally show that, although this theoretically enables higher-order safety verifiers (e.g. \drift{}, \rcaml{}, \mochi{}, and \rethfl{}) to be applied to the effect setting, those tools do not exploit much of the property structure and ultimately struggle on the inherent \added{overhead}\removed{product construction} that comes from these transformations.

To achieve a more scalable solution, our core contribution is a novel \emph{effect abstract domain}. In the concrete semantics, an execution is simply the program execution environment paired with the event trace prefix that was thus far generated, i.e., an element of $(\cvalues^* \times \cenvs)$ where $\cvalues$ is the domain of program values and $\cenvs$ the domain of value environments. We first observe that both the environment and the possible trace prefix, somewhat counterintuitively, can be organized around the automaton control state. That is, an abstraction like $\mathcal{Q} \to \powerset(\cenvs)$
captures the possible execution environments that could be reachable at a control state $q \in \mathcal{Q}$ of the automaton. This 
control state-centric summary of environments enables the abstract domain to capture disjunctive invariants, guided by the target property of the verification. This abstraction often avoids the need for switching to a more expensive abstract domain that is closed under precise joins.
%
%
Having organized around control state, the final abstraction step is to associate with each $q$
a summary of the program environment, e.g. constraints like $\texttt{x}>\texttt{y}$.
The abstract domain \added{for summarizing program environments} can naturally be instantiated using any of a variety of standard numerical domains such as polyhedra~\cite{DBLP:conf/popl/SinghPV17,DBLP:journals/scp/BagnaraHZ08,
  DBLP:conf/popl/CousotH78}, octagons~\cite{DBLP:journals/corr/abs-cs-0703084}, etc.


\subsection{Better Expressiveness Through Accumulator Automata}

The effect abstract domain above turns out to be somewhat parametric over the kind of automaton, opening up another opportunity.

Specifically, there are many temporal safety properties that go beyond basic event sequencing properties especially, for example, if each event emits an integer. Examples include a property that the sum of the emitted integers is below some bound (e.g. resource analysis), or that the last emitted integer is the largest one. Properties could depend on inputs (see, e.g.~\citet{DBLP:conf/lics/Nanjo0KT18}), or involve context-free-like properties such as protocols of stateful \added{APIs~\cite{DBLP:journals/pacmpl/FerlesSD21}} or the sum of production being equal to the sum of consumption.



We support this wider class of temporal safety properties by augmenting our effect abstract domain automata
to symbolic accumulator automata (SAA).
Our automaton model is inspired by the various  notions of (symbolic) register or memory automata considered~\cite{DBLP:journals/tcs/KaminskiF94,DBLP:conf/cp/BeldiceanuCP04,DBLP:conf/cav/DAntoniFS019} and consists of a register ``accumulator'' (e.g., an integer or tuple of integers) that can remember earlier events, calculate summaries, etc. SAA is expressive enough to capture the  example properties above.

To instantiate SAA \removed{automata} in our framework, we refine the effect abstraction to
$\mathcal{Q} \to \powerset(\cvalues \times \cenvs)$, now
capturing the possible \emph{pairs of} accumulator value and execution environment that could be reachable at control state $q \in \mathcal{Q}$. Our abstraction thus associates with each $q$:
\begin{inparaenum}[(i)]
\item a summary of the program environment, e.g. constraints like $\texttt{x}>\texttt{y}$,
\item a summary of the automaton accumulator, e.g. constraints like $\texttt{acc}>0$,
and even
\item relations between the two, e.g. $\texttt{acc} > \texttt{x} - \texttt{y}$.
\end{inparaenum}
Thus, in this example, we capture at location $\ell$ in the program, that control state $q$ is reachable but only in a configuration where the accumulator is positive, the program variable \texttt{x} is greater than \texttt{y} and the accumulator bounds the difference between \texttt{x} and \texttt{y}.




\subsection{Challenges \& Contributions} 

To pursue the effect abstract domain, we address the following challenges in this paper:

\smartparagraph{Accumulative type and effect system (Sec.~\ref{sec:types-and-effects})} Our effect abstract domain, expressing properties of program expressions, is associated with the program through a type-and-effect system\removed{.
  Unfortunately existing type and effect systems
  are not suitable because their judgments of the form $\Gamma \vdash e : \tau \& \varphi$ do not involve event trace prefixes in their context: instead $\varphi$ describes the effects of $e$ alone, without information about what effects preceded the evaluation of $e$.
While this makes the type system more compositional, it also makes effect inference more difficult because one has to analyze $e$ for any possible prefix trace. In terms of our automata-based effect domain, it corresponds to analyzing $e$ for an arbitrary initial control state and accumulator value. (In the rest of this paper we will favor the more expressive accumulator automata, and forego accumulator-free DFAs.)
    We thus present a new effect system}, with judgments of the form 
    $\Gamma\;;\;\eff \vdash \term:\tandeff{\tau}{\eff'}$, where $\eff$ summarizes the prefix up to the evaluation of $\term$, $\eff'$ summarizes the \emph{extended} prefix with the evaluation of $\term$, and term-specific premises dictate how extensions are formed. The system is parametric in the abstract domains used to express dependent effects and dependent type refinements. \added{Our system resembles existing systems for sequential effects such as~\cite{DBLP:conf/ecoop/Gordon17,DBLP:journals/toplas/Gordon21} but is grounded in abstract interpretation to facilitate automated inference of types and effects.}
    
\smartparagraph{Effect abstract domain (Sec.~\ref{sec:acc-autom-eff})} We formalize the abstract domain discussed above as an instantiation of our effect system. 
    A key ingredient is the \emph{effect extension operator} $\effext$ that takes an abstraction of a reachable automaton configuration $\eff$, a type of a new event $\bval$ (we use refinement types for $\bval$ to capture precise information about the possible values of the event to extend a trace prefix), and produces an abstraction of the automaton configurations reachable by the extended trace. 
    The user-provided automata include symbolic error state conditions and so if the effect computed by the analysis associates error states with bottom, then the property encoded by the automaton holds of the program.
    Finally, we have proved the soundness of the effect abstract domain.

\smartparagraph{Automated inference of effects (Sec.~\ref{sec:inference})} We next address the question of automation. Recent work showed that, for programs \emph{without effects}, that abstract interpretation can be used to compute refinement types through a higher-order dataflow analysis~\cite{DBLP:journals/pacmpl/PavlinovicSW21}. We present an extension to \emph{effectful} programs through a translation-oriented embedding of programs with effects to effect-free programs and a specialized abstract transformer that exploits the structure of the translated programs and effect abstract domain. The resulting abstract interpretation propagates effects in addition to values through the program. To obtain the overall soundness of the inference algorithm, we show that the types inferred for the translated programs can be used to reconstruct a derivation in our type and effects system.

\smartparagraph{Verification, Implementation \& Benchmarks (Sec.~\ref{sec:implementation})}
    We implement the type system, effect abstract domain and abstract interpretation in a new tool \evdrift{} for OCaml-like recursive higher-order programs. Our implementation is an  extension of the \drift{} tool, which provides 
    assertion checking of effect-free programs. 
    There are no existing tools that can verify SAA properties of higher-order event-generating programs. Thus, in an effort to find the closest basis for comparison,
    we also implemented \added{a translation}\removed{two translations (one via encoding effects in tuples; another via encoding effects as a CPS parameter)} that reduce SAA verification of effect programs to assertion checking of effect-free programs (to which \drift{}, \rcaml{}, \mochi{}, \rethfl{}, etc. can be applied).
    To improve the precision of our abstract interpretation, we further adapted the classical notion of \emph{trace partitioning}~\cite{DBLP:conf/esop/MauborgneR05} to this higher-order effect setting.
    
    To date there are limited higher-order benchmark programs with properties that require an automaton with a register to express.
    We thus built the first suite of such benchmarks by creating \expBenchCount{} new examples and adapting examples from the literature including 
summation/max-min examples~\cite{DBLP:journals/tcs/KaminskiF94,DBLP:conf/cp/BeldiceanuCP04,DBLP:conf/cav/DAntoniFS019},
monotonicity examples,
programs with temporal event sequences~\cite{DBLP:conf/csl/KoskinenT14,DBLP:conf/popl/MuraseT0SU16,DBLP:conf/lics/Nanjo0KT18}, resource analysis~\cite{DBLP:journals/toplas/IgarashiK05,DBLP:conf/popl/HoffmannAH11,DBLP:conf/popl/HoffmannDW17},  
and an auction smart contract~\cite{DBLP:conf/sp/StephensFMLD21}.
    
\smartparagraph{Evaluation (Sec.~\ref{sec:evaluation})} We evaluated (i) the effectiveness of \evdrift{} at directly verifying SAA-expressible temporal safety properties over the use of \drift{}, \rcaml{}, \removed{and} \mochi{}\added{, and \rethfl{}} when applied via the translation/reduction to assertion checking, and (ii) the degree to which trace partitioning improves precision for \evdrift{}. 
\added{Overall, our approach is able to verify 
\expEDriftVerified{} out of the \expBenchCount{} benchmarks, which is \expNewOverTrans{}, \expNewOverCoarMochi{}, \expNewOverRealMochi{}, and \expNewOverRethfl{} more than \drift{}, \rcaml{}, \mochi{}, and \rethfl{}, respectively, (with our tuple translation) could verify.
Furthermore, \evdrift{} achieved a speedup of \expSpeedupEVoverDrift{}$\times$, \expSpeedupEVoverRcaml{}$\times$, \expSpeedupEVoverRealMochi{}$\times$, and \expSpeedupEVoverRethfl{}$\times$ over \drift{}, \rcaml{}, \mochi{}, and \rethfl{}, respectively, on those benchmarks that both tools could solve. }
\emph{The supplement to this paper includes the \evdrift{} source, all benchmark sources, and the Appendix.}

\section{Overview}
\label{sec:overview}

\newcommand\dd{\lstinline|d|}
\newcommand\cc{\lstinline|c|}
\newcommand\vvv{\lstinline|v|}
\newcommand\xx{\lstinline|x|}
\newcommand\nn{\lstinline|n|}
\newcommand\ttt{\lstinline|t|}
\newcommand\qq{\lstinline|q|}
\newcommand\kk{\lstinline|k|}
\newcommand\ttmain{\lstinline|main|}
\newcommand\ttacc{\lstinline|acc|}
\newcommand\ttbids{\lstinline|bids|}
\newcommand\ttrfds{\lstinline|rfds|}
\newcommand\xxamt{\lstinline|amt|}
\newcommand\xxi{\lstinline|i|}

This paper introduces a method for verifying properties of dependent effects of higher-order programs, through an abstraction that 
can express relationships between the (symbolic) next step of an automaton and the dependent typing context of the program at the location where a next event is emitted. We show that, when combining our approach with data-flow abstract interpretation~\cite{DBLP:journals/pacmpl/PavlinovicSW21}, and an abstract domain of symbolic accumulator automata, 
we can verify a variety of memory-based, dependent temporal safety properties of higher-order programs.

\subsection{Motivating Examples}

\begin{example}\label{ex:overview1} Consider the following example:

\begin{tabular}{l|l}
\begin{minipage}[m]{1.9in}
\begin{lstlisting}
let rec busy n t = 
  if (n <= 0) then ev (-t)
  else busy (n - 1) t
let main (x:int) (n:int) = 
  ev x; busy n x
\end{lstlisting}
\end{minipage}
&
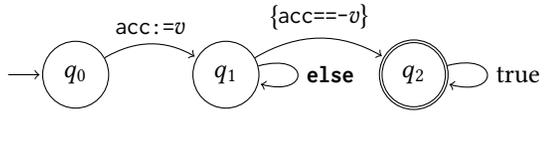
\begin{minipage}[m]{3.5in}
\begin{tikzpicture}[shorten >=1pt, node distance=2.5cm, on grid, auto, initial text=$$]
    \node[state, initial] (q0) {$q_0$};
    \node[state] at (2, 0) (q1) {$q_1$};
    \node[state, accepting] at (4.5, 0) (q2) {$q_2$};

    \path[->] (q0) edge [bend left] node [pos=0.5, sloped, above] {\text{\lstinline{acc:=}}$v$} (q1);
    \path[->] (q1) edge [loop right] node {\text{\lstinline{else}}} (q1);
    \path[->] (q1) edge [bend left] node [pos=0.5, sloped, above] {\{\text{\lstinline{acc}}${}\neq-v$\}} (q2);
    \path[->] (q2) edge [loop right] node {true} (q2);
\end{tikzpicture}
\end{minipage}
\end{tabular}
\end{example}

Above in \ttmain, an integer event \xx\ is emitted, and then a recursive function \lstinline|busy| repeatedly iterates until \nn\ is below 0, at which point the event -\ttt\ (which is equal to -\xx) is emitted. For this program, the possible event traces are simply
\added{$\{x;-x \mid x\in\mathbb{Z}\}$}, i.e., any two-element sequence of an integer and its negation. This property can be expressed by a \emph{symbolic accumulator automaton} (a cousin to symbolic automata and to memory automata, as discussed in Sec.~\ref{sec:acc-autom-eff}), as shown above. \added{The automaton is provided by the user along with the program.} It has an initial control state $q_0$, from which point, whenever an event \texttt{ev($v$)} is observed for any integer $v$, the automaton's internal register \ttacc\ is updated to store value $v$ and a transition is taken to $q_1$. From $q_1$, observing another event whose value is not the negation of the saved \ttacc\ will cause a transition to the final accepting state $q_2$ or otherwise loop at $q_1$. The language of the automaton consists of traces that violate the property of interest. That is, the property expressed by the automaton is the complement of the automaton's language. It consists of the traces:
$\{x(-x)^* \mid x\in\mathbb{Z}\}$, which permits none or arbitrarily many $-x$ events after $x$. \added{We note that a stronger specification that exactly characterizes the set $\{x;-x \mid x\in\mathbb{Z}\}$ can also be expressed (and verified with our approach). We use the weaker specification here to highlight that, in general, the automaton approximates the program's traces.}

{\bf \removed{Na\"{i}ve}\added{Direct} approach: reduction to assertion checking}.
At least in theory, this program/property 
can be verified using existing tools through a cross-product 
transformation between the program and 
property that reduces the problem to an assertion-checking safety problem. As is common, the automaton can be encoded in the programming language
(or the program can be converted to an automaton~\cite{DBLP:conf/cav/HeizmannHP13})  with integer variables \qq\ and \ttacc\ for the automaton's control state and accumulator, respectively. The automaton's transition function is also encoded in the language through simple if-then-else expressions. This is shown in the function \lstinline|ev_step|, which consumes the current automaton configuration, and a next event value \vvv\ and returns the next configuration:
\begin{lstlisting}[aboveskip=0.4\baselineskip, belowskip=0.4\baselineskip]
let ev_step q acc v : (Q * int) =
  (* take one automaton step *)
  if      (q==0) then (1, v)
  else if (q==1 && v==-acc) then (2,acc)
  else if (q==1) then (1,acc)
  else (q,acc)
\end{lstlisting}

A product can then be formed, for example, by passing and returning the (\qq,\ttacc) configuration into and out of every expression, and replacing \lstinline|ev| expressions (which are not meaningful to typical safety verifiers) with a call to \lstinline|ev_step|. For Ex.~\ref{ex:overview1}, this yields the following product program:

\begin{tabular}{ll}
\begin{minipage}[t]{2.5in}
\vspace{-0.75\baselineskip} 
\begin{lstlisting}
let rec busy_prod q acc n t = 
  if (n <= 0) then ev_step q acc (-t)
  else busy_prod q acc (n - 1) t
\end{lstlisting}
\vspace{0.25\baselineskip} 
\end{minipage}

&
\begin{minipage}[t]{3.5in}
\vspace{-0.75\baselineskip} 
\begin{lstlisting}
let main_prod (x:int) (n:int) = 
  let (q,acc) = (0,0) in
  let (q',acc') = ev_step q acc x in
  let (q'',acc'') = busy_prod q acc n x
  in assert(q''==2)
\end{lstlisting}
\vspace{0.25\baselineskip} 
\end{minipage}
\end{tabular}

\noindent
In \lstinline|main_prod| above, the initial configuration is provided for the automaton, then the first event expression is replaced by a call to \lstinline|ev_step|, then the resulting next configuration is passed to \lstinline|busy_prod| and the returned final configuration is input to an \lstinline|assert|. \lstinline|busy_prod| is similar. 

\added{
  We implemented the above translation (details in \apxref{\cref{sec:tuple-translation}}{the extended version~\cite{Nicola2025}}) and used it in combination with a variety of existing verification tools for event-less higher-order programs: (1) the \drift{} tool
which uses a dependent type system and abstract interpretation to verify safety properties of higher-order recursive programs~\cite{DBLP:journals/pacmpl/PavlinovicSW21}, (2) \rcaml{} (part of \coar \cite{github:coar}), another fairly mature tool that can also verify assertions of higher order programs~\cite{hiroshiICFP2024,DBLP:journals/pacmpl/KawamataUST24,DBLP:journals/pacmpl/SekiyamaU23}, (3) \mochi{} \cite{github:mochi}, another software model checker based on higher-order recusion schemes~\cite{DBLP:conf/popl/Kobayashi09,DBLP:conf/pldi/KobayashiSU11}, and (4) \rethfl{}, a type-based validity checker for a fragment of a higher-order fixed-point logic, that leverages CHC solvers to infer predicates within a refinement type system.}

{\bf The problem.}
Although this example tuple product reduction can be verified by these existing tools, unsurprisingly, the approach does not scale well with any of the considered tools. \added{Let us examine another example called \lstinline{auction}, shown in the top left of Fig.~\ref{fig:overview2}, that is only slightly more involved yet already demonstrates the problem for existing tools when the tuple product reduction is used: \drift{} reports a potential assertion violation after \SI{\expBestAuctionTrans}{s}, \rcaml{} times out after \SI{900}{s}, and \mochi{} reports a potential assertion violation after \SI{91}{s}. Only \rethfl{} can verify the example, but still needs \SI{\expBestAuctionRethfl}{s}.
We will describe a technique and tool that can instead verify this example 
in only \SI{\expBestAuctionEVDrift}{s}. In fact, as we will see in our evaluation (\cref{sec:evaluation}), for several more elaborate benchmarks like those inspired by amortized complexity analysis, this techniques is the only one for which verification succeeds.}


\newcommand\fff{\lstinline|f|}
\newcommand\xxx{\lstinline|x|}
\newcommand\pos{\lstinline|pos|}
\newcommand\nneg{\lstinline|neg|}
\newcommand\mfff{\textrm{\lstinline|f|}}
\newcommand\mxxx{\textrm{\lstinline|x|}}
\newcommand\mpos{\textrm{\lstinline|pos|}}
\newcommand\mnneg{\textrm{\lstinline|neg|}}
\newcommand\mtti{\textrm{\lstinline|i|}}
\newcommand\mttj{\textrm{\lstinline|j|}}
\newcommand\mttk{\textrm{\lstinline|k|}}
\newcommand\mttmax{\textrm{\lstinline|iamt|}}
\newcommand\mttf{\textrm{\lstinline|f|}}

\newcommand\circled[1]{\raisebox{.5pt}{\textcircled{\raisebox{-.9pt} {#1}}}}

\begin{figure}
\begin{tabular}{|l|l|}
\hline
{\bf Input Program}: & {\bf Input Property}: Initially, \ttbids$=0$ and \ttrfds$=0$.\\
\begin{minipage}[m]{2.4in}
\begin{lstlisting}
let refund k kamt h _ =
  if k <= 1 then ()
  else ((ev 3)$^\text{\circled{r}}$; h ()) @\label{ln:ev-refund}@
let close j g =
  if j = 1 then ()
  else ((ev 2)$^\text{\circled{c}}$; g ()) @\label{ln:ev-close}@
let rec bid i iamt f = 
  let nmax = iamt + 1 in
  if * then
    ((ev 1)$^\text{\circled{b}}$; @\label{ln:ev-bid}@
    bid (i+1) nmax (refund i iamt f))
  else close i f
let main () = (bid 1 1 ($\lambda$ _.()))$^\text{\circled{f}}$  
\end{lstlisting}
\end{minipage}
&
\begin{minipage}[m]{2.45in}
\begin{tikzpicture}[shorten >=1pt, node distance=1.9cm, on grid, auto, initial text=$$]
    \node[state, initial] (q0) {$q_0$};
    \node[state] at (3.8, 0) (q1) {$q_1$};
    \node[state, accepting] at (1.9, -1.9) (q2) {$q_{err}$};

    \path[->] (q0) edge [loop above] node {\text{\lstinline{\{v=1\} bids:=bids+1}}} (q1);
    \path[->] (q0) edge node [pos=0.5, sloped, above] {\lstinline{\{v=2\}}} (q1);
    \path[->] (q0) edge [bend right] node [pos=0.5, sloped, above] {\text{else}} (q2);
    \path[->] (q1) edge [loop above] node [text width=2.3cm] {\text{\lstinline{\{v=3,}} \text{\lstinline{\ bids > rfds + 1\}}} \text{\lstinline{rfds:=rfds+1}}} (q1);
    \path[->] (q1) edge [bend left] node [pos=0.5, sloped, above] {\text{else}} (q2);
    \path[->] (q2) edge [loop below] (q2);

\end{tikzpicture}
\end{minipage} \\
\hline
\multicolumn{2}{|l|}{
  $\begin{array}{ll}
\textrm{\bf Computed Effect Abstractions}: \\
\text{\underline{Location \circled{b}}}: & \text{\underline{Location \circled{c}}}:\\
q_0^\text{\circled{b}} \mapsto (\textrm{\ttbids} = \textrm{\xxi}) \wedge (\textrm{\xxi} >= 1)
    & q_0^\text{\circled{c}} \mapsto \bot \\ 
q_1^\text{\circled{b}} \mapsto \bot 
    & q_1^\text{\circled{c}} \mapsto \textrm{\ttbids} = (\textrm{\mttj} - 1) \wedge (\textrm{\mttj} >= 2) \wedge (\textrm{\ttrfds} = 0)\\
q_{err}^\text{l\circled{b}} \mapsto \bot
    & q_{err}^\text{l\circled{c}} \mapsto \bot \\
    &\\
\text{\underline{Location \circled{r}}}: & \text{\underline{Location \circled{f}}}:\\
q_0^\text{\circled{r}} \mapsto \bot
    & q_0^\text{\circled{f}} \mapsto (\textrm{\ttbids} = 0) \wedge (\textrm{\ttrfds} = 0)\\ 
q_1^\text{\circled{r}} \mapsto \textrm{\ttbids} = (\textrm{\ttrfds} + \textrm{\mttk} - 1) \wedge (\textrm{\mttk} >= 2) 
    & q_1^\text{\circled{f}} \mapsto \textrm{\ttbids} = \textrm{\ttrfds} + 1 \\
q_{err}^\text{l\circled{r}} \mapsto \bot
    & q_{err}^\text{l\circled{f}} \mapsto \bot \\
    \text{For every other location \circled{i}}: 
    q_{err}^\text{\circled{i}} \mapsto \bot & \\
\end{array}$ 
}\\
\hline
\end{tabular}
\setlength{\belowcaptionskip}{0pt}
\caption{\label{fig:overview2} 
   Top left shows the {\bf auction} example. Top right illustrates the SSA property. At the bottom is the computed effects inferred by our tool.
}
\end{figure}

The \lstinline{auction} example in Fig.~\ref{fig:overview2} 
involves a first stage in the \lstinline{bid} function in which some nondeterministic number of bidders place increasing bids. Each bid event is represented as an \lstinline{ev 1} event (Ln~\ref{ln:ev-bid}). Then, a \lstinline{close} event \lstinline{ev 2} occurs (Ln~\ref{ln:ev-close}), after which point, the \lstinline{k-1} losers are refunded as a \lstinline{refund} event \lstinline{ev 3} (Ln~\ref{ln:ev-refund}).
This recursive program is also {\bf higher-order}: \lstinline{bid} constructs a function \lstinline{(refund i iamt f)} that tracks the amount \lstinline{iamt} to be refunded to bidder \lstinline{i} that was overtaken by the new bid, and \lstinline{f} is a similar function that tracks all previous refunds. When the bidding closes, the last constructed refund function is called to apply all refunds.

The event traces of the program are:
$\{ (1^n; 2; 3^{n-1}) \mid n\in\mathbb{N}\}$,
i.e. any sequence of some $n$ number of ``1''-events, followed by a ``2''-event, followed by $n-1$ occurrences of ``3''-events.
A simple temporal safety property\added{, expressed as an automaton,} that ensures the correct \emph{order} of events could involve three states: an initial state $q_0$ that loops at bid ``1''-events, a transition under close ``2''-events to an accepting state $q_1$, self-loop to $q_1$ under refund ``3''-events, and otherwise transitions to error state $q_{err}$. 
These states and transitions are depicted in the top right of Fig.~\ref{fig:overview2}. 
With an accumulator automaton, this property can be improved to more accurately capture the valid trace histories by counting the bids: we use a tuple accumulator (\lstinline{bids,rfds}) that has a counter for the number of \lstinline{bids} and a counter for the number of refunds \lstinline{rfds}. The self-loop at $q_0$ increments \lstinline{bids}, and the self-loop at $q_1$ ensures that more refunds have not been given than bids, and increments \lstinline{rfds}.

\emph{The struggle.}
A \removed{na\"{i}ve }translation-based reduction to existing safety verification tools for higher-order programs does not fare well and the reason is twofold. First, there is a blowup in the size of the analyzed program due to the translation, which causes a significant increase in analysis time. In addition, tools like \drift{} use abstract domains that are not closed under arbitrary disjunctions. A \removed{na\"{i}ve }translation of the automaton's state space and transition relation into the program will cause loss of precision due to computation of imprecise joins at data-flow join points. This will cause the analysis to infer an effect abstraction that is too imprecise for verifying the desired property. 

\subsection{Effect Abstract Domain}

The key idea of this paper is to 
exploit the structure of the automaton to better capture disjunctive reasoning in the abstract domain. Roughly speaking, the abstract domain will associate each \emph{concrete} automaton control state $q$, with \emph{abstractions} of (i) the event sequences that could lead to $q$ and (ii) the possible program environment at $q$. This abstraction is expressed as a relation between the accumulator value and the program environment.
We will now describe this abstraction and see the resulting computed abstraction depicted in the bottom of Fig.~\ref{fig:overview2}.

We obtain this abstraction in three main steps, provided a given input symbolic accumulator automaton
$A=(\mathcal{Q},\mathcal{V},\delta,\text{\ttacc},\ldots)$ with the alphabet being some set of values $\mathcal{V}$ (in this section let $\mathcal{V}=\mathbb{Z}$)
and transitions updating the control state and accumulator. We now discuss these steps.

\newcommand\hasVal{\!:\!}
\paragraph{Concrete semantics}
To begin, the concrete semantics of the program is simply pairs of event traces $\mathbb{Z}^*$ with program environments, i.e., $\powerset(\mathbb{Z}^* \times Env)$. Transitions in the concrete semantics naturally update the environment in accordance with the reduction rules, and the event sequence is only updated when an expression \lstinline|ev| $v$ is reduced: $\powerset(\mathbb{Z}^* \times Env) \xrightarrow{\textrm{\lstinline|ev|} \; v} \powerset(\mathbb{Z}^* \times Env)$. For the above example, a concrete sequence of states and transitions could be the following:
\[\begin{array}{l}
(\epsilon, [main, (\textrm{empty env})])
\leadsto
(\epsilon, [bid,\textrm{\lstinline|i|}\hasVal{}1, \textrm{\lstinline|iamt|}\hasVal{}1, \textrm{\lstinline|f|}\hasVal{}(\lambda\ \_\ldots)])\\
\xrightarrow{\textrm{\lstinline|ev|} \; 1}
(\{1\}, [bid,\textrm{\lstinline|i|}\hasVal{}1, \textrm{\lstinline|iamt|}\hasVal{}1, \textrm{\lstinline|f|}\hasVal{}(\lambda\ \_\ldots)])
\stackrel{(\textrm{\lstinline|ev|} \; 1)^{41}}{\leadsto}
(\{1^{42}\}, [close,\textrm{\lstinline|j|}\hasVal{}43, \textrm{\lstinline|g|}\hasVal{}(\lambda\ \_.(\lambda\ \_\ldots))])\\
\xrightarrow{\textrm{\lstinline|ev|} \; 2}
(\{1^{42},2\}, [\mathit{refund},\textrm{\lstinline|k|}\hasVal{}42, \textrm{\lstinline|kamt|}\hasVal{}42, \textrm{\lstinline|h|}\hasVal{}(\lambda\ \_.(\lambda\ \_\ldots))]\\
\xrightarrow{\textrm{\lstinline|ev|} \; 3}
(\{1^{42},2,3\}, [\mathit{refund},\textrm{\lstinline|k|}\hasVal{}42, \textrm{\lstinline|kamt|}\hasVal{}42, \textrm{\lstinline|h|}\hasVal{}(\lambda\ \_.(\lambda\ \_\ldots))])\\
\stackrel{(\textrm{\lstinline|ev|} \; 3)^{40}}{\leadsto}
(\{1^{42},2,3^{41}\}, [\mathit{refund},\textrm{\lstinline|k|}\hasVal{}2, \textrm{\lstinline|kamt|}\hasVal{}2, \textrm{\lstinline|h|}\hasVal{}(\lambda\ \_\ldots))])
\end{array}
\]
(Technically a transition takes the powerset of possible sequence/environment pairs to another powerset; here we show only one sequence for simplicity.)
Above the first component is an event sequence, starting with the empty sequence $\epsilon$ and, for this nondeterministic behavior, the trace will accumulate the event sequence $1^{42}; 2; 3^{41}$.

\paragraph{Intermediate abstraction via concrete automaton control states}
With integer variables and integer effect sequences, it is clear that abstraction is 
needed to represent the possible event sequences of a program even as simple as this running example.
In this example, there are infinitely many sequences of the form $1^{k}; 2; 3^{k-1}$.
The first key idea we explore in this paper is to organize the abstraction around the automaton and, crucially, \emph{keep the automaton control state concrete} while abstracting everything else: the environment, the possible event sequence prefixes, and the value of the automaton's accumulator.
The benefit is that this will lead to a somewhat disjunctive abstract effect domain, where event trace prefixes can be categorized according to the control state (and accumulator values and program environments) that those prefixes reach.
To this end, the first layer of abstraction uses the automaton control states $\mathcal{Q}$ (rather than merely event sequences), and associates each automaton control state with the possible set of pairs of accumulator value $\mathbb{Z}$ and program environment that reach that state along some event sequence:
 $\mathcal{Q} \mapsto \powerset(\mathbb{Z} \times Env)$. At this layer, 
 transitions from an  expression \lstinline|ev| $v$ are captured through the automaton's transition function $\delta(v)$, which leads to a (possibly) new automaton state and updates the accumulator value:
$\mathcal{Q} \mapsto \powerset(\mathbb{Z} \times Env)
\xrightarrow{\delta(v)} 
\mathcal{Q} \mapsto \powerset(\mathbb{Z} \times Env) 
$. 
For the \lstinline{auction} example, when an execution iterates \lstinline|bid| 42 times, there is an event trace prefix $1^{42}$, then the following lists some of the effects at body of \lstinline{bid} per each $q$:
{
\setlength{\abovedisplayskip}{6pt}
\setlength{\belowdisplayskip}{6pt}
\[\begin{array}{l}
q_0 \mapsto \{ (\{1\}, (\mtti\hasVal{}1, \mttmax\hasVal{}1, \mttf\hasVal{}\_)), 
    (\{1,1\}, (\mtti\hasVal{}2, \mttmax\hasVal{}2, \mttf\hasVal{}\_)), \ldots \}, \;\;\;
q_1 \mapsto \emptyset, \;\;\;
q_{err} \mapsto \emptyset.
\end{array}\]
}
Above $q_1$ is not reachable yet because at the point when the program reaches location {\circled{c}}, at least one close (``2'') event must have been emitted. Similarly $q_{err}$ is not yet reachable.
$q_0$ is, however, reachable with event sequences of the form $1^k$ and in the corresponding environment $\mtti$ will be equal to $k$.


\paragraph{Abstract relations with the accumulator.}
Thus far we associate event sequence and environment pairs per control state, but 
there are still infinite sets of pairs.
We thus next abstract \emph{relations} between the accumulator values at location $q$ and the environments, employing a parametric abstract domain of base refinement types. That is, the type system provides abstractions of program values, which we can then also relate to abstractions of the accumulator. 
We will discuss the formal details of this abstraction in Sec.~\ref{sec:acc-autom-eff} but illustrate the abstraction in the bottom of Fig.~\ref{fig:overview2}.
For every location {\circled{i}} and automaton state $q_j$, we compute a summary of the possible trace prefixes and corresponding abstraction of the program variables, accumulator, and relations between them. 
In this example, at the \lstinline{ev 1} location denoted \circled{b}, our summary for
$q_0^\text{\circled{b}}$ reflects that the number of bid (1) events in the prefix counted by accumulator \lstinline{bid} is equal to the environment variable $\mtti$, and that $\mtti$ is positive. No other automaton states are reachable.
Meanwhile, at the \lstinline{ev 3} location denoted \circled{r}, our summary for
$q_1^\text{\circled{r}}$ reflects that the number of refund (3) events seen in the prefix so far is $\mttk-1$ away from the number of bid (1) events, and that $\mttk \geq 2$.
The automaton specifies if ever this is violated it will transition to $q_{\mathit{err}}$. The program is safe because at every location \circled{i}, we compute $q_{\mathit{err}}^\text{\circled{i}} \mapsto \bot$.  Our accumulator automata can also include assertions that can be applied at the end of a trace. In this example, we would like to prove that the number of refunds was one less than the number of bids. We also compute abstractions at the final program location denoted \circled{f}, including the fact that
\lstinline{bids=rfds+1} (or \lstinline{bids=rfds=0}), which validates the end-of-program assertion.

\newcommand\monoacc[1]{\phi^{mono}_{acc}(\textrm{\lstinline{#1}})}
\newcommand\auctacc[1]{\phi^{auct}_{acc}(\textrm{\lstinline{#1}})}

\subsection{Type System, Inference, Evaluation}

Our approach to verifying effects is fully automated. Toward achieving this, the rest of this paper addresses the challenges identified in Sec.~\ref{sec:intro}, but here with more detail in the context of this example:

\emph{Accumulative type and effect system} (Sec.~\ref{sec:types-and-effects}). In order to form relations between reachable automaton configurations' accumulator and program variables, we present a novel dependent type and effect system that is \emph{accumulative} in nature. 
 The type system allows us to, for example, express judgments on the $(\evkw ~ \textrm{\lstinline{3}})^\textrm{\circled{r}}$ expression to ensure that the count of bid ($\textrm{\lstinline{2}}$) events is at least one more than the count of refund ($\textrm{\lstinline{3}}$) events. 

First, let $\auctacc{k}$ be shorthand for $\textrm{\lstinline|bids|} = \textrm{\lstinline|rfds| + \lstinline|k| - \lstinline{1}} \wedge \textrm{\lstinline|k|} >= 2$
, i.e., that the accumulator count of bids is equal to the accumulator count of refunds plus program variable $\textrm{\lstinline{k}}$ minus one, and that the value of $\textrm{\lstinline{k}}$ is at least 2. 
Further, due to the nested construction of delayed refund calls with decreasing arguments, when we reach $(\evkw ~ \textrm{\lstinline{3}})$, we have that $k \leadsto (k - 1)$. 
We thus obtain the judgment below.
We focus on the boxed area, in which we compute the abstract effect concatenation operation denoted $\odot$. This concatenation is between $q_1$'s existing effect in the context along with path condition $\textrm{\lstinline{k>1}}$, and the $g$uard/$u$pdate for a next single refund event (3).
\[
\begin{array}{l}
\Gamma ; [ q_0 \mapsto \bot, q_1 \mapsto \auctacc{k+1}, q_{err} \mapsto \bot, \ldots ]\\
\vdash
\termkw{ev}\; 3 :\tandeff{()}{
\left[
  \begin{array}{l}
    q_0 \mapsto \bot, \\
    q_1 \mapsto 
    \boxed{[\auctacc{k+1} \wedge (\textrm{\lstinline{k > 1}})] \odot [g: \textrm{\lstinline{bids}} > \textrm{\lstinline{rfds + 1}}]; [u: \textrm{\lstinline{rfds}} := \textrm{\lstinline{rfds + 1}}]},\\
   q_{err} \mapsto \bot
   \end{array} \right]
} \end{array}
\]
The context information $\auctacc{k+1}$ is strengthened by the constraints on the program variable \lstinline{k} imposed by the branching condition and this is sufficient to ensure the validity of the transition guard $\textrm{\lstinline{bids}} > \textrm{\lstinline{rfds + 1}}$.
The update of the accumulator $\textrm{\lstinline{rfds}}$ to $\textrm{\lstinline{rfds+1}}$ reestablishes $\auctacc{k}$ at $q_1$.
Moreover, the result of the concatenation guarantees that $q_{err}$ remains unreachable.


%
%
%


\emph{Effect abstract domain} (Sec.~\ref{sec:acc-autom-eff}). We formalize the effect abstract domain discussed above.

\emph{Automated inference of effects} (Sec.~\ref{sec:inference}). 
We introduce a dataflow abstract interpretation inference of types that 
calculates summaries of effects, organized around concrete automaton control states, as seen in the example in Fig.~\ref{fig:overview2}.  To achieve this, we exploit the parametricity of type systems (like~\cite{DBLP:journals/pacmpl/PavlinovicSW21}) over the kinds of constructs in the language, introducing \emph{sequences} as a new base type. We then embed sequences into the $q$-indexed effect components. 

\emph{Verification, Implementation \& Benchmarks} (Sec.~\ref{sec:implementation}).
To verify examples like \lstinline|auction| (and others among the \expBenchCount{} benchmarks), we have implemented our 
(i) abstract effect domain, (ii) accumulative type and effect system and (iii) automated inference in a new tool called \evdrift{}. 
\evdrift{} takes, as input, the program in an OCaml-like language (Fig.~\ref{fig:overview2}) as well as a symbolic accumulation automaton, written in a simple specification language (control states and the accumulator are integers and the automaton transition function is given by \evdrift{} expressions).
In Sec.~\ref{sec:implementation} we discuss how 
our inference is used for verification,
and implement the \removed{na\"{i}ve }product reductions to compare against tools for effect-free programs.

\emph{Evaluation} (Sec.~\ref{sec:evaluation}). 
\evdrift{} verifies \lstinline|auction| in \expBestAuctionEVDrift{}s,
whereas previous assertion-verifiers (combined with our translations) \rvremoved{either timeout (\drift{}) or }\rvadded{either took significantly longer to verify it (\expBestAuctionRethfl{}s for \rethfl), timeout (\rcaml{}) or fail to verify (\drift{} and \mochi{})}. More generally, \evdrift{} verifies more examples and otherwise outperforms \drift{}, \rcaml{}, \mochi{}, and \rethfl{} by \expSpeedupEVoverDrift{}$\times$,
\expSpeedupEVoverRcaml{}$\times$,
\expSpeedupEVoverRealMochi{}$\times$, 
\expSpeedupEVoverRethfl{}$\times$ resp.
on benchmarks that each solve.

\section{Preliminaries}
\label{sec:preliminaries}

We briefly summarize background definitions and notation.
The formal development of our approach uses an idealized language based on a lambda calculus with terms
$\term \in \cterms ::= {} \const ~|~ x 
~|~ \termite{\term}{\term}{\term} \allowbreak
~|~ \lambda x.\,\term 
~|~ (\term ~ \term)
~|~ \evkw ~ \term$ and 
values 
$
\cval \in \cvalues ::= {} \const ~|~ \lambda x.\,\term
$.
Expressions $\term$ in the language consist of constant values $\const \in \consts$ (e.g. integers and Booleans), variables $x \in
\vars$, function applications, lambda abstractions, conditionals, and event expressions $\evkw~\term_1$. We assume the existence of a dedicated unit value $\unitval \in \consts$ and the Boolean constants $\termtrue,\termfalse \in \consts$. Values $\cval \in \cvalues$ consist of constants and lambda abstractions.
We will often treat expressions as equal modulo alpha-renaming
and write $\term[\term'/x]$ for the term obtained by substituting all free occurrences of $x$ in $\term$ with term $\term'$ while avoiding variable capturing. We further write $\fv(\term)$ for the set of free variables occurring in $\term$.

A \emph{value environment} $\cenv$ is a total map from variables to values: $\cenv \in \cenvs \Def= \vars \to \cvalues$.

The operational semantics of the language is defined with respect to a transition relation over configurations $\pair{\term}{\ceff} \in \cterms \times \cvalues^*$ where $\term$ is a closed expression representing the continuation and $\ceff$ is a sequence of values representing the events that have been emitted so far. All configurations are considered initial and configurations $\pair{\cval}{\ceff}$ are terminal. To simplify the reduction rules, we use evaluation contexts $E$ that specify evaluation order:
$
  E ::= [] ~|~ E \; e \;|\; v \; E \;|\; \termite{E}{\term_1}{\term_2} \;|\; \evkw E
$.
The transition relation $\step{\term,\ceff}{\term',\ceff'}$ is then defined in \cref{fig:operational-semantics}. In particular, the rule \ruleref{e-ev} captures the semantics of event expressions: the evaluation of $\evkw \cval$ returns the unit value and its effect is to append the value $\cval$ to the event sequence $\ceff$. We write $\pair{\term}{\ceff} \leadsto \pair{\term'}{\ceff'}$ to mean that $\pair{\term}{\ceff} \rightarrow^* \pair{\term'}{\ceff'}$ and there exists no $\pair{\term''}{\ceff''}$ such that $\step{\term',\ceff'}{\term'',\ceff''}$.

{
 \setlength{\abovedisplayskip}{6pt}
 \setlength{\belowdisplayskip}{6pt}
\begin{figure}[t]
\vspace{-0.8\baselineskip}

\begin{mathpar}
  \axiomHtop{e-app}
  {\step{(\lambda x.\term) \;v, \pi}{\term[v/x], \pi}}
  \and
  \axiomHtop{e-ev}
  {\step{\evkw \cval, \pi}{\unitval, \pi \cdot \cval}}
  \and
  \inferH{e-context}
  {\step{\term, \pi}{\term', \pi'}}
  {\step{E[\term],\pi}{E[\term'], \pi'}}
  \and
  \axiomHtop{e-ite-true}
  {\step{\termite{\termtrue}{\term_1}{\term_2},\pi}{\term_1, \pi}}
  \and
  \axiomHtop{e-ite-false}
  {\step{\termite{\termfalse}{\term_1}{\term_2},\pi}{\term_2, \pi}}
\end{mathpar}
\caption{Reduction rules of operational semantics.\label{fig:operational-semantics}}
\end{figure}
}

\paragraph{(Non-accumulative) type and effect systems}

Conventional type and effect systems~\cite{DBLP:conf/popl/LucassenG88} typically take the form $\Gamma \vdash e : \tau\&\eff$ and capture the local effects that occur during the evaluation of expression $e$ to value $v$. 
Such systems have also been extended to the setting of higher-order programs~\cite{DBLP:conf/lics/Nanjo0KT18,DBLP:journals/jfp/SkalkaSH08,DBLP:conf/aplas/SkalkaS04}.
While these systems are generally suitable to deductive reasoning,
the judgements assume no information describing the program's behavior up to the evaluation of the respective expression. They thus 
fail to provide contextual reasoning for effects and so they suffer from a  lack of precision and increase the difficulty of automation.



\section{Accumulative Type and Effect System}
\label{sec:types-and-effects}

In this section, we present an abstract formalization of our dependent type and effect system for checking accumulative effect safety properties. The notion is parameterized by the notion of basic refinement types, which abstract sets of constant values, and the notion of dependent effects, which abstract sets of event sequences. Both abstractions take into account the environmental dependencies of values and events according to the context where they occur in the program. To facilitate the static inference of dependent types and effects, we formalize these parameters in terms of abstract domains in the style of abstract interpretation.

\paragraph{Base refinement types} We assume a lattice of base refinement types
$\langle \bdom, \bord, \bbot, \btop, \bjoin, \bmeet \rangle$.
Intuitively, a basic refinement type $\bval \in \bdom$ represents a set of pairs $\pair{c}{\cenv}$ where $c \in \consts$ and $\cenv \in \cenvs$ is a value environment capturing $c$'s environmental dependencies. To formalize this intuition, we assume a \emph{concretization function} $\bgamma \in \bdom \to \powerset(\cvalues \times \cenvs)$. We require that $\bgamma$ is monotone and top-strict (i.e., $\bgamma(\btop) = \cvalues \times \cenvs$). We assume the existence of a basic refinement type $\mathit{bool}$ such that $\bgamma(\mathit{bool}) = \{\termtrue,\termfalse\} \times \cenvs$.

We let $\dom(\bval)$ denote the set of variables $x \in \vars$ that are constrained by $\bval$. Formally:
\[\dom(\bval) = \pset{x \in \vars}{\exists v,\cenv,\cenv'. \pair{v}{\cenv} \in \bgamma(\bval) \notni \pair{v}{\cenv'} \land \cenv(x) \neq \cenv'(x) \land \cenv[x \mapsto \cenv'(x)] = \cenv'(x)}\enspace.\]
Examples of possible choices for $\bdom$ include \added{base types of the shape $\bval = \{\nu:t \mid \varphi\}$ where $t$ is a simple type like $\mathsf{int}$ and $\varphi$ a value in a standard} relational abstract domain such as octagons and polyhedra \added{that relates $\nu$ with the variables in the environment}. For instance, when considering the polyhedra domain, basic refinement types can represent values subject to a system of linear constraints, such as the following, where $x, y, z$ are the variables evaluated in the environments: 
\begin{align*}
    \bval = \{\nu: \mathsf{int} \mid x + y + z \leq \nu \wedge x-y \leq 0 \wedge y + z \leq 2x\}\enspace.
\end{align*}
\added{Note that the set notation in the example is just syntactic sugar. The value $\bval$ is not actually a set, but an element of $\bdom$ that denotes a set (of pairs) under $\bgamma$.}


\paragraph{Dependent effects}
Let $\langle \effdom, \efford, \effjoin, \effmeet, \effbot, \efftop \rangle $ denote a lattice of dependent effects. Similar to basic refinement types, a dependent effect $\eff \in \effdom$ represents a set of pairs $\pair{\ceff}{\cenv}$ where $\ceff$ is a trace and $\cenv$ captures its environmental dependencies. Again, we formalize this by assuming a monotone and top-strict function $\effgamma \in \effdom \to \powerset(\ceffs \times \cenvs)$. Similar to basic types, we denote by $\dom(\eff)$ the set of variables that are constrained by $\eff$.
We assume some additional operations on our abstract domains for dependent types and effects that we will introduce below.

\paragraph{Types}


With basic refinement types and dependent effects in place, we define our types as follows:
\begin{align*}
  \tval \in \tvalues ~::=~ {} & \bval ~\mid~ \tfun{x}{\tval_2}{\eff_2}{\tval_1}{\eff_1} ~\mid~ \texists{x}{\tval_1}{\tval_2} \enspace.
\end{align*}
Intuitively, a function type $\tfun{x}{\tval_2}{\eff_2}{\tval_1}{\eff_1}$ describes functions that take an input value $x$ of type $\tval_2$ and a prefix trace described by $\eff_2$ such that evaluating the body $\term$ produces a result value of type $\tval_1$ and extends the prefix trace to a trace described by $\eff_1$. Type refinements in $\tval_1$ may depend on $x$. Existential types $\texists{x}{\tval_1}{\tval_2}$ represent values of type $\tval_2$ that depend on the existence of a witness value $x$ of type $\tval_1$.

We lift the function $\dom$ from basic types and effects to types in the expected way:
\begin{align*}
\dom(\tfun{x}{\tval_2}{\eff_2}{\tval_1}{\eff_1}) = {} & \dom(\tval_2) \cup ((\dom(\eff_2) \cup \dom(\tval_1) \cup \dom(\eff_1)) \setminus \set{x})\\
\dom(\texists{x}{\tval_1}{\tval_2}) = {} & \dom(\tval_1) \cup (\dom(\tval_2) \setminus \set{x})
\end{align*}

We also lift $\bgamma$ to a concretization function $\tgamma \in \tvalues \to \powerset(\cvalues \times \cenvs)$ on types:
\begin{align*}
& \tgamma(\bval) = \bgamma(\bval)\\
& \tgamma(\tfun{x}{\tval_1}{\eff_1}{\tval_2}{\eff_2}) = \cvalues \times \cenvs \\
& \tgamma(\texists{x}{\tval_1}{\tval_2}) = \pset{\pair{\cval}{\cenv}}{\pair{\cval'}{\cenv} \in \tgamma(\tval_1) \land \pair{\cval}{\cenv[x \mapsto \cval']} \in \tgamma(\tval_2)}\enspace.
\end{align*}
Note that the function $\tgamma$ uses a coarse approximation of function values. The reason is that we will use $\tgamma$ to give meaning to typing environments, which we will in turn use to define what it means to strengthen a type with respect to dependencies expressed by a given typing environment. When strengthening with respect to a typing environment, we will only track dependencies to values of base types, but not function types. 

We define typing environments $\tenv$ as binding lists between variables and types:
$\tenv ::= {} \; \varnothing ~\mid~ \tenv,x:\tval$.
We extend $\dom$ to typing environments as:
$\dom(\varnothing)=\emptyset$ and $\dom(\tenv,x:\tval) = \dom(\tenv) \cup \set{x}$.
We then impose a well-formedness condition $\wf(\tenv)$ on typing environments. Intuitively, the condition states that bindings in $\tenv$ do not constrain variables that are outside of the scope of the preceding bindings in $\tenv$:
\begin{mathpar}
\axiomHtop{wf-emp}
{\wf(\varnothing)}
\and
\inferH{wf-bind}
{\wf(\tenv) \\
  \dom(\tval) \subseteq \dom(\tenv) \\
  x \notin \dom(\tenv)}
{\wf(\tenv,x:\tval)}
\end{mathpar}
If $\wf(\tenv)$ and $x \in \dom(\tenv)$, then we write $\tenv(x)$ for the unique type bound to $x$ in $\tenv$.

As previously mentioned, we lift $\tgamma$ to a concretization function for typing environments:
\[\tgamma(\varnothing) = \cenvs \qquad \tgamma(\tenv, x: \tval) = \tgamma(\tenv) \cap \pset{\cenv}{\exists \cval.\, \pair{\cval}{\cenv} \in \tgamma(\tenv(x))}\enspace.\]

\paragraph{Typing judgements}

Our type system builds on existing refinement type systems with semantic subtyping~\cite{DBLP:conf/plpv/KnowlesF09,DBLP:journals/pacmpl/BorkowskiVJ24}. Subtyping judgements take the form $\subtjudge{\tenv}{\tval_1}{\tval_2}$ and are defined by the rules in \cref{fig:subtyping}. We implicitly restrict these judgments to well-formed typing environments.

The rule \ruleref{s-base} handles subtyping on basic types by reducing it to the ordering $\bord$. Importantly, the basic type $\bval_1$ on the left side is strengthened with the environmental dependencies expressed by $\tenv$. To this end, we assume the existence of an operator $\tstrengthen{\bval}{\tenv}$ that satisfies the following specification:
\[\bgamma(\tstrengthen{\bval}{\tenv}) \;\supseteq\; \bgamma(\bval) \cap (\cvalues \times \tgamma(\tenv))\enspace.\]
We require this operator to be monotone in both arguments where $\tenv \leq \tenv'$ iff for all $x \in \dom(\tenv')$, $\tenv(x) = \tenv'(x)$.
We also assume a strengthening operator $\tstrengthen{\eff}{\tenv}$ on effects with corresponding assumptions.

The rule \ruleref{s-fun} handles subtyping of function types. As expected, the input type and effect are ordered contravariantly and the output type and effect covariantly. Note that we allow the input effect to depend on the parameter $x$.

The rule \ruleref{s-exists} introduces existential types on the left side of the subtyping relation whereas \ruleref{s-wit} introduces them on the right side. The latter rule relies on an operator $\tval[y/x]$ that expresses substitution of the dependent variable $x$ in type $\tval$ by the variable $y$. This operator is defined by lifting corresponding substitution operators $\bval[y/x]$ on basic types and $\eff[y/x]$ on effects in the expected way. The soundness of these operators is captured by the following assumption:
\begin{align*}
  \bgamma(\bval[y/x]) \supseteq {} & \pset{\pair{\cval}{\cenv[x \mapsto \rho(y)]}}{\pair{\cval}{\cenv} \in \bgamma(\bval)}\\
  \effgamma(\eff[y/x]) \supseteq {} & \pset{\pair{\ceff}{\cenv[x \mapsto \rho(y)]}}{\pair{\ceff}{\cenv} \in \effgamma(\eff)}
\enspace.
\end{align*}

\begin{figure}
\begin{mathpar}
\inferHtop{s-base}
{\tstrengthen{\btau_1}{\tenv} \bord \btau_2}
{\subtjudge{\tenv}{\bval_1}{\bval_2}}
\and
\inferHtop{s-wit}
{\subtjudge{\tenv}{\tval'}{\tval} \\
  \subtjudge{\tenv,y:\tval'}{\tval_1}{\tval_2[y/x]}}
{\subtjudge{\tenv,y:\tval'}{\tval_1}{\texists{x}{\tval}{\tval_2}}}
\and
\inferHtop{s-exists}
{ \subtjudge{\tenv,x:\tval}{\tval_1}{\tval_2}}
{\subtjudge{\tenv}{\texists{x}{\tval}{\tval_1}}{\tval_2}}
\and
\inferH{s-fun}	
{\subtjudge{\tenv}{\tval_{2}'}{\tval_{2}} \\
  \tstrengthen{\eff_{2}'}{\tenv,x: \tval_2'} \efford \eff_{2} \\
  \subtjudge{\tenv,x: \tval_{2}'}{\tval_{1}}{\tval_{1}'} \\
  \tstrengthen{\eff_{1}}{\tenv,x: \tval_2'} \efford \eff_{1}' }
{ \subtjudge{\tenv}{(\tfun{x}{\tval_{2}}{\eff_{2}}{\tval_{1}}{\eff_{1}})}{(\tfun{x}{\tval_{2}'}{\eff_{2}'}{\tval_{1}'}{\eff_{1}'})}}
\end{mathpar}
\caption{Semantic subtype relation.\label{fig:subtyping}}
\end{figure}

Typing judgments take the form $\tjudge{\tenv ; \eff}{\term}{\tval}{\eff'}$ and are defined by the rules in \cref{fig:typing}\footnote{\added{All auxiliary operators occurring in the typing rules such as $\odot$ are defined below.}}. Intuitively, such a judgement states that under typing environment $\tenv$, expression $\term$ extends the event sequences described by effect $\eff$ to the event sequences described by effect $\eff'$ and upon termination, produces a value described by type $\tval$. Again, the typing environments occurring in typing judgements are implicitly restricted to be well-formed. Moreover, we implicitly require $\dom(\eff) \subseteq \dom(\tenv)$.

\begin{figure}
\begin{mathpar}
\vspace{-0.5em}      
\axiomHtop{t-const}
{\tjudge{\tenv;\eff}{c}{\{\nu = c\}}{\eff}}
\and 
\axiomHtop{t-var}
{\tjudge{\tenv;\eff}{x}{\tenv(x)}{\eff}}
\and 
\inferH{t-ev}
{\Gamma\;;\;\eff \vdash \term:\tandeff{\tau}{\eff'}}
{\tjudge{\tenv;\eff}{\termkw{ev}~\term}{\{\nu = \unitval\}}{(\eff' \odot \tau)}}
\and
\vspace{-.5em}      
\inferH{t-abs}
{\Gamma, x:\tau_2\;;\;\eff_2 \vdash e: \tandeff{\tau_1}{\eff_1}}
{\Gamma\;;\;\eff \vdash \lambda x.\term:\tandeff{(x:\tandeff{\tau_2}{\eff_2} \to \tandeff{\tau_1}{\eff_1})}{\eff}}
\and
\vspace{-.5em}      
\inferH{t-app}
{\tjudge{\tenv ; \eff}{\term_1}{\tval_1}{\eff_1} \\
  \tjudge{\tenv ; \eff_1}{\term_2}{\tval_2}{\eff_2} \\
  \tval_1 = \tfun{\absn}{\tval_2}{\eff_2}{\tval}{\eff'}}
{\tjudgesimp{\tenv ; \eff}{\term_1 \; \term_2}{\exists x: \tau_2.\,(\tandeff{\tval}{\eff'})}}
\and 
\vspace{-.5em}      
\inferH{t-weaken}
{\tstrengthen{\eff}{\tenv} \efford \effp \\
  \tjudge{\tenv ; \effp}{e}{\tval'}{\effp'}\\
  \subtjudge{\tenv}{\tval'}{\tval}\\
\tstrengthen{\effp'}{\tenv} \efford \eff'}
{\tjudge{\tenv ; \eff}{e}{\tval}{\eff'}}
\and 
\vspace{-.5em}      
\inferH{t-cut}
{ \tjudge{\tenv;\eff}{\cval}{\tval}{\eff} \\
  x \notin \fv(\term)\\
  \tjudge{\tenv,x:\tval ; \eff}{\term}{\tval'}{\eff'}}
{\tjudgesimp{\tenv ; \eff}{\term}{\texists{x}{\tval}{(\tandeff{\tval'}{\eff'})}}}
\and
\vspace{-.5em}
\inferH{t-ite}
{\tjudge{\tenv ; \eff}{x}{\mathit{bool}}{\eff_0}\\
  \tjudge{\tenv[x = \mathsf{true}]; \eff_0}{\term_1}{\tval}{\eff'} \\
\tjudge{\tenv[x = \mathsf{false}]; \eff_0}{\term_2}{\tval}{\eff'}}
{\tjudge{\tenv;\eff}{\termkw{if}~x~\termkw{then}~\term_1~\termkw{else}~\term_2}{\tval}{\eff'}}
\end{mathpar}
\caption{Typing relation.\label{fig:typing}}
\end{figure}

The rule \ruleref{t-const} is used to type primitive values. For this, we assume an operator that maps a primitive value $c$ to a basic type $\{\nu = c\} \in \bdom$ such that $\bgamma(\{\nu = c\}) \supseteq \{c\} \times \cenvs$.

The rule \ruleref{t-ev} is used to type event expressions $\evkw \term$. For this, we assume an \emph{effect extension operator} $\eff \odot \tval$ that abstracts the extension of the traces represented by effect $\eff$ with the values represented by the type $\tval$, synchronized on the value environment:
\begin{equation}\effgamma(\eff \odot \tval) \;\supseteq\; \{\pair{\ceff \cdot \cval}{\cenv} \mid \tuple{\cval,\cenv} \in \tgamma(\tval) \land \tuple{\ceff,\cenv} \in \effgamma(\eff)\}\enspace.
\label{eq:effect-extension-soundness}
\end{equation}
We require that $\odot$ is monotone in both of its arguments.

%

The following is an example judgment for the bid event $\evkw~1$ expression in the \lstinline{auction} example from Sec.~\ref{sec:overview}:
\[
\tenv, i : \{\nu ~|~ \nu >= 1\}; [..q_1 \mapsto \mathit{bids} = i - 1 >= 0] \vdash \evkw~1 : \mathsf{unit} \& [..q_1 \mapsto \mathit{bids} = i >= 1]\
\]
The effect to the left of the turnstile describes event prefixes associated to all the executions leading to the evaluation of expression $\evkw~1$. It states that $q_1$ is the only reachable state and the accumulator $\textrm{\lstinline{bid}}$ is equal to $\textrm{\lstinline{i-1}}$.
The typing judgment states that, for all executions, the extended effect that account for a new bidding represented by the observable bid $(\textrm{\lstinline{1}})$ event, preserves the invariant between the accumulator and the program variable $i$, and that $i >= 1$ according to its type constraints. 

The rule \ruleref{t-ite} assumes without loss of generality that only variables are allowed to be used as test conditions. It is defined in terms of an environment strengthening operator $\tenv[x = c]$ for $x \in \dom(\tenv)$ defined as $\tenv[x = c](x) = \tenv(x) \bmeet \{\nu = c\}$ and $\tenv[x = c](y) = \tenv(y)$ for $y \in \dom(\tenv) \setminus \{x\}$.

The notation $\texists{x}{\tval}{(\tandeff{\tval'}{\eff})}$ used in the conclusion of rules \ruleref{t-app} and \ruleref{t-cut} is a shorthand for $\tandeff{(\texists{x}{\tval}{\tval'})}{(\texists{x}{\tval}{\eff})}$, where $\texists{x}{\tval}{\eff}$ computes the projection of the dependent variable $x$ in effect $\eff$, subject to the constraints captured by type $\tval$. That is, this operator must satisfy:
\[\effgamma(\texists{x}{\tval}{\eff}) \supseteq \pset{\pair{\ceff}{\cenv[x \mapsto \cval]}}{\pair{\ceff}{\cenv} \in \effgamma(\tstrengthen{\eff}{x:\tval})}\enspace.\]
As with our other abstract domain operators, we require this to be monotone in both $\tval$ and $\eff$.

The rule \ruleref{t-cut} allows one to introduce an existential type $\texists{x}{\tval}{\tval'}$, provided one can show the existence of a witness value $\cval$ of type $\tval'$ for $x$. In other dependent refinement type systems, this rule is replaced by a variant of rule \ruleref{s-wit} as part of the rules defining the subtyping relation. We use the alternative formulation to avoid mutual recursion between the subtyping and typing rules.

The remaining rules are as expected. In particular, the rule \ruleref{t-weaken} allows one to weaken a typing judgement using the subtyping relation (and ordering on effects), relative to the given typing environment.

\paragraph{Soundness} We prove the following soundness theorem. Intuitively, the theorem states that (1) well-typed programs do not get stuck and (2) the output effect established by the typing judgement approximates the set of event traces that the program's evaluation may generate.

\begin{theorem}[Soundness]
  \label{thm:soundness}
  If $\tjudge{\eff}{\term}{\tval}{\eff'}$ and $\pair{\ceff}{\cenv} \in \effgamma(\eff)$, then $\pair{\term}{\ceff} \leadsto \pair{\term'}{\ceff'}$ implies $\term' \in \cvalues$ and $\pair{\ceff'}{\cenv} \in \effgamma(\eff')$.
\end{theorem}

The soundness proof details are available in \apxref{Apx.~\ref{sec:soundness}}{the extended version~\cite{Nicola2025}}, but we summarize here.
The proof of \cref{thm:soundness} proceeds in two steps. We first show that any derivation of a typing judgement $\tjudge{\eff}{\term}{\tval}{\eff'}$ can be replayed in a concretized version of the type system where basic types are drawn from the concrete domain $\powerset(\cvalues \times \cenvs)$ and effects from the concrete domain $\powerset(\ceffs \times \cenvs)$ (i.e., both $\bgamma$ and $\effgamma$ are the identity on their respective domain). Importantly, in this concretized type system all operations such as strengthening $\tstrengthen{\tval}{\tenv}$ and effect extension $\eff \odot \tval$ are defined to be precise. That is, we have e.g.
$\eff \odot \tval \Def= \{\pair{\ceff \cdot \cval}{\cenv} \mid \tuple{\cval,\cenv} \in \tval \land \tuple{\ceff,\cenv} \in \eff\}\enspace$.
In a second step, we then show standard progress and preservation properties for the concretized type system.

While one could prove progress and preservation directly for the abstract type system, this would require stronger assumptions on the abstract domain operations. By first lowering the abstract typing derivations to the concrete level, the rather weak assumptions above suffice.

\section{Automata-Based Dependent Effects Domain}
\label{sec:acc-autom-eff}

In this section, we introduce an automata-based dependent effects domain $\effdom_\mnfa$ \added{that can be used to instantiate the domain of dependent effects $\effdom$ assumed by our type and effect system presented in \cref{sec:types-and-effects}. The domain is parametric in an automaton} 
$\mnfa$ that specifies the property to be verified for a given program. That is, the dependent effects domain is designed to support solving the following verification problem: given a program, show that the prefixes of the traces generated by the program are disjoint from the language recognized by $\mnfa$. To this end, the abstract domain tracks the reachable states of the automaton: each time the program emits an event, $\mnfa$ advances its state according to its transition relation. The set of automata states is in general infinite, so we abstract $\mnfa$'s transition relation by abstract interpretation. The abstraction takes into account the program environment at the point where the event is emitted, thus, yielding a domain of \emph{dependent} effects.

\subsection{Symbolic Accumulator Automata}

Our automaton model is loosely inspired by the various notions of (symbolic) register or memory automata considered in the literature \cite{DBLP:journals/tcs/KaminskiF94,DBLP:conf/cp/BeldiceanuCP04,DBLP:conf/cav/DAntoniFS019}. A \emph{symbolic accumulator automaton (SAA)} is defined over a potentially infinite alphabet and a potentially infinite data domain. In the following, we will fix both of these sets to coincide with the set of primitive values $\cvalues$ of our object language.
Formally, an SAA is a tuple $\mnfa=\tuple{Q, \Delta, \tuple{q_0,a_0},F}$. 
We specify the components of the tuple on-the-fly as we define the semantics of the automaton.

A state $\tuple{q,\vacc}$ of $\mnfa$ consists of a control location $q$ drawn from the finite set $Q$ and a value $\vacc \in \cvalues$ that indicates the current value of the accumulator register. The pair $\tuple{q_0,a_0}$ with $q_0 \in Q$ and $a_0 \in \cvalues$ specifies the initial state of $\mnfa$. The set $F \subseteq Q$ is the set of final control locations.

The symbolic transition relation $\Delta$ denotes a set of transitions $\tuple{q, \vacc} \xrightarrow{\cval} \tuple{q', \vacc'}$ that take a state $\tuple{q,\vacc}$ to a successor state $\tuple{q',\vacc'}$ under input symbol $\cval \in \cvalues$. The transitions are specified as a finite set of symbolic transitions $\tuple{q, \guard, \upd, q'} \in \Delta$, written $q \xrightarrow{\{\guard\}\upd} q'$, where $\guard \in \guards$ is a \emph{guard} and $\upd \in \updates$ an \emph{(accumulator) update}. Both guards and updates can depend on the input symbol $\cval$ consumed by the transition and the accumulator value $\vacc$ in the pre state, allowing the automaton to capture non-regular properties and complex program variable dependencies. We make our formalization parametric in the choice of the languages that define the sets $\guards$ and $\updates$\footnote{\added{In our implementation, we use integer arithmetic expressions for $\updates$ and conjunctions of (in)equality predicates for $\guards$.}}. To this end, we assume denotation functions $\den{\guard}(\cval,\vacc) \in \BB$ and $\den{\upd}(\cval,\vacc) \in \cvalues$ that evaluate a guard $\guard$ to its truth value, respectively, an update $\upd$ to the new accumulator value. We then have $\tuple{q,\vacc} \xrightarrow{\cval} \tuple{q',\vacc'}$ if there exists $q \xrightarrow{\{\guard\}\upd} q' \in \Delta$ such that $\den{\guard}(\cval,\vacc)=\mathtt{true}$ and 
$\den{\upd}(\cval,\vacc)=\vacc'$.
We require that $\Delta$ is such that this transition relation is total.
For $\ceff \in \cvalues^*$, we denote by $\tuple{q,\vacc} \xrightarrow{\ceff}\!\!{}^* \tuple{q',\vacc'}$ the reflexive transitive closure of this relation 
and define the \emph{semantics of a state} as the set of traces that reach that state:
\[\den{\tuple{q,\vacc}} = \pset{\ceff}{\tuple{q_0,a_0} \xrightarrow{\ceff}\!\!{}^* \tuple{q,\vacc}}\enspace.\]
With this, the language of $\mnfa$ is defined as
\[\lang(\mnfa) = \bigcup \pset{\den{\tuple{q,\vacc}}}{q \in F}\enspace.\]
Intuitively, $\lang(\mnfa)$ is the set of all \emph{bad} prefixes of event traces that the program is supposed to avoid.

\subsection{Automata-Based Dependent Effects Domain}

We now describe the domain $\effdom_\mnfa$. For the remainder of this section, we fix an SAA $\mnfa$ and omit subscript $\mnfa$ for $\effdom_\mnfa$ and all its operations. 

\added{
\paragraph{Concrete automata domain of dependent effects}
Recall from \cref{sec:types-and-effects} that a dependent effect domain $\effdom$ represents a sublattice of $\powerset(\cvalues^* \times \cenvs)$. Since the states of $\mnfa$ represents sets of event traces, a natural first step to define such a sublattice is to pair off automaton states with value environments: $\effdom_C = \powerset(Q \times \cvalues \times \cenvs)$. 

The corresponding concretization function $\effgamma_C:\effdom_C \to \powerset(\cvalues^* \times \cenv)$ is given by:
\[\effgamma_C(\eff_C) = \bigcup_{\tuple{q,\vacc,\cenv} \in \eff_C} \pset{\tuple{\ceff,\cenv}}{\ceff \in \den{\tuple{q,\vacc}}} \enspace.\]
Since $\effgamma_C$ is defined element-wise on $\effdom_C$, it is easy to see that it is monotone and preserves arbitrary meets. It is therefore the upper adjoint of a Galois connection between $\powerset(\cvalues^* \times \cenvs)$ and $\effdom_C$. Let $\effalpha_C$ be the corresponding lower adjoint, which is uniquely determined by $\effgamma_C$.

The operations on the dependent effect domain $\effdom_C$ are then obtained calculationally as the best abstractions of their concrete counterparts. In particular, we define:
\begin{align*}
\eff_C \odot_C \bval = {} & \effalpha_C(\pset{\pair{\ceff \cdot \cval}{\cenv}}{\tuple{\cval,\cenv} \in \bgamma(\bval) \land \tuple{\ceff,\cenv} \in \effgamma_C(\eff_C)})\\
= {} & \pset{\tuple{q',\vacc',\cenv}}{\exists \tuple{\cval,\cenv} \in \bgamma(\bval), \tuple{q,\vacc,\cenv} \in S.\, \tuple{q,\vacc} \xrightarrow{\cval} \tuple{q',\vacc'}}
\enspace.
\end{align*}
The characterization of $\odot_C$ relies on the fact that the transition relation of the automaton is total.
Note that the soundness condition on $\odot_C$ imposed in \cref{sec:types-and-effects} is obtained by construction from the properties of Galois connections:
\[\effgamma_C(\eff_C \odot_C \bval) \supseteq \pset{\pair{\ceff \cdot \cval}{\cenv}}{\tuple{\cval,\cenv} \in \bgamma(\bval) \land \tuple{\ceff,\cenv} \in \effgamma_C(\eff_C)}\enspace.\]
The remaining operations $\tstrengthen{\eff}{\tenv}$, $\eff_C[y/x]$, and $\texists{x}{\tval}{\eff_C}$ are obtained accordingly.

\paragraph{Abstract automata domain of dependent effects.}
Since the elements $S \in \effdom_C$ can be infinite sets, the operations on $\effdom_C$ such as $\odot_C$ are typically not computable. We therefore layer further abstractions on top of $\effdom_C$ to obtain an abstract automata domain of dependent effects with computable operations.

We proceed in two steps. Firstly, we change the representation of our abstract domain elements by partitioning the elements of each $\eff_C \in \effdom_C$ based on the control location of the automaton state. That is, we switch to the effect domain $\effdom_R = Q \to \powerset(\cvalues \times \cenvs)$, ordered by pointwise subset inclusion. The corresponding concretization function $\effgamma_R \in \effdom_R \to \effdom_C$ is given by
\[\effgamma_R(\eff_R) = \pset{\tuple{q,\vacc,\cenv}}{\tuple{\vacc,\cenv} \in \eff_R(q)}\enspace.\]
Clearly, we do not lose precision when changing the representation of the elements $\eff_C \in \effdom_C$ to elements of $\effdom_R$. In fact, $\effgamma_R$ is a lattice isomorphism. Its inverse $\effalpha_R = {\effgamma_R}^{-1}$ is the lower adjoint of a Galois connection between $\effdom_C$ and $\effdom_R$.

As before, we obtain the abstract domain operations on $\effdom_R$ by defining them as the best abstractions of their counterparts on $\effdom_C$. In particular, we define 
\begin{align*}
\eff_R \odot_R \bval =  \effalpha_R (\effgamma_R(\eff_R) \odot_C \bval)
= \lambda q'.\, \pset{\tuple{\vacc',\cenv}}{\exists q', \tuple{\cval,\cenv} \in \bgamma(\bval), \tuple{\vacc,\cenv} \in \eff_R(q).\, \tuple{q,\vacc} \xrightarrow{\cval} \tuple{q',\vacc'}}\enspace.
\end{align*}

Now consider again our abstract domain of base refinement types $\langle \bdom, \bord, \bbot, \btop, \bjoin, \bmeet \rangle$ that we have assumed as a parameter of the type and effects system of \cref{sec:types-and-effects}. Recall that each element $\bval \in \bdom$ abstracts a relation between values and value environments: $\bgamma(\bval) \subseteq \powerset(\cvalues \times \cenvs)$. We can thus reuse this domain to abstract the relations $\eff_R(q)\subseteq \powerset(\cvalues \times \cenvs)$ between the reachable accumulator values at location $q$ of the automaton and the environments. This leads to the following definition of our final automata-based effect domain: $\effdom = Q \to \bdom$. The accompanying concretization function $\effgamma_\bdom \in \effdom \to \effdom_R$ is naturally obtained by pointwise lifting of $\bgamma$: $\effgamma_\bdom(\eff) = \bgamma \circ \eff$. The overall concretization function $\effgamma:\effdom \to \powerset(\cvalues^* \times \cenvs)$ is defined by composition of the intermediate concretization functions: $\effgamma = \effgamma_C \circ \effgamma_R \circ \effgamma_\bdom$.

We then define the operations on $\effdom$ in terms of the operations on $\bdom$. Again, we focus on the operator $\odot$. The remaining operations are defined similarly.

Our goal is to ensure that the overall soundness condition on $\odot $ is satisfied. We achieve this by defining $\eff \odot \bval$ such that
\begin{equation}
  \label{eq:ext-op-sound}
  \effgamma_\bdom(\eff \odot \bval) \; \supseteq \; \effgamma_\bdom(\eff) \odot_R \bval  \enspace.
\end{equation}
Assuming (\ref{eq:ext-op-sound}) the overall soundness of $\odot$ then follows by construction:
\begin{lemma}
  For all $\eff \in \effdom$ and $\bval \in \bdom$,
  $\effgamma(\eff \odot \bval) \supseteq \{\pair{\ceff \cdot \cval}{\cenv} \mid \tuple{\cval,\cenv} \in \tgamma(\bval) \land \tuple{\ceff,\cenv} \in \effgamma(\eff)\}$.
\end{lemma}

Let us thus define an appropriate $\odot$ that satisfies (\ref{eq:ext-op-sound}). To this end, we first expand $\effgamma_\bdom(\eff) \odot_R \bval$:
\begin{align*}
& \effgamma_\bdom(\eff) \odot_R \bval\\
  = {} & \lambda q'.\, \pset{\tuple{\vacc',\cenv}}{\exists q', \tuple{\cval,\cenv} \in \bgamma(\bval), \tuple{\vacc,\cenv} \in \bgamma(\eff(q)).\, \tuple{q,\vacc} \xrightarrow{\cval} \tuple{q',\vacc'}}\\
= {} & \lambda q'.\, \!\!\! \bigcup_{q \xrightarrow{\{\guard\}\upd} q' \in \Delta} \!\!\! \pset{\tuple{\vacc',\cenv}}{\exists \tuple{\cval,\cenv} \in \bgamma(\bval), \tuple{\vacc,\cenv} \in \bgamma(\eff(q)).\, \den{\guard}(\cval,\vacc)=\mathtt{true} \land \den{\upd}(\cval,\vacc) = \vacc'}\enspace.
\end{align*}
The last equation suggests that we can compute $\eff \odot \beta$ by abstracting for each $q'$, each symbolic transition $q \xrightarrow{\{\guard\}\upd} q' \in \Delta$ of the automaton separately, and then take the join of the results. In order to abstract a symbolic transition, we need appropriate abstractions of the semantics of guards and updates with respect to base refinement types. For the sake of our formalization, we therefore assume an abstract interpreter $\den{\cdot}^\#: (\guards \cup \updates) \to \bdom \times \bdom \to \bdom$ such that for all $t \in \guards \cup \updates$ and $\bval,\bval' \in \bdom$
\begin{align*}
\bgamma(\den{t}^\#(\bval,\bval')) \supseteq \pset{\tuple{\cval',\rho}}{\exists \cval,\vacc.\, \tuple{\cval,\rho} \in \bgamma(\bval) \land \tuple{\vacc,\rho
} \in \bgamma(\bval') \land \cval' = \den{t}(\cval,\vacc)}\enspace.
\end{align*}
We then derive $\eff \odot \beta$ from the above equation as follows:
\[\eff \odot \beta = \lambda q'.\, \!\!\! \bigsqcup_{q \xrightarrow{\{\guard\}\upd} q' \in \Delta} \!\!\! \pset{\den{\upd}^\#(\beta \sqcap \beta_\guard, \eff(q) \sqcap \beta_\guard)}{\beta_\guard = (\exists \nu.\, \den{\guard}^\#(\beta,\eff(q)) \sqcap \{\nu = \mathtt{true}\})}\enspace.\]
Here, $\beta_\guard$ captures the environments $\rho$ shared by $\bval$ and $\eff(q)$ for which the guard $\guard$ evaluates to $\mathtt{true}$. It is used to strengthen $\bval$ and $\eff(q)$ when evaluating the update expression $\upd$. We here assume that $\bdom$ provides an operator $\exists v. \bval$ that projects out the value component of the pairs represented by some $\bval \in \bdom$. That is, we must have:
\[\bgamma(\exists v.\, \bval) = \pset{\tuple{\cval,\cenv}}{\exists \cval'.\, \tuple{\cval',\cenv} \in \bgamma(\bval)}\enspace.
\]
The fact that $\odot$ indeed satisfies (\ref{eq:ext-op-sound}) then follows from the assumption on the abstract interpreter for guards and expressions as well as the soundness of the abstract domain operations of $\bdom$.
\begin{lemma}
  For all $\eff \in \effdom$ and $\bval \in \bdom$,
  $\effgamma_\bdom(\eff \odot \bval) \supseteq \effgamma_\bdom(\eff) \odot_R \bval$.
\end{lemma}
Similarly, monotonicity of $\odot$ follows immediately from the monotonicity of the operations on $\bdom$.

The remaining operators on $\effdom$ assumed in \cref{sec:types-and-effects} (i.e., $\tstrengthen{\eff}{\tenv}$, $\eff[y/x]$, and $\texists{x}{\tval}{\eff}$) are obtained directly by a pointwise lifting of the corresponding operators on $\bdom$.
}

\section{Automated Inference and Verification}
\label{sec:inference}
\label{sec:inference-algo}
\label{apx:inference}

We now describe how to verify temporal safety properties of higher order programs, through automatic inference of accumulative types and effects. 
To facilitate the calculation of precise effects we build upon the existing data flow refinement type inference algorithm~\cite{DBLP:journals/pacmpl/PavlinovicSW21} based on abstract interpretation. We provide an abridged description of the original algorithm and explain how we adapt it for our purposes. We then briefly discuss the soundness of the resulting algorithm and how it can be used to automatically verify temporal safety properties.

\subsection{Type and Effect Inference by Abstract Interpretation}


\newcommand\steptf[1]{\text{$\mathsf{step^\#}$[\![$#1$]\!]}}
\newcommand\cveff[2]{\pair{#1}{#2}}
\newcommand\cvef{\mathit{tef}}


The data flow type inference, as described by \cite{DBLP:journals/pacmpl/PavlinovicSW21}, employs a calculational approach in an abstract interpretation style to iteratively compute a dependent refinement type for every subexpression of a program. The corresponding inference algorithm is implemented in the \drift{} tool.
%
%
The algorithm is parametric in the choice of an abstract domain of basic types $\bdom$ (which coincides with our parametrization of the types and effect system) as well as the supported primitive operations on values represented by these basic types (e.g., arithmetic operations, etc.).



The \drift{} algorithm is a whole program analysis. It assumes that every subexpression $\term$ of the program is labeled with a unique program location $\ell$, written $\term_\ell$. The abstract domain consists of \emph{execution maps}
$\cmap^\# \in \cmaps^\#$. Roughly speaking, an execution map assigns a type to every program location $\ell$. The type inference works by iteratively computing a fixpoint of an abstract transformer $\steptf{-}$. The abstract transformer takes an expression $\term_\ell$ and an execution map $\cmap^\#$, and computes an updated execution map reflecting the data flow in $\term_\ell$ based on what values may occur at each program location as specified by $\cmap^\#$. The abstract transformer is defined by structural recursion over $\term_\ell$.

At a conceptual level, we simply instantiate the \drift{} algorithm by treating event sequences as values that can be manipulated directly by the program, akin to the \removed{na\"{i}ve }translation approach. The only primitive operator defined on event sequences is $e_1 \eop e_2$ where $e_1$ is expected to evaluate to an event sequence $\ceff_1$ and $e_2$ to a value $\cval_2$. The result of the operation is the concatenated event sequence $\ceff_1 \cdot \cval_2$. We additionally have the constant expression $\epsilon$ denoting the empty event sequence. We also have a built-in pair constructor $\pair{e_1}{e_2}$ and projection operators $\#_1(e)$ and $\#_2(e)$ on pairs. Event sequences are then abstracted by treating an abstract effect $\eff$ as yet another kind of base type.

However, instead of just applying the instantiated \drift{} algorithm on translated programs that manipulate pairs of values and event sequences, we specialize the abstract transformer to take advantage of the knowledge that every expression $\term_\ell$ evaluates to such a pair. As such it can fuse together what would otherwise be costly joins and projections needed for analysis of the \removed{na\"{i}ve }product construction. Moreover, the specialized abstract transformer interprets the sequence concatenation operator using the abstract effect domain. This is in contrast to a \removed{na\"{i}ve }translation approach where, say, for the SAA effect domain, we would embed the automatons transfer function into the analyzed program and abstract it using the ordinary base types in $\bdom$. This specialization is key to improving both the efficiency and precision of the resulting analysis.

To build more intuition, we describe the specialized abstract transformer in some more detail. Its precise signature is
\[\steptf{-} : \cterms \to (\cenvs^\# \times \effdom) \to \cmaps^\# \to ((\tvalues \times \effdom)\times \cmaps^\#)\enspace.\]
Intuitively, for each well-formed expression $\term_\ell$ in a given environment $\Gamma \in \cenvs^\#$ and effect context $\eff \in \effdom$, and for a given execution map $\cmap \in \cmaps$, the transformer $\steptf{\term_\ell}(\Gamma, \eff)(\cmap)$ returns the updated abstract value and effect at $\ell$, along with an updated execution map.


%

At the core of the definition of $\steptf{-}$ lies the monotonic data flow propagation function $\tval_1 \ltimes \tval_2$ on refinement types shown in \cref{fig:transf1}. Intuitively, it ensures that an argument type $\tval$ at the call site of a function $f$ is propagated back to $f$'s definition site. After inferring the result type $\tval'$ of $f$ for $\tval$ from $f$'s body, $\tval'$ is in turn propagated forward to $f$'s call site.
For example, if $\tval_1$ is the current inferred type of some variable $x$ bound at location $\ell_1$, and $\tval_2$ is the current type inferred for some usage of $x$ at location $\ell_2$, then ${\tval_1} \ltimes {\tval_2}$ returns a new pair of types $\pair{\tval_1'}{\tval_2'}$ for locations $\ell_1$ and $\ell_2$ that reflects the forward data flow from $\ell_1$ to $\ell_2$ and backward data flow from $\ell_2$ to $\ell_1$.

In most cases, the abstract transformer behaves according to the original definition in the \drift{} algorithm and, additionally, simply carries along the effect.
The most interesting case
\begin{wrapfigure}[9]{r}[0pt]{7.3cm}
\vspace{-7pt}
\begin{minipage}{\linewidth}
\[\begin{array}{l}
  \steptf{(\evkw ~ e_1)_\ell}(\tenv, \eff) \;\triangleq\; \Mdo \\
  \quad \pair{\tau_\ell}{\eff_\ell} \;\leftarrow\; \Mget ~ (\tenv, \ell)\\
  \quad \cveff{\tval_1}{\eff_1} \;\leftarrow\; \steptf{e_1}(\tenv, \eff) \\
  \quad \massert{\tval_1 \neq \bot}\\
  \quad \Mlet \_, \pair{\tau'_\ell}{\eff'_\ell} \;=\; \cveff{\{\nu = \unitval\}[\tenv]}{\eff \odot \tval_1} \ltimes \pair{\tau_\ell}{\eff_\ell}\\
  \quad \Mupdate (\ell, \pair{\tau'_\ell}{\eff'_\ell})\\
  \quad \Mreturn \pair{\tau'_\ell}{\eff'_\ell}
\end{array}
\]
\vspace{7pt}
\end{minipage}
\end{wrapfigure}
is for event emission $\evkw~e_1$, shown on the right, which we discuss in more detail.
The definition uses a similar monadic style as in \cite{DBLP:journals/pacmpl/PavlinovicSW21} that treats $\steptf{-}$ as a state monad over execution maps. We use $\Mdo$ notation for the monadic composition, allowing $\Mlet$ to introduce new bindings, 
and we assume two operations, $\Mget$ and $\Mupdate$, that read from or write to the execution map encapsulated by the monad.
The transformer starts by extracting the type and effect currently associated with $\ell$ from the execution map and takes a recursive step on the expression $e_i$ that computes the value to be emitted. The $\massert$ construct aborts with the current execution map if the type inferred from $e_i$ is still $\bbot$ (indicating that $e_i$ has not yet produced an abstract value in the current iteration of the abstract interpretation).
Otherwise, it continues by computing the abstract result  of the event emission using the effect extension operator $\odot$ of the abstract effect domain. It then uses data flow propagation with the old type and effect at $\ell$ to compute the new $\pair{\tau'_\ell}{\eff'_\ell}$. This pair is then written back to the execution map at $\ell$ and returned.

\newcommand{\ttv}{\tval}
\newcommand{\ttef}[1]{\ttv_{#1}}

\begin{figure}
\[
  \arraycolsep=7pt\def\arraystretch{1.3}
  \begin{array}[t]{l|l}
    \begin{array}[t]{l}
      (x:\ttef{1i} \to \ttef{1o}) \ltimes (x:\ttef{2i} \to \ttef{2o}) \Def=\\
      \qquad\Mlet \pair{\ttef{2i}'}{\ttef{1i}'} = \ttef{2i} \ltimes \ttef{1i} \;\Min \\
      \qquad\Mlet \pair{\ttef{1o}'}{\ttef{2o}'} = \ttef{1o}[x: \tval_{2i}] \ltimes \ttef{2o}[x: \tval_{2i}] \;\Min\\
      \qquad\pair{x:\ttef{1i}' \to \ttef{1o}'\;}{\;x:\ttef{2i}' \to \ttef{2o}'}\\
      (x:\ttef{1} \to \ttef{2}) \ltimes \bbot \Def= \pair{x:\ttef{1} \to \ttef{2}\;}{\;x:\bbot \to \bbot}
    \end{array}
    &
      \begin{array}[t]{l}
        \ttv_{1} \ltimes \ttv_{2} \Def= \pair{\ttv_{1}}{\ttv_{1} \sqcup \ttv_{2}} \\
        \pair{\ttef{1a}}{\ttef{1b}} \ltimes \pair{\ttef{2a}}{\ttef{2b}} \Def= \\
        \qquad\pair{\pair{\ttef{1a}}{\ttef{1b}}}{\pair{\ttv_{1a} \sqcup \ttv_{2a}}{\ttv_{1b} \sqcup \ttv_{2b}}} \\
        \ttef{1} \ltimes \btop \Def= \pair{\btop}{\btop}
      \end{array}
  \end{array}
\]
\caption{\label{fig:transf1}
Data flow propagation}
\end{figure}

\subsection{Soundness of Type and Effect Inference}

A challenge in connecting the type inference result with our type and effects system is that the inference algorithm has been proven sound with respect to a bespoke dataflow semantics of functional program rather than a standard operational semantics like the one underlying our system. However, \cite{DBLP:journals/pacmpl/PavlinovicSW21} shows that the inference result yields a valid typing derivation in a bespoke data flow refinement type system. To bridge the gap in the soundness argument, we relate the \drift{} type system with our type and effect system at the abstract level by showing that, from the typing derivation for a translated program produced by the soundness proof of \cite{DBLP:journals/pacmpl/PavlinovicSW21}, one can reconstruct a typing derivation in the types and effects system for the original effectful program. Further details can be found in \apxref{\cref{sec:inference-soundness-top}}{the extended version~\cite{Nicola2025} of this paper}.


\subsection{Automated Verification}

As discussed in Sec.~\ref{sec:overview}, 
our abstract effect domain seeks to improve over a \added{direct}\removed{na\"{i}ve} approach of translating \removed{(via tuples or CPS)} an input program/property of effects into an effect-free product program that carries its effect trace and employs existing assertion checking techniques~\cite{DBLP:conf/pldi/KobayashiSU11,DBLP:journals/pacmpl/PavlinovicSW21, DBLP:journals/pacmpl/KawamataUST24,DBLP:conf/esop/YamadaKSS25}. 
This algorithm places a substantial burden on the type system (or other verification strategy) to track effect sequences as program values that flow from each (translated) event expression to the next. In this strategy, where an $\termkw{ev}\ e$ expression occurred in the original input program, the translated program has an event prefix variable (and accumulator variable) and constructs an extended event sequence. Unfortunately, today's higher-order program verifiers do not have good methods for summarizing program value sequences, nor do they exploit the automaton structure to organize possible sequence values. Thus, those tools struggle to validate the later $\termkw{assert}$ions.

The inference discussed above offers an alternative verification algorithm.
Once effects are inferred through the instantiation of our effect abstract domain (over the translated event sequences), it is straightforward to construct a verification algorithm. One merely has to ensure that at every program location  \circled{i}, the computed summary associates $\bot$ with every accepting state $q^\text{\circled{i}}_{err}$. 
Our organization of event prefixes around concrete automaton states allows us to better summarize those prefixes into categories, and can be thought of as a control-state-wise disjunctive partitioning. Thus, at each $\termkw{ev}\ e$ expression, the (dataflow) type system directly updates each $q$'s summary with the next event.
Sec.~\ref{sec:evaluation} experimentally evaluates both of these algorithms and compares them.

\section{Implementation, Trace Partitioning, and Benchmarks}
\label{sec:implementation}

{\bf Implementation.}
We implemented both the tuple/CPS translation (Sec.~\ref{sec:overview}) and the type and effect inference (Sec.~\ref{sec:inference}) verification algorithms in a prototype tool called \evdrift{}, as an extension to the \drift{}~\cite{DBLP:journals/pacmpl/PavlinovicSW21} type inference tool,
which
builds on top of the \textsc{Apron} library~\cite{DBLP:conf/cav/JeannetM09} to support various numerical abstract domains of type refinements.

\evdrift{} takes programs written in a subset of OCaml along with an automaton property specification file as input. 
\evdrift{} supports higher-order recursive functions, operations on primitive types such as integers and booleans, as well as
a non-deterministic if-then-else branching operator.
The property specification lists the set of automaton states, a deterministic transition function and an initial state. The specification also includes 
two kinds of effect-related assertions: those that must hold after every transition, and those that must hold after the final transition.
Assertions related to program variables (as in \drift) may be specified in the program itself. 
Whereas assertions related to effects may be specified in the property specification file.

We also implemented \added{three} improvements to the dataflow abstract interpretation.
First, \added{we integrate \textsc{Apron}'s} grid-polyhedra abstract domain~\cite{DBLP:phd/ethos/Dobson08}---a reduced product of the polyhedra~\cite{DBLP:conf/popl/CousotH78} and the grid~\cite {DBLP:conf/lopstr/BagnaraDHMZ06} abstract domains---to interpret type refinements of the form of $x \equiv y\mod 2$.
Second, we implemented 
trace partitioning~\cite{DBLP:journals/toplas/RivalM07} for increased disjunctive precision.
Although these benefits are somewhat orthogonal to our contributions, our evaluation (Sec.~\ref{sec:evaluation}) also experimentally quantifies the disjunctive benefit of trace partitioning in our setting vis-a-vis the benefit of our abstract effect domain.
\added{Third, we optimize the propagation step of the analysis by implementing the suggestion in \cite{DBLP:journals/pacmpl/PavlinovicSW21}.}
(Note that these three improvements also benefit the prior \drift{} tool.)
%


\medskip
\noindent{\bf Trace Partitioning.}
To improve the precision in our analysis, we implemented a type inference
\begin{wrapfigure}[7]{r}[7pt]{4.2cm}
\vspace{-12pt}
\begin{lstlisting}
let f x y = 
  let z =
    if y >= 0 then 1
    else -1
  in 
  assert z != 0; x/z
\end{lstlisting}
\end{wrapfigure}
algorithm with trace partitioning~\cite{DBLP:journals/toplas/RivalM07}, but instantiating it here in a higher-order setting. 
Consider
the example to the right adapted from~\cite{DBLP:journals/toplas/RivalM07}.
It is easy to see that this program does not raise an assertion error as \lstinline|z| is either equal to 1 or -1.
However, when using convex abstract domains like polyhedra and octagons, \lstinline|z| will have an abstract representation that includes the integer 0 because the abstract domain elements of the two branches of the conditional are joined together.
Consequently, an analysis would raise an undesirable assertion error.

Roughly speaking, with trace-partitioning, every if-then-else expression is analyzed twice: once in a context where the condition is true and once for false.
This directly alleviates the problem described in the above example as \lstinline|z$\ne$0| in the abstract representation of either of the branches.
\begin{wrapfigure}[5]{r}[7pt]{4.2cm}
\vspace{-12pt}
\begin{lstlisting}
let g x1 y1 = 
  assert y1 != 0; x1/y1
let f x2 y2 = 
  g x2 y2 + g (-x2) y2
\end{lstlisting}
\end{wrapfigure}
Beyond if-then-else expressions, trace-partitioning is also useful to increase precision in the analysis of functions with multiple callsites.
As a demonstration, consider the example on the right.
Using convex abstract domains, the argument \lstinline|y1| to the function \lstinline|g| would have an abstract representation that includes the integer 0.
In the same spirit as for if-then-else expressions, we analyze the function \lstinline|g| in separate contexts for each call site.

The original \drift{} type system~\cite{DBLP:journals/pacmpl/PavlinovicSW21} extends the typing judgement by including call stacks to infer distinct refinement types for program nodes under distinct call stacks.
The extension for if-then-else partitioning, which involves program traces composed of locations where the conditional branching takes different paths, largely works in a similar way.
The only difference is that program traces, along with the respective context of the chosen condition, are \emph{emitted} by program nodes which have if-then-else expressions and are implemented as lists, whereas call stacks are naturally implemented as stacks.
Importantly, for this work, these extensions to the \drift{} type system can easily carry over to the \evdrift{} type system.
We give more details about this extended type system in \apxref{Apx.~\ref{sec:tp-details}}{the extended version~\cite{Nicola2025} of this paper}.

\begin{figure}[t]
\lstset{numbers=none,xleftmargin=0.0in,xrightmargin=.25in,basicstyle=\footnotesize\sffamily}
	\begin{tabular}{|l|l|l|l|}
		\hline
		\begin{minipage}[t]{1.1in}

\vspace{-8pt}
\begin{lstlisting}
let rec order d c = 
  if d > 0 then
    if d mod 2 = 0 
    then ev c 
    else ev (-c);
    order (d - 2) c
  else 0
let _ (dd cc:int) =
  order dd cc
\end{lstlisting}
		\end{minipage} &
		\begin{minipage}[t]{1.2in}
\vspace{-8pt}
\begin{lstlisting}
let rec spend n =
  ev (-1);
  if n <= 0 then 0 
  else spend (n-1)
let _ (gas n:int) = 
  if gas >= n and 
    n >= 0
  then
   (ev gas; spend n) 
  else 0
\end{lstlisting}
	\end{minipage} & 
\begin{minipage}[t]{1.3in}
\vspace{-8pt}
\begin{lstlisting}
let rec reent d =
  ev 1; (* Acq *)
  if d > 0 then
    if nondet() then
      reent (d-1);
      ev -1 (* Rel *)
    else skip
let _ d = 
  reent d;
  ev -1 (* Rel *)

\end{lstlisting}
\end{minipage} &
\begin{minipage}[t]{1.9in}
\vspace{-8pt}
\begin{lstlisting}
let rec compute vv bound inc = 
  ev vv;
  if vv = bound then 0 else
    compute (inc vv) bound inc
let min_max v = 
  let f = (fun t -> 
    if v>=0 then t-1 else t+1) in
  if v>=0
  then compute v (-1 * v) f
  else compute v (-1 * v) f
\end{lstlisting}
	\end{minipage} 
	\\
	\small{Only "c" or "-c" events}  & 
	\small{$\sum_{i}^{N} \leq \texttt{gas}$} &
	\small{\# Rel $\le$ \# Acq} & 
	\small{$\forall i>0. -v < \pi[i] < v$}\\
	\hline
\end{tabular}
\caption{\label{fig:more-examples} Further examples of our benchmarks. (See the supplement for sources and automata specifications.)}
\end{figure}

\medskip
\noindent{\bf Benchmarks}.
To our knowledge there are no existing benchmarks for higher-order programs with the general class of SAA properties described, although there are related examples in some fragments of SAA. We thus created such a suite from the literature, extending them, and creating new ones. We plan to contribute these \expBenchCount{} benchmarks to SV-COMP \cite{SVCOMP24}.
\Cref{fig:more-examples} lists some of them.
These benchmarks test our tool to verify a variety of SAA properties like (left-to-right in \cref{fig:more-examples}) tracking disjoint branches of a program, resource analysis, verifying a reentrant lock, and tracking the minimum/maximum of a program variable.
Other examples use the accumulator for
summation, maximum/minimum, monotonicity, etc. similar to those found in automata literature~\cite{DBLP:journals/tcs/KaminskiF94,DBLP:conf/cp/BeldiceanuCP04,DBLP:conf/cav/DAntoniFS019}.
We also include an auction smart contract~\cite{DBLP:conf/sp/StephensFMLD21} and adapt some example programs proposed in \cite{DBLP:conf/lics/Nanjo0KT18}.
These benchmarks involve verification of 
amortized analysis~\cite{DBLP:journals/toplas/IgarashiK05,DBLP:conf/popl/HoffmannAH11,DBLP:conf/popl/HoffmannDW17} for a pair of queues, and the verification of liveness and fairness for a non-terminating web-server. 
Finally, for several benchmarks, we created corresponding \emph{unsafe} variants by tweaking the program or property.
All benchmarks are provided in the supplement, and publicly available (\emph{URL omitted for reviewing}).

\section{Evaluation}
\label{sec:evaluation}


\newcommand*\Chk{\ding{52}\xspace}
\newcommand*\TO{\textbf{T}\xspace}
\newcommand*\ERR{\Radioactivity\xspace}
\newcommand*\MO{\textbf{M}\xspace}
\newcommand*\NOK{\ding{55}\xspace}
\newcommand*\FAIL{\Radioactivity\xspace}
\newcommand*\Unk{{\textbf{?}\xspace}}

\newcommand{\unsound}[1]{\Lightning #1}

\newcommand\humanCfgtrans[6]{$\langle tl\!\!:\!#2, tp\!\!:\!#3, io\!\!:\!#6, #5 \rangle$}
\newcommand\humanCfgdirect[6]{$\langle tl\!\!:\!#2, tp\!\!:\!#3, io\!\!:\!#6, #5 \rangle$} 

\newcommand*\benchtablerowstretch{0.7} 
\newcommand*\benchtabletabcolsep{1pt} 

We sought to answer two research questions:
\begin{enumerate}[\bf(1)]
	\item How does \evdrift{} compare with other state-of-the-art automated verification tools for higher-order programs?
	\item What is the effect of trace-partitioning on efficiency and accuracy?
\end{enumerate}

\paragraph{Comparing our approach with other methods.}
We aim to compare against the somewhat mature prior tools \drift{}, \rcaml{}, \mochi{}~\cite{github:mochi} and \rethfl{}~\cite{github:rethfl} which support higher-order programs and can validate assertions and even some temporal properties. \citet{DBLP:conf/popl/MuraseT0SU16} focus on liveness properties and reduce the problem to termination (also see Secs.~\ref{sec:intro} and~\ref{sec:conclusion} for discussions of other works and tools).
\drift{}, \rcaml{}, \mochi{}, and \rethfl{} do not directly operate on programs with effects, so we used our reduction to verifying assertions of higher-order (effect-free) programs, which also enables those tools to technically now be applied to SAA properties.
\added{We first apply a selective store passing transformation to effectful programs (based on an algorithm adapted from \citet{DBLP:journals/entcs/Nielsen01}), producing an optimized product program that preserves parts of the original effectful program.
The translation is guided by whether an expression observes an event generated during its evaluation.
Consequently, the configuration is passed selectively to only those expressions that may observe such events.}
%
%
%
\drift{} is discussed in Sec.~\ref{sec:inference}.
\rcaml{}\footnote{We use \rcaml{} at commit \texttt{299e979bfce7d9b0532586bfc42b449fd0451531} with the \coar{} config \texttt{config/solver/rcaml\_wopp\_spacer.json}.} is based on extensions of Constrained Horn Clauses, is part of \coar{}~\cite{github:coar}, and is built on top of several prior works~\cite{hiroshiICFP2024,DBLP:journals/pacmpl/KawamataUST24,DBLP:journals/pacmpl/SekiyamaU23}.
\mochi{}~\cite{DBLP:conf/pldi/KobayashiSU11} is a CEGAR-style software model checker based on higher-order recursion schemes and relies on either interpolating theorem provers or an ICE-based solver of Constrained Horn Clauses (CHC) for predicate discovery.
\added{\rethfl{} is a type-based validity checker for a fragment $\nu$HFL(Z) of HFL(Z), a higher-order fixed point logic extended with integers, to which the verification of higher-order functional programs is known to be reducible, and it leverages CHC solvers to infer predicates within a bespoke refinement type system for $\nu$HFL(Z).}
%
%
We also considered LiquidHaskell~\cite{DBLP:conf/haskell/VazouSJ14}, which includes an implementation~\cite{github:liquidhaskell}.
However, LiquidHaskell is somewhat incomparable because (i) it requires user interaction whereas our aim is full automation and (ii) the eager-versus-lazy evaluation order difference impacts the language semantics and possible event traces, so it is difficult to perform a meaningful comparison.
%
For each tool, we use the latest version available at the time of experiments and corresponded with the respective developers to ensure proper usage.
All our experiments were conducted on an x86\_64 AMD EPYC 7452 32-Core Machine with \SI{125}{Gi} memory.
We used \textsc{BenchExec} 3.29 \cite{github:benchexec} to ensure precise measurement for each run.

In our \evdrift{} experiments, we run all our benchmarks using several configurations: various combinations of context sensitivity, trace partitioning, numerical abstract domains, and added precision for inference of effect sequences. 
Context sensitivity---denoted ``$cs$''---is either set to "none" (0) or else to a call-site depth of 1.
So for example, a configuration with context-sensitivity 1 and trace partitioning enabled remembers the last call site and also the last if-else branch location.
For research question \#1, we evaluate the end-to-end improvement of all of our work on \evdrift{} over existing tools, so we use \drift{} (plus the tuple translation) \emph{without} trace partitioning---denoted ``$tp\!\!:\!\!F$''---and \evdrift{} \emph{with} trace partitioning ``$tp\!\!:\!\!T$''. Further below in research question \#2 we evaluate the degree to which \evdrift{} improves the state of the art due to the use of the abstract effect domain, versus through the use of trace partitioning (as well as the performance overhead of trace partitioning).
Regarding abstract domains, we use the loose version of the polyhedra domain~\cite{DBLP:journals/pacmpl/PavlinovicSW21} for all our benchmarks except for those that involve \textsf{mod} operations where we use the grid-polyhedra domain.
For polyhedra, we further consider two different widening configurations: standard widening and widening with thresholds.
For widening with thresholds \cite{DBLP:journals/corr/abs-cs-0701193}---denoted ``$th$''---we use a simple heuristic that chooses the conditional expressions in the analyzed programs as well as pairwise inequalities between the variables in scope as constraints.
The grid-polyhedra domain does not properly support threshold widening, so we only use standard widening here.
Finally, we add an option---denoted ``$io$''---to increase precision for effect sequences through relationships between accumulator variables before and after a function body is evaluated. $io$
helps capture the consequence every function has over accumulator variables more precisely, albeit at the cost of some performance due to increased number of variables in the abstract environment.
In the discussion below, we report only the result for the configuration that verified the respective benchmark in the least amount of time (as identified in {\bf Config} columns in the tables).
For instances where all versions fail to verify a benchmark, we report results for the fastest configuration for brevity.
We also include results for other configurations in \apxref{Apx.~\ref{apx:extended} found in the supplement}{the extended version~\cite{Nicola2025} of this paper}.

\newcommand*{\hlbad}{\cellcolor{red!20}}
\newcommand*{\hlgood}{\cellcolor{green!20}}

\newcommand\bestCfgdrift{{\red{cfg}}}
\newcommand\bestCfgevdrift{{\red{cfg}}}

\renewcommand\humanCfgdirect[6]{$\langle cs\!\!:\!\!#2, tp\!\!:\!\!#3, io\!\!:\!\!#6, #5 \rangle$}
\renewcommand\humanCfgtrans[6]{$\langle cs\!\!:\!\!#2, #5 \rangle$}
\begin{table}
{\renewcommand{\arraystretch}{\benchtablerowstretch}\setlength{\tabcolsep}{\benchtabletabcolsep}
\caption{\label{table:results} Comparison of \evdrift{} against assertion verifiers for effect-free programs: 
\drift{}, \rcaml{}, \mochi{} \added{and \rethfl{}} via our tuple reduction.
\rvremoved{Benchmarks marked $^+$ with involve the mod operation.}For each verifier, we show the result ("\Chk{}" for successful verification; "\Unk{}" for unknown; "\ERR{}" for some failure other than unknown; "\TO{}" for timeout - over 900 seconds; "\MO{}" for out of memory - over 2GB), CPU time in seconds, maximum memory used in megabytes and the chosen configuration for \drift{} and \evdrift{}.
For \drift{} and \evdrift{}, the configuration tuples show what context-sensitivity ($cs$) was used, if trace-partitioning was used ($tp$), if added precision for effects was used ($io$), and which abstract domain was used ($ls$ for loose-polyhedra, and $pg$ for grid-polyhedra).
\evdrift{} {\bf verified additional \expNewOverTrans{}}, {\bf \expNewOverCoarMochi}, {\bf \expNewOverRealMochi{}} and {\bf \expNewOverRethfl} that \drift{}, \rcaml{}, \mochi{} and \rethfl{}, respectively, could not, and it is {\bf \expSpeedupEVoverDrift{}$\times$ faster} than \drift{} on \drift{}-verifiable examples, {\bf \expSpeedupEVoverRcaml{}$\times$ faster} than \rcaml{} on \rcaml{}-verifiable examples, {\bf \expSpeedupEVoverRealMochi{}$\times$ faster} than \mochi{} on \mochi{}-verifiable examples, \added{and {\bf \expSpeedupEVoverRethfl{}$\times$ faster} than \rethfl{} on \rethfl{}-verifiable examples}.
}
  \begin{tabular}{l|c@{\hspace{0pt}}rc|cr|cr|cr|c@{\hspace{-2pt}}rc}
			\toprule
			& \multicolumn{9}{c|}{Prior Tools (via Tuple Reduction)} & \\
			& \multicolumn{3}{c|}{\drift{}} & \multicolumn{2}{c|}{\rcamlONLY{} } & \multicolumn{2}{c|}{\mochi{}}& \multicolumn{2}{c|}{\rethfl{}}  & \multicolumn{3}{c}{\evdrift{}}   \\ 
			{\bf Bench}  & {\bf Res} & {\bf CPU} & {\bf Config.} & {\bf Res} & {\bf CPU} & {\bf Res} & {\bf CPU} & {\bf Res} & {\bf CPU} & {\bf Res} & {\bf CPU} & {\bf Config.} \\ 
			\midrule
1. \texttt{\scriptsize all-ev-pos} & \Chk  & 0.6 & \humanCfgtrans{oopsla25july22}{0}{F}{T}{ls}{F} & \Chk  & 0.8 & \Chk  & 0.8 & \Chk  & 2.5 & \Chk  & 0.2 & \humanCfgdirect{oopsla25july22}{0}{F}{T}{ls}{F} \\ 
2. \texttt{\scriptsize alt-inev} & \Unk  & 60.8 & \humanCfgtrans{oopsla25july22}{0}{F}{F}{pg}{F} & \TO   & 901.1 & \Unk  & 69.2 & \Chk  & 4.7 & \Chk  & 2.4 & \humanCfgdirect{oopsla25july22}{0}{F}{T}{ls}{T} \\ 
3. \texttt{\scriptsize auction} & \Unk  & 55.1 & \humanCfgtrans{oopsla25july22}{1}{F}{T}{ls}{F} & \TO   & 901.1 & \Unk  & 90.7 & \Chk  & 18.6 & \Chk  & 2.7 & \humanCfgdirect{oopsla25july22}{0}{F}{T}{ls}{F} \\ 
4. \texttt{\scriptsize binomial\_heap} & \Unk  & 544.1 & \humanCfgtrans{oopsla25july22}{1}{F}{F}{pg}{F} & \TO   & 901.1 & \MO   & 581.2 & \Chk  & 4.3 & \Chk  & 2.2 & \humanCfgdirect{oopsla25july22}{0}{F}{T}{ls}{T} \\ 
5. \texttt{\scriptsize concurrent\_sum} & \Chk  & 1.7 & \humanCfgtrans{oopsla25july22}{0}{F}{T}{ls}{F} & \MO   & 12.4 & \Unk  & 190.7 & \Chk  & 4.2 & \Chk  & 0.2 & \humanCfgdirect{oopsla25july22}{0}{F}{T}{ls}{F} \\ 
6. \texttt{\scriptsize depend} & \Chk  & 0.1 & \humanCfgtrans{oopsla25july22}{0}{F}{T}{ls}{F} & \Chk  & 0.1 & \Chk  & 0.7 & \Chk  & 1.9 & \Chk  & 0.0 & \humanCfgdirect{oopsla25july22}{1}{T}{T}{ls}{F} \\ 
7. \texttt{\scriptsize disj-gte} & \Unk  & 35.8 & \humanCfgtrans{oopsla25july22}{0}{F}{F}{pg}{F} & \MO   & 505.6 & \Chk  & 198.6 & \Chk  & 4.9 & \Chk  & 2.2 & \humanCfgdirect{oopsla25july22}{0}{F}{T}{ls}{F} \\ 
8. \texttt{\scriptsize disj-nondet} & \Unk  & 10.1 & \humanCfgtrans{oopsla25july22}{0}{F}{T}{ls}{F} & \MO   & 528.5 & \Chk  & 45.0 & \Chk  & 3.8 & \Chk  & 2.3 & \humanCfgdirect{oopsla25july22}{0}{F}{T}{ls}{F} \\ 
9. \texttt{\scriptsize higher-order} & \Chk  & 1.5 & \humanCfgtrans{oopsla25july22}{0}{F}{T}{ls}{F} & \Chk  & 13.2 & \Chk  & 16.4 & \Chk  & 3.4 & \Chk  & 0.6 & \humanCfgdirect{oopsla25july22}{0}{F}{T}{ls}{T} \\ 
10. \texttt{\scriptsize intro-ord3} & \Chk  & 24.0 & \humanCfgtrans{oopsla25july22}{1}{F}{T}{ls}{F} & \MO   & 164.6 & \Unk  & 64.1 & \Unk  & 8.1 & \Chk  & 3.8 & \humanCfgdirect{oopsla25july22}{0}{F}{T}{ls}{T} \\ 
11. \texttt{\scriptsize lics18-amortized} & \Unk  & 273.6 & \humanCfgtrans{oopsla25july22}{0}{F}{F}{pg}{F} & \MO   & 307.6 & \MO   & 228.8 & \TO   & 901.1 & \Chk  & 6.4 & \humanCfgdirect{oopsla25july22}{0}{F}{T}{ls}{F} \\ 
12. \texttt{\scriptsize lics18-hoshrink} & \Unk  & 9.9 & \humanCfgtrans{oopsla25july22}{0}{F}{F}{pg}{F} & \Unk  & 0.1 & \Unk  & 1.3 & \Unk  & 3.7 & \Unk  & 7.0 & \humanCfgdirect{oopsla25july22}{1}{F}{F}{pg}{F} \\ 
13. \texttt{\scriptsize lics18-web} & \Unk  & 55.0 & \humanCfgtrans{oopsla25july22}{0}{F}{T}{ls}{F} & \ERR  & 3.7 & \TO   & 901.2 & \Chk  & 11.2 & \Chk  & 7.1 & \humanCfgdirect{oopsla25july22}{0}{F}{T}{ls}{F} \\ 
14. \texttt{\scriptsize market} & \Unk  & 276.9 & \humanCfgtrans{oopsla25july22}{1}{F}{T}{ls}{F} & \MO   & 481.9 & \TO   & 901.1 & \MO   & 5.5 & \Unk  & 56.2 & \humanCfgdirect{oopsla25july22}{0}{F}{F}{pg}{T} \\ 
15. \texttt{\scriptsize max-min} & \Unk  & 143.3 & \humanCfgtrans{oopsla25july22}{1}{F}{F}{pg}{F} & \MO   & 28.3 & \TO   & 900.9 & \Chk  & 78.5 & \Chk  & 43.1 & \humanCfgdirect{oopsla25july22}{1}{T}{T}{ls}{T} \\ 
16. \texttt{\scriptsize monotonic} & \Chk  & 2.3 & \humanCfgtrans{oopsla25july22}{0}{F}{T}{ls}{F} & \TO   & 901.1 & \Chk  & 2.8 & \Chk  & 4.1 & \Chk  & 0.5 & \humanCfgdirect{oopsla25july22}{0}{F}{T}{ls}{T} \\ 
17. \texttt{\scriptsize nondet\_max} & \Chk  & 2.3 & \humanCfgtrans{oopsla25july22}{0}{F}{T}{ls}{F} & \Chk  & 1.8 & \TO   & 901.1 & \TO   & 901.0 & \Chk  & 0.6 & \humanCfgdirect{oopsla25july22}{0}{F}{T}{ls}{F} \\ 
18. \texttt{\scriptsize num\_evens} & \Chk  & 9.2 & \humanCfgtrans{oopsla25july22}{0}{F}{T}{ls}{F} & \TO   & 900.7 & \Chk  & 17.2 & \Chk  & 4.3 & \Chk  & 4.7 & \humanCfgdirect{oopsla25july22}{1}{F}{T}{ls}{T} \\ 
19. \texttt{\scriptsize order-irrel-nondet} & \Unk  & 26.2 & \humanCfgtrans{oopsla25july22}{1}{F}{F}{pg}{F} & \Chk  & 2.9 & \Chk  & 61.9 & \Chk  & 8.6 & \Chk  & 2.7 & \humanCfgdirect{oopsla25july22}{1}{T}{T}{ls}{T} \\ 
20. \texttt{\scriptsize overview1} & \Chk  & 1.8 & \humanCfgtrans{oopsla25july22}{1}{F}{T}{ls}{F} & \Chk  & 1.7 & \Chk  & 3.5 & \Chk  & 2.5 & \Chk  & 0.3 & \humanCfgdirect{oopsla25july22}{0}{F}{T}{ls}{T} \\ 
21. \texttt{\scriptsize reentr} & \Chk  & 3.4 & \humanCfgtrans{oopsla25july22}{0}{F}{T}{ls}{F} & \TO   & 900.1 & \Unk  & 590.2 & \Chk  & 6.7 & \Chk  & 0.2 & \humanCfgdirect{oopsla25july22}{0}{F}{T}{ls}{T} \\ 
22. \texttt{\scriptsize resource-analysis} & \Chk  & 2.9 & \humanCfgtrans{oopsla25july22}{0}{F}{T}{ls}{F} & \MO   & 65.6 & \Chk  & 1.3 & \Chk  & 2.6 & \Chk  & 0.2 & \humanCfgdirect{oopsla25july22}{0}{F}{T}{ls}{F} \\ 
23. \texttt{\scriptsize sum-appendix} & \Chk  & 1.3 & \humanCfgtrans{oopsla25july22}{0}{F}{T}{ls}{F} & \ERR  & 0.1 & \Chk  & 1.2 & \Chk  & 1.8 & \Chk  & 0.0 & \humanCfgdirect{oopsla25july22}{0}{F}{T}{ls}{T} \\ 

\emph{geomean for \Chk's:} & & {\bf \expGMevtrans{}} & & & {\bf \expGMcoarmochi{} } & & {\bf \expGMrealmochi{}} & & {\bf \expGMrethfl{}} & & {\bf \expGMdirect{}} \\
\bottomrule
\end{tabular}
\begin{center}
\emph{In addition to the above, we also provide \expBenchCount{} \emph{unsafe} benchmarks.\\
\evdrift{} analyzed all of them (using the Config in the last column above) in 
273s.}
\end{center}
}
\vspace{-1.3em}
\end{table}

Table~\ref{table:results} summarizes the results of our comparison. 
\evdrift{} significantly outperforms the other three tools in terms of number of benchmarks verified and efficiency.
\drift{} via tuple reduction was only able to verify \expDriftVerified{} of the \expBenchCount{} benchmarks, while \evdrift{} could verify \expEDriftVerified{}.
\added{Moreover, \evdrift{} could also verify the \texttt{market} benchmark with configuration \humanCfgdirect{oopsla25july22}{2}{F}{T}{ls}{T} in under 30 seconds.}
Across all benchmarks that \drift{} could solve, it had a geomean of \SI{\expGMevtrans}{s}. Across all benchmarks that \evdrift{} could solve, it had a geomean of \SI{\expGMdirect}{s}.
For those benchmarks that \emph{both} \drift{} and \evdrift{} could verify, \evdrift{} was \expSpeedupEVoverDrift{}$\times$ faster.
\rcaml{} successfully verified 6 out of \expBenchCount{} benchmarks, with a geomean of \SI{\expGMcoarmochi}{s} across these, and was unable to verify the others due to either imprecision, timeout, or memory blowup.
%
%
%
\rethfl{} verified \expRethflVerified{} benchmarks for which it reported a geomean of \SI{\expGMrethfl}{s}.
\rvadded{Using the web interface for \rethfl{}~\cite{web:rethfl}, we found that an additional benchmark (\lstinline{nondet_max}) is verified in under \SI{5}{s}, suggesting that configuration issues in the running environments caused the difference. Moreover, the verification of \lstinline{market} and \lstinline{lics18-amortized} using the web interface was not conclusive, with the printed output remaining in a seemingly frozen state.
}
  Finally, \mochi{} verified \expRealMochiVerified{} benchmarks for which it reported a geomean of \SI{\expGMrealmochi}{s}.
\rvremoved{However, 5 of our benchmarks (indicated in Table~1)  use the mod operation that \mochi{} and \rethfl{} soundly approximate by treating it as an interpreted function, and that leads to imprecision.}
As anticipated, the cross-product transformation of the original program and property significantly increases the program size, thus requiring high context-sensitivity, which no existing tool provides with high precision.
%
In addition, we also ran \evdrift{} on the unsafe variants of our benchmarks. 
We used the same configurations as for their respective safe versions in Table~\ref{table:results}. 
The individual results are omitted for lack of space, but \evdrift{} analyzed all unsafe benchmarks in 273 seconds, returning unknown on each of them. \added{(The tool can only prove the absence of errors, not their presence.)}

We deduced at least three major factors behind \evdrift{}'s superior performance. 
(1) \evdrift{} evaluates the $\termkw{ev}$ expressions inline which reduces program size significantly. 
This leads to significantly faster runtimes and smaller memory footprint for \evdrift{} for all benchmarks. 
This also reduces a function call and the need to remember another call site for \evdrift{} in some cases like \texttt{overview1} and \texttt{sum-appendix} where the $\termkw{ev}$ expression might have different arguments at different locations.
Moreover, some benchmarks require the inference of non-convex input-output relations for functions, which the used numerical abstract domains cannot express.
This is why \evdrift{} can verify several benchmarks like \removed{\texttt{disj},} 
\texttt{lics18-web}, \texttt{higher-order}, etc.\ that \drift{} cannot.
(2) \evdrift{} establishes concrete relationships between program variables and accumulator variables leading to increased precision especially in resource-analysis-like benchmarks. 
(3) \evdrift{}'s abstract domain adds some inbuilt disjunctivity reasoning that learns different relationships for different final states.
This adds efficiency and precision to \evdrift{}'s analysis as it is able to verify some benchmarks, like \texttt{disj-gte} \removed{and \texttt{disj}} that \removed{have}\added{has} if-then-else statements, without using trace partitioning.
\added{Furthermore, this disjunctivity enables evDrift to effectively track how individual states transition in every function, which is especially useful for higher-order programs like \texttt{intro-ord3} and \texttt{market}.}

Due to their similarities, \evdrift{} also inherits a few limitations from \drift{}.
\evdrift{} fails to verify \texttt{lics18-hoshrink}, which involves non-linear invariants presently not expressible by any of the abstract domains in the \textsc{Apron} library.
\added{\drift{} evaluates all nodes repeatedly until convergence of the whole program.
We think it is possible to avoid such repetitions by selecting the order in which nodes are evaluated based on the data-flow dependency graph.}
We expect this to bring significant improvements in practice.

We also evaluated the impact of trace partitioning on the precision and performance of both \drift{} and \evdrift{}. 
For lack of space, the details can be found in \apxref{Apx.~\ref{apx:tp-impact}}{the extended version~\cite{Nicola2025} of this paper}.
In summary, \evdrift{} is able to verify two more benchmarks with trace partitioning, yet it slows down by \expTPSpeedupevDrift{}$\times$.
\rvadded{Finally, since \evdrift{} is parametric on abstract domains, it can apply specialized domains when program features necessitate it. For example, the first program in Fig.~\ref{fig:more-examples} involves the \lstinline|mod| operator. \evdrift{} is able to verify this example and four others examples (see benchmarks \texttt{last-ev-even}, \texttt{order-irrel}, \texttt{sum-of-ev-even} and \texttt{temperature} included with the release of \evdrift{}).}

%
%
%
%
%


\section{Conclusion}
\label{sec:conclusion}

We have introduced the first abstract interpretation for inferring types and effects of higher-order programs. The effect abstract domain disjunctively organizes summaries (abstractions) of partitions of possible event trace prefixes around the concrete automaton states they reach. Our effects are captured in a refinement type-and-effect system and we described how to automate their inference through abstract interpretation. We then showed that our implementation \evdrift{} enables numerous new benchmarks to be verified (or enables faster verification by $\expSpeedupEVoverDrift\times$ on \drift{}-verifiable programs), as compared with prior effect-less tools (\drift{}, \rcaml{}, \mochi{}, and \rethfl{}) which require translations to encode effects.

\paragraph{Related work}
We discussed some related works in Sec.~\ref{sec:intro} and as relevant throughout the paper. We now remark in some more detail and mention further related works.
The work of \citet{DBLP:journals/pacmpl/PavlinovicSW21} is the most related, but their type system does not include effects or automata, nor do they support any of our new benchmarks. However, we have been inspired by their work and build on aspects of their type system, abstract interpretation and implementation. 

\added{Several prior works have explored systems for reasoning about sequential effects such as thread synchronization~\cite{DBLP:conf/pldi/FlanaganQ03}, heap mutation~\cite{DBLP:conf/pldi/SwamyWSCL13}, producer effects~\cite{DBLP:conf/popl/Tate13}, and temporal properties~\cite{DBLP:conf/csl/KoskinenT14,DBLP:conf/csl/0001C14,DBLP:conf/lics/Nanjo0KT18} as well as unified frameworks for such systems~\cite{DBLP:conf/popl/Katsumata14,DBLP:conf/ecoop/Gordon17,DBLP:journals/toplas/Gordon21}. Notably, our accumulative type and effect system bears similarity with instances of Gordon's polymorphic type and effect system~\cite{DBLP:conf/ecoop/Gordon17,DBLP:journals/toplas/Gordon21}. Similar to our work, his framework is parametric in an algebraic structure of effects, a so-called \emph{effect quantale}, that abstracts from how effects are accumulated along and across program traces. However, the focus of \cite{DBLP:conf/ecoop/Gordon17,DBLP:journals/toplas/Gordon21} is on developing the meta-theory of such type systems rather than the problem of practical type and effect inference. So there are some key technical differences that stem from our focus on the latter problem.
Notably, sequential composition of effects in an effect quantale must distribute over joins in both arguments. In contrast, the effect extension operator in our work must only be a monotone upper-approximation of trace extension (\cref{eq:effect-extension-soundness}), a weaker requirement that gives more flexibility when designing abstract domains for effect inference. We exploit this flexibility in our implementation to trade precision for efficiency.} \added{\citet{DBLP:conf/sas/GordonY23} explore constraint-based type inference and error localization algorithms for effect quantales, but they do not consider any instances of the general type system that feature dependent effects and type refinements comparable to ours.}

\added{
\citet{DBLP:journals/pacmpl/ZhouYDJ24} propose Hoare Automata Types (HAT), which augment a refinement type system with an automata-based representation of pre- and post conditions for tracking sequential effects. Unlike our work, which focuses on effect inference, their work provides an algorithm for checking user-provided type and effect annotations.
The user expresses temporal effects in linear temporal logic (LTL) on finite traces (respectively, symbolic regular languages)~\cite{DBLP:conf/ijcai/GiacomoV13}. These formulas are then compiled to a variant of symbolic automata~\cite{DBLP:conf/cav/DAntoniFS019}, enabling SMT-based type checking. The compilation to automata requires that the underlying symbolic automata class is closed under effective complementation. This limits expressivity. Notably, it rules out accumulator registers like those supported by SAA. In contrast, to enable automatic effect inference, our approach approximates the rich SAA semantics by abstract interpretation using a broad class of abstract domains (such as polyhedra, which are not closed under complementation). That said, HAT provide support for rich properties of algebraic data types, which our current implementation does not yet handle. Though, we see the extension of the analysis with support for algebraic data types as mostly orthogonal to the handling of effects.}

\added{\citet{DBLP:journals/pacmpl/NguyenGTH18} present a method for verifying contracts of stateful untyped higher-order programs. The analysis uses symbolic execution and relies on a form of predicate abstraction to obtain refinement predicates for over-approximating the program semantics. Unlike our whole program analysis, this approach is compositional, enabling the verification of program components. However, the analysis requires user-provided contracts to enable this compositional analysis.}

%
\citet{DBLP:conf/csl/0001C14} discuss abstractions of B\"{u}chi automata, building their abstractions by using equivalence classes and subsequences of traces to separately summarize the finite and the infinite traces. They then discuss a B\"{u}chi type \& effect system, but it is not accumulative in nature, nor do they provide an implementation.
\citet{DBLP:conf/popl/MuraseT0SU16} described a method of verifying temporal properties of higher-order programs through the \citet{DBLP:conf/lics/Vardi87} reduction to fair termination. We considered using some of their benchmarks, however none were suitable because the overlap between their work and ours is limited for two reasons: (i) they focus on verifying \added{both liveness and safety} while we only verify safety properties and (ii) we support expressive SAA-based safety properties, which they do not support. 
\rcamlONLY{} is a verifier for OCaml-like programs with refinement types, is based on extensions of Constrained Horn Clauses and is part of \coar{}. \rcamlONLY{} was developed as part of several works~\cite{hiroshiICFP2024,DBLP:journals/pacmpl/KawamataUST24,DBLP:journals/pacmpl/SekiyamaU23}.
Kobayashi's \cite{DBLP:conf/popl/Kobayashi09,DBLP:conf/pldi/KobayashiSU11} higher-order model checking, notably the approaches based on counterexample-guided abstraction refinement, CEGAR, is orthogonal to our proposed analysis. We make different trade-offs both in terms of algorithmic techniques as well as theoretical guarantees. In particular, unlike CEGAR-based approaches, our analysis is guaranteed to always terminate (although, CEGAR-based approaches can also show that the program is unsafe).
\added{Recently, \citet{DBLP:conf/esop/YamadaKSS25} explored the relationship between Dijkstra Monads~\cite{DBLP:conf/pldi/SwamyWSCL13} and the higher-order fixpoint logic HFL(Z)~\cite{DBLP:journals/pacmpl/KobayashiTST23} for  automated verification of higher-order programs~\cite{DBLP:conf/esop/0001TW18}.}
\rvremoved{The work has resulted in the new tool \rethfl{}, which we have included in our
experimental comparison.}\rvadded{Their proposed verification approach has been implemented in the tool \rethfl{}~\cite{conf/aplas/KatsuraIKT20}, which we have included in our experimental comparison.}

We have focused on events/effects that simply emit a value ($\termkw{ev}\; v$) that is unobservable to the program, and merely appears in the resulting event trace. By contrast, numerous recent works are focused on higher-order programming languages with \emph{algebraic effects and their handlers}. Such features allow programmers to define effects in the language, and create exception-like control structures for how to handle the effects. 
\citet{journals/pacmpl/LagoG24} detail semantics and model checking problems for higher-order programs that have effects such as references, effect handlers, etc. Although this work is quite general, it focuses on semantics and decidability, does not specifically target symbolic accumulator properties, and does not include an implementation.
~\citet{DBLP:journals/pacmpl/KawamataUST24} discuss a refinement type system for algebraic effects and handlers that supports changes to the so-called ``answer type.''



\paragraph{Future work} A natural next direction is to automate verification of properties extend beyond safety to liveness specified by, say, B\"{u}chi automata or other infinite word automata, perhaps with an accumulator. Such an extension would require infinite trace semantics for the programming language and type \& effect system (e.g.~\cite{DBLP:conf/csl/KoskinenT14}), as well as a combination of both least and greatest fixpoint reasoning for abstract interpretations.


\section*{Data-Availability Statement}
\ifdefined\extdversion{
We include the extended version of the paper that includes a few appendices in the supplement zip file along with the submission.
It includes a completely functioning codebase for evDrift and our benchmark set.
Since evDrift is built on top of Drift, one can also run Drift using the codebase we provide.
We include a readme file for running the tool and running the scripts to generate the results for Drift and evDrift in the tables in this paper.
A user can also give their own programs along with an automaton property specification file for analysis.
}
\else{
The software developed for this work and the full set of benchmarks used in our evaluation are available at Zenodo~\cite{zenodo:evdrift}
}
\fi

\section*{Acknowledgments}
We thank the OOPSLA 2025 anonymous reviewers for their constructive comments on this draft. We are also grateful to Hiroshi Unno for generously providing his guidance on the use of RCaml which we employ in our evaluation.
The work is supported by National Science Foundation under the grant agreements 2315363, 2107169, 2008633, and 2304758.

\bibliography{dblp,others}


\begin{thebibliography}{68}


\ifx \showCODEN    \undefined \def \showCODEN     #1{\unskip}     \fi
\ifx \showISBNx    \undefined \def \showISBNx     #1{\unskip}     \fi
\ifx \showISBNxiii \undefined \def \showISBNxiii  #1{\unskip}     \fi
\ifx \showISSN     \undefined \def \showISSN      #1{\unskip}     \fi
\ifx \showLCCN     \undefined \def \showLCCN      #1{\unskip}     \fi
\ifx \shownote     \undefined \def \shownote      #1{#1}          \fi
\ifx \showarticletitle \undefined \def \showarticletitle #1{#1}   \fi
\ifx \showURL      \undefined \def \showURL       {\relax}        \fi
\providecommand\bibfield[2]{#2}
\providecommand\bibinfo[2]{#2}
\providecommand\natexlab[1]{#1}
\providecommand\showeprint[2][]{arXiv:#2}

\bibitem[Bagnara et~al\mbox{.}(2006)]%
        {DBLP:conf/lopstr/BagnaraDHMZ06}
\bibfield{author}{\bibinfo{person}{Roberto Bagnara},
  \bibinfo{person}{Katy~Louise Dobson}, \bibinfo{person}{Patricia~M. Hill},
  \bibinfo{person}{Matthew Mundell}, {and} \bibinfo{person}{Enea Zaffanella}.}
  \bibinfo{year}{2006}\natexlab{}.
\newblock \showarticletitle{Grids: {A} Domain for Analyzing the Distribution of
  Numerical Values}. In \bibinfo{booktitle}{\emph{Logic-Based Program Synthesis
  and Transformation, 16th International Symposium, {LOPSTR} 2006, Venice,
  Italy, July 12-14, 2006, Revised Selected Papers}}
  \emph{(\bibinfo{series}{Lecture Notes in Computer Science},
  Vol.~\bibinfo{volume}{4407})},
  \bibfield{editor}{\bibinfo{person}{Germ{\'{a}}n Puebla}} (Ed.).
  \bibinfo{publisher}{Springer}, \bibinfo{pages}{219--235}.
\newblock
\href{https://doi.org/10.1007/978-3-540-71410-1\_16}{doi:\nolinkurl{10.1007/978-3-540-71410-1\_16}}


\bibitem[Bagnara et~al\mbox{.}(2008)]%
        {DBLP:journals/scp/BagnaraHZ08}
\bibfield{author}{\bibinfo{person}{Roberto Bagnara},
  \bibinfo{person}{Patricia~M. Hill}, {and} \bibinfo{person}{Enea Zaffanella}.}
  \bibinfo{year}{2008}\natexlab{}.
\newblock \showarticletitle{The Parma Polyhedra Library: Toward a complete set
  of numerical abstractions for the analysis and verification of hardware and
  software systems}.
\newblock \bibinfo{journal}{\emph{Sci. Comput. Program.}} \bibinfo{volume}{72},
  \bibinfo{number}{1-2} (\bibinfo{year}{2008}), \bibinfo{pages}{3--21}.
\newblock
\href{https://doi.org/10.1016/J.SCICO.2007.08.001}{doi:\nolinkurl{10.1016/J.SCICO.2007.08.001}}


\bibitem[Beldiceanu et~al\mbox{.}(2004)]%
        {DBLP:conf/cp/BeldiceanuCP04}
\bibfield{author}{\bibinfo{person}{Nicolas Beldiceanu}, \bibinfo{person}{Mats
  Carlsson}, {and} \bibinfo{person}{Thierry Petit}.}
  \bibinfo{year}{2004}\natexlab{}.
\newblock \showarticletitle{Deriving Filtering Algorithms from Constraint
  Checkers}. In \bibinfo{booktitle}{\emph{Principles and Practice of Constraint
  Programming - {CP} 2004, 10th International Conference, {CP} 2004, Toronto,
  Canada, September 27 - October 1, 2004, Proceedings}}
  \emph{(\bibinfo{series}{Lecture Notes in Computer Science},
  Vol.~\bibinfo{volume}{3258})}, \bibfield{editor}{\bibinfo{person}{Mark
  Wallace}} (Ed.). \bibinfo{publisher}{Springer}, \bibinfo{pages}{107--122}.
\newblock
\href{https://doi.org/10.1007/978-3-540-30201-8\_11}{doi:\nolinkurl{10.1007/978-3-540-30201-8\_11}}


\bibitem[Beyer(2024)]%
        {SVCOMP24}
\bibfield{author}{\bibinfo{person}{D. Beyer}.} \bibinfo{year}{2024}\natexlab{}.
\newblock \showarticletitle{State of the Art in Software Verification and
  Witness Validation: {SV-COMP 2024}}. In \bibinfo{booktitle}{\emph{Proc.\
  TACAS~(3)}} \emph{(\bibinfo{series}{LNCS~14572})}.
  \bibinfo{publisher}{Springer}, \bibinfo{pages}{299--329}.
\newblock
\href{https://doi.org/10.1007/978-3-031-57256-2_15}{doi:\nolinkurl{10.1007/978-3-031-57256-2_15}}


\bibitem[Blanchet et~al\mbox{.}(2007)]%
        {DBLP:journals/corr/abs-cs-0701193}
\bibfield{author}{\bibinfo{person}{Bruno Blanchet}, \bibinfo{person}{Patrick
  Cousot}, \bibinfo{person}{Radhia Cousot}, \bibinfo{person}{J{\'{e}}r{\^{o}}me
  Feret}, \bibinfo{person}{Laurent Mauborgne}, \bibinfo{person}{Antoine
  Min{\'{e}}}, \bibinfo{person}{David Monniaux}, {and} \bibinfo{person}{Xavier
  Rival}.} \bibinfo{year}{2007}\natexlab{}.
\newblock \showarticletitle{A Static Analyzer for Large Safety-Critical
  Software}.
\newblock \bibinfo{journal}{\emph{CoRR}}  \bibinfo{volume}{abs/cs/0701193}
  (\bibinfo{year}{2007}).
\newblock
\showeprint[arXiv]{cs/0701193}
\urldef\tempurl%
\url{http://arxiv.org/abs/cs/0701193}
\showURL{%
\tempurl}


\bibitem[Borkowski et~al\mbox{.}(2024)]%
        {DBLP:journals/pacmpl/BorkowskiVJ24}
\bibfield{author}{\bibinfo{person}{Michael Borkowski}, \bibinfo{person}{Niki
  Vazou}, {and} \bibinfo{person}{Ranjit Jhala}.}
  \bibinfo{year}{2024}\natexlab{}.
\newblock \showarticletitle{Mechanizing Refinement Types}.
\newblock \bibinfo{journal}{\emph{Proc. {ACM} Program. Lang.}}
  \bibinfo{volume}{8}, \bibinfo{number}{{POPL}} (\bibinfo{year}{2024}),
  \bibinfo{pages}{2099--2128}.
\newblock
\href{https://doi.org/10.1145/3632912}{doi:\nolinkurl{10.1145/3632912}}


\bibitem[Cousot and Halbwachs(1978)]%
        {DBLP:conf/popl/CousotH78}
\bibfield{author}{\bibinfo{person}{Patrick Cousot} {and}
  \bibinfo{person}{Nicolas Halbwachs}.} \bibinfo{year}{1978}\natexlab{}.
\newblock \showarticletitle{Automatic Discovery of Linear Restraints Among
  Variables of a Program}. In \bibinfo{booktitle}{\emph{Conference Record of
  the Fifth Annual {ACM} Symposium on Principles of Programming Languages,
  Tucson, Arizona, USA, January 1978}},
  \bibfield{editor}{\bibinfo{person}{Alfred~V. Aho},
  \bibinfo{person}{Stephen~N. Zilles}, {and} \bibinfo{person}{Thomas~G.
  Szymanski}} (Eds.). \bibinfo{publisher}{{ACM} Press},
  \bibinfo{pages}{84--96}.
\newblock
\href{https://doi.org/10.1145/512760.512770}{doi:\nolinkurl{10.1145/512760.512770}}


\bibitem[Dal\space{}Lago and Ghyselen(2024)]%
        {journals/pacmpl/LagoG24}
\bibfield{author}{\bibinfo{person}{Ugo Dal\space{}Lago} {and}
  \bibinfo{person}{Alexis Ghyselen}.} \bibinfo{year}{2024}\natexlab{}.
\newblock \showarticletitle{On Model-Checking Higher-Order Effectful Programs}.
\newblock \bibinfo{journal}{\emph{Proc. {ACM} Program. Lang.}}
  \bibinfo{volume}{8}, \bibinfo{number}{{POPL}} (\bibinfo{year}{2024}),
  \bibinfo{pages}{2610--2638}.
\newblock
\href{https://doi.org/10.1145/3632929}{doi:\nolinkurl{10.1145/3632929}}


\bibitem[D'Antoni et~al\mbox{.}(2019)]%
        {DBLP:conf/cav/DAntoniFS019}
\bibfield{author}{\bibinfo{person}{Loris D'Antoni}, \bibinfo{person}{Tiago
  Ferreira}, \bibinfo{person}{Matteo Sammartino}, {and}
  \bibinfo{person}{Alexandra Silva}.} \bibinfo{year}{2019}\natexlab{}.
\newblock \showarticletitle{Symbolic Register Automata}. In
  \bibinfo{booktitle}{\emph{Computer Aided Verification - 31st International
  Conference, {CAV} 2019, New York City, NY, USA, July 15-18, 2019,
  Proceedings, Part {I}}} \emph{(\bibinfo{series}{Lecture Notes in Computer
  Science}, Vol.~\bibinfo{volume}{11561})},
  \bibfield{editor}{\bibinfo{person}{Isil Dillig} {and} \bibinfo{person}{Serdar
  Tasiran}} (Eds.). \bibinfo{publisher}{Springer}, \bibinfo{pages}{3--21}.
\newblock
\href{https://doi.org/10.1007/978-3-030-25540-4\_1}{doi:\nolinkurl{10.1007/978-3-030-25540-4\_1}}


\bibitem[Dobson(2008)]%
        {DBLP:phd/ethos/Dobson08}
\bibfield{author}{\bibinfo{person}{Katy~Louise Dobson}.}
  \bibinfo{year}{2008}\natexlab{}.
\newblock \emph{\bibinfo{title}{Grid domains for analysing software}}.
\newblock \bibinfo{thesistype}{Ph.\,D. Dissertation}.
  \bibinfo{school}{University of Leeds, {UK}}.
\newblock
\urldef\tempurl%
\url{http://etheses.whiterose.ac.uk/1370/}
\showURL{%
\tempurl}


\bibitem[Ferles et~al\mbox{.}(2021)]%
        {DBLP:journals/pacmpl/FerlesSD21}
\bibfield{author}{\bibinfo{person}{Kostas Ferles}, \bibinfo{person}{Jon
  Stephens}, {and} \bibinfo{person}{Isil Dillig}.}
  \bibinfo{year}{2021}\natexlab{}.
\newblock \showarticletitle{Verifying correct usage of context-free {API}
  protocols}.
\newblock \bibinfo{journal}{\emph{Proc. {ACM} Program. Lang.}}
  \bibinfo{volume}{5}, \bibinfo{number}{{POPL}} (\bibinfo{year}{2021}),
  \bibinfo{pages}{1--30}.
\newblock
\href{https://doi.org/10.1145/3434298}{doi:\nolinkurl{10.1145/3434298}}


\bibitem[Flanagan and Qadeer(2003)]%
        {DBLP:conf/pldi/FlanaganQ03}
\bibfield{author}{\bibinfo{person}{Cormac Flanagan} {and} \bibinfo{person}{Shaz
  Qadeer}.} \bibinfo{year}{2003}\natexlab{}.
\newblock \showarticletitle{A type and effect system for atomicity}. In
  \bibinfo{booktitle}{\emph{Proceedings of the {ACM} {SIGPLAN} 2003 Conference
  on Programming Language Design and Implementation 2003, San Diego,
  California, USA, June 9-11, 2003}}, \bibfield{editor}{\bibinfo{person}{Ron
  Cytron} {and} \bibinfo{person}{Rajiv Gupta}} (Eds.).
  \bibinfo{publisher}{{ACM}}, \bibinfo{pages}{338--349}.
\newblock
\href{https://doi.org/10.1145/781131.781169}{doi:\nolinkurl{10.1145/781131.781169}}


\bibitem[Freeman and Pfenning(1991)]%
        {DBLP:conf/pldi/FreemanP91}
\bibfield{author}{\bibinfo{person}{Timothy~S. Freeman} {and}
  \bibinfo{person}{Frank Pfenning}.} \bibinfo{year}{1991}\natexlab{}.
\newblock \showarticletitle{Refinement Types for {ML}}. In
  \bibinfo{booktitle}{\emph{Proceedings of the {ACM} SIGPLAN'91 Conference on
  Programming Language Design and Implementation (PLDI), Toronto, Ontario,
  Canada, June 26-28, 1991}}, \bibfield{editor}{\bibinfo{person}{David~S.
  Wise}} (Ed.). \bibinfo{publisher}{{ACM}}, \bibinfo{pages}{268--277}.
\newblock
\href{https://doi.org/10.1145/113445.113468}{doi:\nolinkurl{10.1145/113445.113468}}


\bibitem[Giacomo and Vardi(2013)]%
        {DBLP:conf/ijcai/GiacomoV13}
\bibfield{author}{\bibinfo{person}{Giuseppe~De Giacomo} {and}
  \bibinfo{person}{Moshe~Y. Vardi}.} \bibinfo{year}{2013}\natexlab{}.
\newblock \showarticletitle{Linear Temporal Logic and Linear Dynamic Logic on
  Finite Traces}. In \bibinfo{booktitle}{\emph{{IJCAI} 2013, Proceedings of the
  23rd International Joint Conference on Artificial Intelligence, Beijing,
  China, August 3-9, 2013}}, \bibfield{editor}{\bibinfo{person}{Francesca
  Rossi}} (Ed.). \bibinfo{publisher}{{IJCAI/AAAI}}, \bibinfo{pages}{854--860}.
\newblock
\urldef\tempurl%
\url{http://www.aaai.org/ocs/index.php/IJCAI/IJCAI13/paper/view/6997}
\showURL{%
\tempurl}


\bibitem[Gordon(2017)]%
        {DBLP:conf/ecoop/Gordon17}
\bibfield{author}{\bibinfo{person}{Colin~S. Gordon}.}
  \bibinfo{year}{2017}\natexlab{}.
\newblock \showarticletitle{A Generic Approach to Flow-Sensitive Polymorphic
  Effects}. In \bibinfo{booktitle}{\emph{31st European Conference on
  Object-Oriented Programming, {ECOOP} 2017, June 19-23, 2017, Barcelona,
  Spain}} \emph{(\bibinfo{series}{LIPIcs}, Vol.~\bibinfo{volume}{74})},
  \bibfield{editor}{\bibinfo{person}{Peter M{\"{u}}ller}} (Ed.).
  \bibinfo{publisher}{Schloss Dagstuhl - Leibniz-Zentrum f{\"{u}}r Informatik},
  \bibinfo{pages}{13:1--13:31}.
\newblock
\href{https://doi.org/10.4230/LIPICS.ECOOP.2017.13}{doi:\nolinkurl{10.4230/LIPICS.ECOOP.2017.13}}


\bibitem[Gordon(2021)]%
        {DBLP:journals/toplas/Gordon21}
\bibfield{author}{\bibinfo{person}{Colin~S. Gordon}.}
  \bibinfo{year}{2021}\natexlab{}.
\newblock \showarticletitle{Polymorphic Iterable Sequential Effect Systems}.
\newblock \bibinfo{journal}{\emph{{ACM} Trans. Program. Lang. Syst.}}
  \bibinfo{volume}{43}, \bibinfo{number}{1} (\bibinfo{year}{2021}),
  \bibinfo{pages}{4:1--4:79}.
\newblock
\href{https://doi.org/10.1145/3450272}{doi:\nolinkurl{10.1145/3450272}}


\bibitem[Gordon and Yun(2023)]%
        {DBLP:conf/sas/GordonY23}
\bibfield{author}{\bibinfo{person}{Colin~S. Gordon} {and}
  \bibinfo{person}{Chaewon Yun}.} \bibinfo{year}{2023}\natexlab{}.
\newblock \showarticletitle{Error Localization for Sequential Effect Systems}.
  In \bibinfo{booktitle}{\emph{Static Analysis - 30th International Symposium,
  {SAS} 2023, Cascais, Portugal, October 22-24, 2023, Proceedings}}
  \emph{(\bibinfo{series}{Lecture Notes in Computer Science},
  Vol.~\bibinfo{volume}{14284})}, \bibfield{editor}{\bibinfo{person}{Manuel~V.
  Hermenegildo} {and} \bibinfo{person}{Jos{\'{e}}~F. Morales}} (Eds.).
  \bibinfo{publisher}{Springer}, \bibinfo{pages}{343--370}.
\newblock
\href{https://doi.org/10.1007/978-3-031-44245-2\_16}{doi:\nolinkurl{10.1007/978-3-031-44245-2\_16}}


\bibitem[Heizmann et~al\mbox{.}(2013)]%
        {DBLP:conf/cav/HeizmannHP13}
\bibfield{author}{\bibinfo{person}{Matthias Heizmann}, \bibinfo{person}{Jochen
  Hoenicke}, {and} \bibinfo{person}{Andreas Podelski}.}
  \bibinfo{year}{2013}\natexlab{}.
\newblock \showarticletitle{Software Model Checking for People Who Love
  Automata}. In \bibinfo{booktitle}{\emph{Computer Aided Verification - 25th
  International Conference, {CAV} 2013, Saint Petersburg, Russia, July 13-19,
  2013. Proceedings}} \emph{(\bibinfo{series}{Lecture Notes in Computer
  Science}, Vol.~\bibinfo{volume}{8044})},
  \bibfield{editor}{\bibinfo{person}{Natasha Sharygina} {and}
  \bibinfo{person}{Helmut Veith}} (Eds.). \bibinfo{publisher}{Springer},
  \bibinfo{pages}{36--52}.
\newblock
\href{https://doi.org/10.1007/978-3-642-39799-8\_2}{doi:\nolinkurl{10.1007/978-3-642-39799-8\_2}}


\bibitem[Hoffmann et~al\mbox{.}(2011)]%
        {DBLP:conf/popl/HoffmannAH11}
\bibfield{author}{\bibinfo{person}{Jan Hoffmann}, \bibinfo{person}{Klaus
  Aehlig}, {and} \bibinfo{person}{Martin Hofmann}.}
  \bibinfo{year}{2011}\natexlab{}.
\newblock \showarticletitle{Multivariate amortized resource analysis}. In
  \bibinfo{booktitle}{\emph{Proceedings of the 38th {ACM} {SIGPLAN-SIGACT}
  Symposium on Principles of Programming Languages, {POPL} 2011, Austin, TX,
  USA, January 26-28, 2011}}, \bibfield{editor}{\bibinfo{person}{Thomas Ball}
  {and} \bibinfo{person}{Mooly Sagiv}} (Eds.). \bibinfo{publisher}{{ACM}},
  \bibinfo{pages}{357--370}.
\newblock
\href{https://doi.org/10.1145/1926385.1926427}{doi:\nolinkurl{10.1145/1926385.1926427}}


\bibitem[Hoffmann et~al\mbox{.}(2017)]%
        {DBLP:conf/popl/HoffmannDW17}
\bibfield{author}{\bibinfo{person}{Jan Hoffmann}, \bibinfo{person}{Ankush Das},
  {and} \bibinfo{person}{Shu{-}Chun Weng}.} \bibinfo{year}{2017}\natexlab{}.
\newblock \showarticletitle{Towards automatic resource bound analysis for
  OCaml}. In \bibinfo{booktitle}{\emph{Proceedings of the 44th {ACM} {SIGPLAN}
  Symposium on Principles of Programming Languages, {POPL} 2017, Paris, France,
  January 18-20, 2017}}, \bibfield{editor}{\bibinfo{person}{Giuseppe Castagna}
  {and} \bibinfo{person}{Andrew~D. Gordon}} (Eds.). \bibinfo{publisher}{{ACM}},
  \bibinfo{pages}{359--373}.
\newblock
\href{https://doi.org/10.1145/3009837.3009842}{doi:\nolinkurl{10.1145/3009837.3009842}}


\bibitem[Hofmann and Chen(2014)]%
        {DBLP:conf/csl/0001C14}
\bibfield{author}{\bibinfo{person}{Martin Hofmann} {and} \bibinfo{person}{Wei
  Chen}.} \bibinfo{year}{2014}\natexlab{}.
\newblock \showarticletitle{Abstract interpretation from B{\"{u}}chi automata}.
  In \bibinfo{booktitle}{\emph{Joint Meeting of the Twenty-Third {EACSL} Annual
  Conference on Computer Science Logic {(CSL)} and the Twenty-Ninth Annual
  {ACM/IEEE} Symposium on Logic in Computer Science (LICS), {CSL-LICS} '14,
  Vienna, Austria, July 14 - 18, 2014}},
  \bibfield{editor}{\bibinfo{person}{Thomas~A. Henzinger} {and}
  \bibinfo{person}{Dale Miller}} (Eds.). \bibinfo{publisher}{{ACM}},
  \bibinfo{pages}{51:1--51:10}.
\newblock
\href{https://doi.org/10.1145/2603088.2603127}{doi:\nolinkurl{10.1145/2603088.2603127}}


\bibitem[Igarashi and Kobayashi(2005)]%
        {DBLP:journals/toplas/IgarashiK05}
\bibfield{author}{\bibinfo{person}{Atsushi Igarashi} {and}
  \bibinfo{person}{Naoki Kobayashi}.} \bibinfo{year}{2005}\natexlab{}.
\newblock \showarticletitle{Resource usage analysis}.
\newblock \bibinfo{journal}{\emph{{ACM} Trans. Program. Lang. Syst.}}
  \bibinfo{volume}{27}, \bibinfo{number}{2} (\bibinfo{year}{2005}),
  \bibinfo{pages}{264--313}.
\newblock
\href{https://doi.org/10.1145/1057387.1057390}{doi:\nolinkurl{10.1145/1057387.1057390}}


\bibitem[Jeannet and Min{\'{e}}(2009)]%
        {DBLP:conf/cav/JeannetM09}
\bibfield{author}{\bibinfo{person}{Bertrand Jeannet} {and}
  \bibinfo{person}{Antoine Min{\'{e}}}.} \bibinfo{year}{2009}\natexlab{}.
\newblock \showarticletitle{Apron: {A} Library of Numerical Abstract Domains
  for Static Analysis}. In \bibinfo{booktitle}{\emph{Computer Aided
  Verification, 21st International Conference, {CAV} 2009, Grenoble, France,
  June 26 - July 2, 2009. Proceedings}} \emph{(\bibinfo{series}{Lecture Notes
  in Computer Science}, Vol.~\bibinfo{volume}{5643})},
  \bibfield{editor}{\bibinfo{person}{Ahmed Bouajjani} {and}
  \bibinfo{person}{Oded Maler}} (Eds.). \bibinfo{publisher}{Springer},
  \bibinfo{pages}{661--667}.
\newblock
\href{https://doi.org/10.1007/978-3-642-02658-4\_52}{doi:\nolinkurl{10.1007/978-3-642-02658-4\_52}}


\bibitem[Jhala(2025)]%
        {github:liquidhaskell}
\bibfield{author}{\bibinfo{person}{Ranjit Jhala}.}
  \bibinfo{year}{2025}\natexlab{}.
\newblock \bibinfo{title}{LiquidHaskell}.
\newblock
\urldef\tempurl%
\url{https://ucsd-progsys.github.io/liquidhaskell/}
\showURL{%
\tempurl}


\bibitem[Kaminski and Francez(1994)]%
        {DBLP:journals/tcs/KaminskiF94}
\bibfield{author}{\bibinfo{person}{Michael Kaminski} {and}
  \bibinfo{person}{Nissim Francez}.} \bibinfo{year}{1994}\natexlab{}.
\newblock \showarticletitle{Finite-Memory Automata}.
\newblock \bibinfo{journal}{\emph{Theor. Comput. Sci.}} \bibinfo{volume}{134},
  \bibinfo{number}{2} (\bibinfo{year}{1994}), \bibinfo{pages}{329--363}.
\newblock
\href{https://doi.org/10.1016/0304-3975(94)90242-9}{doi:\nolinkurl{10.1016/0304-3975(94)90242-9}}


\bibitem[Katsumata(2014)]%
        {DBLP:conf/popl/Katsumata14}
\bibfield{author}{\bibinfo{person}{Shin{-}ya Katsumata}.}
  \bibinfo{year}{2014}\natexlab{}.
\newblock \showarticletitle{Parametric effect monads and semantics of effect
  systems}. In \bibinfo{booktitle}{\emph{The 41st Annual {ACM} {SIGPLAN-SIGACT}
  Symposium on Principles of Programming Languages, {POPL} '14, San Diego, CA,
  USA, January 20-21, 2014}}, \bibfield{editor}{\bibinfo{person}{Suresh
  Jagannathan} {and} \bibinfo{person}{Peter Sewell}} (Eds.).
  \bibinfo{publisher}{{ACM}}, \bibinfo{pages}{633--646}.
\newblock
\href{https://doi.org/10.1145/2535838.2535846}{doi:\nolinkurl{10.1145/2535838.2535846}}


\bibitem[Katsura et~al\mbox{.}(2020)]%
        {conf/aplas/KatsuraIKT20}
\bibfield{author}{\bibinfo{person}{Hiroyuki Katsura}, \bibinfo{person}{Naoki
  Iwayama}, \bibinfo{person}{Naoki Kobayashi}, {and} \bibinfo{person}{Takeshi
  Tsukada}.} \bibinfo{year}{2020}\natexlab{}.
\newblock \showarticletitle{A New Refinement Type System for Automated {$\nu
  \text{HFL}_\mathbb{Z}$} Validity Checking}. In
  \bibinfo{booktitle}{\emph{Programming Languages and Systems - 18th Asian
  Symposium, {APLAS} 2020, Fukuoka, Japan, November 30 - December 2, 2020,
  Proceedings}} \emph{(\bibinfo{series}{Lecture Notes in Computer Science},
  Vol.~\bibinfo{volume}{12470})}, \bibfield{editor}{\bibinfo{person}{Bruno~C.
  d.~S.~Oliveira}} (Ed.). \bibinfo{publisher}{Springer},
  \bibinfo{pages}{86--104}.
\newblock
\href{https://doi.org/10.1007/978-3-030-64437-6\_5}{doi:\nolinkurl{10.1007/978-3-030-64437-6\_5}}


\bibitem[Katsura et~al\mbox{.}(2025a)]%
        {github:rethfl}
\bibfield{author}{\bibinfo{person}{Hiroyuki Katsura}, \bibinfo{person}{Naoki
  Iwayama}, \bibinfo{person}{Naoki Kobayashi}, {and} \bibinfo{person}{Takeshi
  Tsukada}.} \bibinfo{year}{2025}\natexlab{a}.
\newblock \bibinfo{title}{ReTHFL}.
\newblock
\urldef\tempurl%
\url{https://github.com/hopv/rethfl}
\showURL{%
\tempurl}


\bibitem[Katsura et~al\mbox{.}(2025b)]%
        {web:rethfl}
\bibfield{author}{\bibinfo{person}{Hiroyuki Katsura}, \bibinfo{person}{Naoki
  Iwayama}, \bibinfo{person}{Naoki Kobayashi}, {and} \bibinfo{person}{Takeshi
  Tsukada}.} \bibinfo{year}{2025}\natexlab{b}.
\newblock \bibinfo{title}{Web interface for ReTHFL}.
\newblock
\urldef\tempurl%
\url{https://www-kb.is.s.u-tokyo.ac.jp/~koba/nuhfl/}
\showURL{%
\tempurl}


\bibitem[Kawamata et~al\mbox{.}(2024)]%
        {DBLP:journals/pacmpl/KawamataUST24}
\bibfield{author}{\bibinfo{person}{Fuga Kawamata}, \bibinfo{person}{Hiroshi
  Unno}, \bibinfo{person}{Taro Sekiyama}, {and} \bibinfo{person}{Tachio
  Terauchi}.} \bibinfo{year}{2024}\natexlab{}.
\newblock \showarticletitle{Answer Refinement Modification: Refinement Type
  System for Algebraic Effects and Handlers}.
\newblock \bibinfo{journal}{\emph{Proc. {ACM} Program. Lang.}}
  \bibinfo{volume}{8}, \bibinfo{number}{{POPL}} (\bibinfo{year}{2024}),
  \bibinfo{pages}{115--147}.
\newblock
\href{https://doi.org/10.1145/3633280}{doi:\nolinkurl{10.1145/3633280}}


\bibitem[Knowles and Flanagan(2009)]%
        {DBLP:conf/plpv/KnowlesF09}
\bibfield{author}{\bibinfo{person}{Kenneth~L. Knowles} {and}
  \bibinfo{person}{Cormac Flanagan}.} \bibinfo{year}{2009}\natexlab{}.
\newblock \showarticletitle{Compositional reasoning and decidable checking for
  dependent contract types}. In \bibinfo{booktitle}{\emph{Proceedings of the
  3rd {ACM} Workshop Programming Languages meets Program Verification, {PLPV}
  2009, Savannah, GA, USA, January 20, 2009}},
  \bibfield{editor}{\bibinfo{person}{Thorsten Altenkirch} {and}
  \bibinfo{person}{Todd~D. Millstein}} (Eds.). \bibinfo{publisher}{{ACM}},
  \bibinfo{pages}{27--38}.
\newblock
\href{https://doi.org/10.1145/1481848.1481853}{doi:\nolinkurl{10.1145/1481848.1481853}}


\bibitem[Kobayashi(2009)]%
        {DBLP:conf/popl/Kobayashi09}
\bibfield{author}{\bibinfo{person}{Naoki Kobayashi}.}
  \bibinfo{year}{2009}\natexlab{}.
\newblock \showarticletitle{Types and higher-order recursion schemes for
  verification of higher-order programs}. In
  \bibinfo{booktitle}{\emph{Proceedings of the 36th {ACM} {SIGPLAN-SIGACT}
  Symposium on Principles of Programming Languages, {POPL} 2009, Savannah, GA,
  USA, January 21-23, 2009}}, \bibfield{editor}{\bibinfo{person}{Zhong Shao}
  {and} \bibinfo{person}{Benjamin~C. Pierce}} (Eds.).
  \bibinfo{publisher}{{ACM}}, \bibinfo{pages}{416--428}.
\newblock
\href{https://doi.org/10.1145/1480881.1480933}{doi:\nolinkurl{10.1145/1480881.1480933}}


\bibitem[Kobayashi et~al\mbox{.}(2011)]%
        {DBLP:conf/pldi/KobayashiSU11}
\bibfield{author}{\bibinfo{person}{Naoki Kobayashi}, \bibinfo{person}{Ryosuke
  Sato}, {and} \bibinfo{person}{Hiroshi Unno}.}
  \bibinfo{year}{2011}\natexlab{}.
\newblock \showarticletitle{Predicate abstraction and {CEGAR} for higher-order
  model checking}. In \bibinfo{booktitle}{\emph{Proceedings of the 32nd {ACM}
  {SIGPLAN} Conference on Programming Language Design and Implementation,
  {PLDI} 2011, San Jose, CA, USA, June 4-8, 2011}},
  \bibfield{editor}{\bibinfo{person}{Mary~W. Hall} {and}
  \bibinfo{person}{David~A. Padua}} (Eds.). \bibinfo{publisher}{{ACM}},
  \bibinfo{pages}{222--233}.
\newblock
\href{https://doi.org/10.1145/1993498.1993525}{doi:\nolinkurl{10.1145/1993498.1993525}}


\bibitem[Kobayashi et~al\mbox{.}(2023)]%
        {DBLP:journals/pacmpl/KobayashiTST23}
\bibfield{author}{\bibinfo{person}{Naoki Kobayashi}, \bibinfo{person}{Kento
  Tanahashi}, \bibinfo{person}{Ryosuke Sato}, {and} \bibinfo{person}{Takeshi
  Tsukada}.} \bibinfo{year}{2023}\natexlab{}.
\newblock \showarticletitle{{HFL(Z)} Validity Checking for Automated Program
  Verification}.
\newblock \bibinfo{journal}{\emph{Proc. {ACM} Program. Lang.}}
  \bibinfo{volume}{7}, \bibinfo{number}{{POPL}} (\bibinfo{year}{2023}),
  \bibinfo{pages}{154--184}.
\newblock
\href{https://doi.org/10.1145/3571199}{doi:\nolinkurl{10.1145/3571199}}


\bibitem[Kobayashi et~al\mbox{.}(2018)]%
        {DBLP:conf/esop/0001TW18}
\bibfield{author}{\bibinfo{person}{Naoki Kobayashi}, \bibinfo{person}{Takeshi
  Tsukada}, {and} \bibinfo{person}{Keiichi Watanabe}.}
  \bibinfo{year}{2018}\natexlab{}.
\newblock \showarticletitle{Higher-Order Program Verification via {HFL} Model
  Checking}. In \bibinfo{booktitle}{\emph{Programming Languages and Systems -
  27th European Symposium on Programming, {ESOP} 2018, Held as Part of the
  European Joint Conferences on Theory and Practice of Software, {ETAPS} 2018,
  Thessaloniki, Greece, April 14-20, 2018, Proceedings}}
  \emph{(\bibinfo{series}{Lecture Notes in Computer Science},
  Vol.~\bibinfo{volume}{10801})}, \bibfield{editor}{\bibinfo{person}{Amal
  Ahmed}} (Ed.). \bibinfo{publisher}{Springer}, \bibinfo{pages}{711--738}.
\newblock
\href{https://doi.org/10.1007/978-3-319-89884-1\_25}{doi:\nolinkurl{10.1007/978-3-319-89884-1\_25}}


\bibitem[Koskinen and Terauchi(2014)]%
        {DBLP:conf/csl/KoskinenT14}
\bibfield{author}{\bibinfo{person}{Eric Koskinen} {and} \bibinfo{person}{Tachio
  Terauchi}.} \bibinfo{year}{2014}\natexlab{}.
\newblock \showarticletitle{Local temporal reasoning}. In
  \bibinfo{booktitle}{\emph{Joint Meeting of the Twenty-Third {EACSL} Annual
  Conference on Computer Science Logic {(CSL)} and the Twenty-Ninth Annual
  {ACM/IEEE} Symposium on Logic in Computer Science (LICS), {CSL-LICS} '14,
  Vienna, Austria, July 14 - 18, 2014}},
  \bibfield{editor}{\bibinfo{person}{Thomas~A. Henzinger} {and}
  \bibinfo{person}{Dale Miller}} (Eds.). \bibinfo{publisher}{{ACM}},
  \bibinfo{pages}{59:1--59:10}.
\newblock
\href{https://doi.org/10.1145/2603088.2603138}{doi:\nolinkurl{10.1145/2603088.2603138}}


\bibitem[Kura and Unno(2024)]%
        {hiroshiICFP2024}
\bibfield{author}{\bibinfo{person}{Satoshi Kura} {and} \bibinfo{person}{Hiroshi
  Unno}.} \bibinfo{year}{2024}\natexlab{}.
\newblock \bibinfo{title}{Automated Verification of Higher-Order Probabilistic
  Programs via a Dependent Refinement Type System}.
\newblock
\showeprint[arxiv]{2407.02975}~[cs.LO]
\urldef\tempurl%
\url{https://arxiv.org/abs/2407.02975}
\showURL{%
\tempurl}


\bibitem[Kuwahara et~al\mbox{.}(2014)]%
        {DBLP:conf/esop/KuwaharaTU014}
\bibfield{author}{\bibinfo{person}{Takuya Kuwahara}, \bibinfo{person}{Tachio
  Terauchi}, \bibinfo{person}{Hiroshi Unno}, {and} \bibinfo{person}{Naoki
  Kobayashi}.} \bibinfo{year}{2014}\natexlab{}.
\newblock \showarticletitle{Automatic Termination Verification for Higher-Order
  Functional Programs}. In \bibinfo{booktitle}{\emph{Programming Languages and
  Systems - 23rd European Symposium on Programming, {ESOP} 2014, Held as Part
  of the European Joint Conferences on Theory and Practice of Software, {ETAPS}
  2014, Grenoble, France, April 5-13, 2014, Proceedings}}
  \emph{(\bibinfo{series}{Lecture Notes in Computer Science},
  Vol.~\bibinfo{volume}{8410})}, \bibfield{editor}{\bibinfo{person}{Zhong
  Shao}} (Ed.). \bibinfo{publisher}{Springer}, \bibinfo{pages}{392--411}.
\newblock
\href{https://doi.org/10.1007/978-3-642-54833-8\_21}{doi:\nolinkurl{10.1007/978-3-642-54833-8\_21}}


\bibitem[Lago and Ghyselen(2024)]%
        {DBLP:journals/pacmpl/LagoG24}
\bibfield{author}{\bibinfo{person}{Ugo~Dal Lago} {and} \bibinfo{person}{Alexis
  Ghyselen}.} \bibinfo{year}{2024}\natexlab{}.
\newblock \showarticletitle{On Model-Checking Higher-Order Effectful Programs}.
\newblock \bibinfo{journal}{\emph{Proc. {ACM} Program. Lang.}}
  \bibinfo{volume}{8}, \bibinfo{number}{{POPL}} (\bibinfo{year}{2024}),
  \bibinfo{pages}{2610--2638}.
\newblock
\href{https://doi.org/10.1145/3632929}{doi:\nolinkurl{10.1145/3632929}}


\bibitem[Lucassen and Gifford(1988)]%
        {DBLP:conf/popl/LucassenG88}
\bibfield{author}{\bibinfo{person}{John~M. Lucassen} {and}
  \bibinfo{person}{David~K. Gifford}.} \bibinfo{year}{1988}\natexlab{}.
\newblock \showarticletitle{Polymorphic Effect Systems}. In
  \bibinfo{booktitle}{\emph{Conference Record of the Fifteenth Annual {ACM}
  Symposium on Principles of Programming Languages, San Diego, California, USA,
  January 10-13, 1988}}, \bibfield{editor}{\bibinfo{person}{Jeanne Ferrante}
  {and} \bibinfo{person}{Peter Mager}} (Eds.). \bibinfo{publisher}{{ACM}
  Press}, \bibinfo{pages}{47--57}.
\newblock
\href{https://doi.org/10.1145/73560.73564}{doi:\nolinkurl{10.1145/73560.73564}}


\bibitem[Mauborgne and Rival(2005)]%
        {DBLP:conf/esop/MauborgneR05}
\bibfield{author}{\bibinfo{person}{Laurent Mauborgne} {and}
  \bibinfo{person}{Xavier Rival}.} \bibinfo{year}{2005}\natexlab{}.
\newblock \showarticletitle{Trace Partitioning in Abstract Interpretation Based
  Static Analyzers}. In \bibinfo{booktitle}{\emph{Programming Languages and
  Systems, 14th European Symposium on Programming, {ESOP} 2005, Held as Part of
  the Joint European Conferences on Theory and Practice of Software, {ETAPS}
  2005, Edinburgh, UK, April 4-8, 2005, Proceedings}}
  \emph{(\bibinfo{series}{Lecture Notes in Computer Science},
  Vol.~\bibinfo{volume}{3444})}, \bibfield{editor}{\bibinfo{person}{Shmuel
  Sagiv}} (Ed.). \bibinfo{publisher}{Springer}, \bibinfo{pages}{5--20}.
\newblock
\href{https://doi.org/10.1007/978-3-540-31987-0\_2}{doi:\nolinkurl{10.1007/978-3-540-31987-0\_2}}


\bibitem[Min{\'{e}}(2007)]%
        {DBLP:journals/corr/abs-cs-0703084}
\bibfield{author}{\bibinfo{person}{Antoine Min{\'{e}}}.}
  \bibinfo{year}{2007}\natexlab{}.
\newblock \showarticletitle{The Octagon Abstract Domain}.
\newblock \bibinfo{journal}{\emph{CoRR}}  \bibinfo{volume}{abs/cs/0703084}
  (\bibinfo{year}{2007}).
\newblock
\showeprint[arXiv]{cs/0703084}
\urldef\tempurl%
\url{http://arxiv.org/abs/cs/0703084}
\showURL{%
\tempurl}


\bibitem[Murase et~al\mbox{.}(2016)]%
        {DBLP:conf/popl/MuraseT0SU16}
\bibfield{author}{\bibinfo{person}{Akihiro Murase}, \bibinfo{person}{Tachio
  Terauchi}, \bibinfo{person}{Naoki Kobayashi}, \bibinfo{person}{Ryosuke Sato},
  {and} \bibinfo{person}{Hiroshi Unno}.} \bibinfo{year}{2016}\natexlab{}.
\newblock \showarticletitle{Temporal verification of higher-order functional
  programs}. In \bibinfo{booktitle}{\emph{Proceedings of the 43rd Annual {ACM}
  {SIGPLAN-SIGACT} Symposium on Principles of Programming Languages, {POPL}
  2016, St. Petersburg, FL, USA, January 20 - 22, 2016}},
  \bibfield{editor}{\bibinfo{person}{Rastislav Bod{\'{\i}}k} {and}
  \bibinfo{person}{Rupak Majumdar}} (Eds.). \bibinfo{publisher}{{ACM}},
  \bibinfo{pages}{57--68}.
\newblock
\href{https://doi.org/10.1145/2837614.2837667}{doi:\nolinkurl{10.1145/2837614.2837667}}


\bibitem[Nanjo et~al\mbox{.}(2018)]%
        {DBLP:conf/lics/Nanjo0KT18}
\bibfield{author}{\bibinfo{person}{Yoji Nanjo}, \bibinfo{person}{Hiroshi Unno},
  \bibinfo{person}{Eric Koskinen}, {and} \bibinfo{person}{Tachio Terauchi}.}
  \bibinfo{year}{2018}\natexlab{}.
\newblock \showarticletitle{A Fixpoint Logic and Dependent Effects for Temporal
  Property Verification}. In \bibinfo{booktitle}{\emph{Proceedings of the 33rd
  Annual {ACM/IEEE} Symposium on Logic in Computer Science, {LICS} 2018,
  Oxford, UK, July 09-12, 2018}}, \bibfield{editor}{\bibinfo{person}{Anuj
  Dawar} {and} \bibinfo{person}{Erich Gr{\"{a}}del}} (Eds.).
  \bibinfo{publisher}{{ACM}}, \bibinfo{pages}{759--768}.
\newblock
\href{https://doi.org/10.1145/3209108.3209204}{doi:\nolinkurl{10.1145/3209108.3209204}}


\bibitem[Nguyen et~al\mbox{.}(2018)]%
        {DBLP:journals/pacmpl/NguyenGTH18}
\bibfield{author}{\bibinfo{person}{Phuc~C. Nguyen}, \bibinfo{person}{Thomas
  Gilray}, \bibinfo{person}{Sam Tobin{-}Hochstadt}, {and}
  \bibinfo{person}{David~Van Horn}.} \bibinfo{year}{2018}\natexlab{}.
\newblock \showarticletitle{Soft contract verification for higher-order
  stateful programs}.
\newblock \bibinfo{journal}{\emph{Proc. {ACM} Program. Lang.}}
  \bibinfo{volume}{2}, \bibinfo{number}{{POPL}} (\bibinfo{year}{2018}),
  \bibinfo{pages}{51:1--51:30}.
\newblock
\href{https://doi.org/10.1145/3158139}{doi:\nolinkurl{10.1145/3158139}}


\bibitem[Nicola et~al\mbox{.}(2025a)]%
        {Nicola2025}
\bibfield{author}{\bibinfo{person}{Mihai Nicola}, \bibinfo{person}{Chaitanya
  Agarwal}, \bibinfo{person}{Eric Koskinen}, {and} \bibinfo{person}{Thomas
  Wies}.} \bibinfo{year}{2025}\natexlab{a}.
\newblock \showarticletitle{Abstract Interpretation of Temporal Safety Effects
  of Higher Order Programs [Extended Version]}.
\newblock  (\bibinfo{year}{2025}).
\newblock
\showeprint[arxiv]{2408.02791}~[cs.PL]


\bibitem[Nicola et~al\mbox{.}(2025b)]%
        {zenodo:evdrift}
\bibfield{author}{\bibinfo{person}{Mihai Nicola}, \bibinfo{person}{Chaitanya
  Agarwal}, \bibinfo{person}{Eric Koskinen}, {and} \bibinfo{person}{Thomas
  Wies}.} \bibinfo{year}{2025}\natexlab{b}.
\newblock \bibinfo{title}{The evDrift Verifier}.
\newblock \bibinfo{howpublished}{Zenodo}.
\newblock
\href{https://doi.org/10.5281/zenodo.16602546}{doi:\nolinkurl{10.5281/zenodo.16602546}}


\bibitem[Nielsen(2001)]%
        {DBLP:journals/entcs/Nielsen01}
\bibfield{author}{\bibinfo{person}{Lasse~R. Nielsen}.}
  \bibinfo{year}{2001}\natexlab{}.
\newblock \showarticletitle{A Selective {CPS} Transformation}. In
  \bibinfo{booktitle}{\emph{Seventeenth Conference on the Mathematical
  Foundations of Programming Semantics, {MFPS} 2001, Aarhus, Denmark, May
  23-26, 2001}} \emph{(\bibinfo{series}{Electronic Notes in Theoretical
  Computer Science}, Vol.~\bibinfo{volume}{45})},
  \bibfield{editor}{\bibinfo{person}{Stephen~D. Brookes} {and}
  \bibinfo{person}{Michael~W. Mislove}} (Eds.). \bibinfo{publisher}{Elsevier},
  \bibinfo{pages}{311--331}.
\newblock
\href{https://doi.org/10.1016/S1571-0661(04)80969-1}{doi:\nolinkurl{10.1016/S1571-0661(04)80969-1}}


\bibitem[Pavlinovic et~al\mbox{.}(2021)]%
        {DBLP:journals/pacmpl/PavlinovicSW21}
\bibfield{author}{\bibinfo{person}{Zvonimir Pavlinovic}, \bibinfo{person}{Yusen
  Su}, {and} \bibinfo{person}{Thomas Wies}.} \bibinfo{year}{2021}\natexlab{}.
\newblock \showarticletitle{Data flow refinement type inference}.
\newblock \bibinfo{journal}{\emph{Proc. {ACM} Program. Lang.}}
  \bibinfo{volume}{5}, \bibinfo{number}{{POPL}} (\bibinfo{year}{2021}),
  \bibinfo{pages}{1--31}.
\newblock
\href{https://doi.org/10.1145/3434300}{doi:\nolinkurl{10.1145/3434300}}


\bibitem[Pnueli(1977)]%
        {DBLP:conf/focs/Pnueli77}
\bibfield{author}{\bibinfo{person}{Amir Pnueli}.}
  \bibinfo{year}{1977}\natexlab{}.
\newblock \showarticletitle{The Temporal Logic of Programs}. In
  \bibinfo{booktitle}{\emph{18th Annual Symposium on Foundations of Computer
  Science, Providence, Rhode Island, USA, 31 October - 1 November 1977}}.
  \bibinfo{publisher}{{IEEE} Computer Society}, \bibinfo{pages}{46--57}.
\newblock
\href{https://doi.org/10.1109/SFCS.1977.32}{doi:\nolinkurl{10.1109/SFCS.1977.32}}


\bibitem[Rival and Mauborgne(2007)]%
        {DBLP:journals/toplas/RivalM07}
\bibfield{author}{\bibinfo{person}{Xavier Rival} {and} \bibinfo{person}{Laurent
  Mauborgne}.} \bibinfo{year}{2007}\natexlab{}.
\newblock \showarticletitle{The trace partitioning abstract domain}.
\newblock \bibinfo{journal}{\emph{{ACM} Trans. Program. Lang. Syst.}}
  \bibinfo{volume}{29}, \bibinfo{number}{5} (\bibinfo{year}{2007}),
  \bibinfo{pages}{26}.
\newblock
\href{https://doi.org/10.1145/1275497.1275501}{doi:\nolinkurl{10.1145/1275497.1275501}}


\bibitem[Rondon et~al\mbox{.}(2008)]%
        {DBLP:conf/pldi/RondonKJ08}
\bibfield{author}{\bibinfo{person}{Patrick~Maxim Rondon}, \bibinfo{person}{Ming
  Kawaguchi}, {and} \bibinfo{person}{Ranjit Jhala}.}
  \bibinfo{year}{2008}\natexlab{}.
\newblock \showarticletitle{Liquid types}. In
  \bibinfo{booktitle}{\emph{Proceedings of the {ACM} {SIGPLAN} 2008 Conference
  on Programming Language Design and Implementation, Tucson, AZ, USA, June
  7-13, 2008}}, \bibfield{editor}{\bibinfo{person}{Rajiv Gupta} {and}
  \bibinfo{person}{Saman~P. Amarasinghe}} (Eds.). \bibinfo{publisher}{{ACM}},
  \bibinfo{pages}{159--169}.
\newblock
\href{https://doi.org/10.1145/1375581.1375602}{doi:\nolinkurl{10.1145/1375581.1375602}}


\bibitem[Sato(2025)]%
        {github:mochi}
\bibfield{author}{\bibinfo{person}{Ryosuke Sato}.}
  \bibinfo{year}{2025}\natexlab{}.
\newblock \bibinfo{title}{MoCHi}.
\newblock
\urldef\tempurl%
\url{https://github.com/hopv/mochi}
\showURL{%
\tempurl}


\bibitem[Sekiyama and Unno(2023)]%
        {DBLP:journals/pacmpl/SekiyamaU23}
\bibfield{author}{\bibinfo{person}{Taro Sekiyama} {and}
  \bibinfo{person}{Hiroshi Unno}.} \bibinfo{year}{2023}\natexlab{}.
\newblock \showarticletitle{Temporal Verification with Answer-Effect
  Modification: Dependent Temporal Type-and-Effect System with Delimited
  Continuations}.
\newblock \bibinfo{journal}{\emph{Proc. {ACM} Program. Lang.}}
  \bibinfo{volume}{7}, \bibinfo{number}{{POPL}} (\bibinfo{year}{2023}),
  \bibinfo{pages}{2079--2110}.
\newblock
\href{https://doi.org/10.1145/3571264}{doi:\nolinkurl{10.1145/3571264}}


\bibitem[Singh et~al\mbox{.}(2017)]%
        {DBLP:conf/popl/SinghPV17}
\bibfield{author}{\bibinfo{person}{Gagandeep Singh}, \bibinfo{person}{Markus
  P{\"{u}}schel}, {and} \bibinfo{person}{Martin~T. Vechev}.}
  \bibinfo{year}{2017}\natexlab{}.
\newblock \showarticletitle{Fast polyhedra abstract domain}. In
  \bibinfo{booktitle}{\emph{Proceedings of the 44th {ACM} {SIGPLAN} Symposium
  on Principles of Programming Languages, {POPL} 2017, Paris, France, January
  18-20, 2017}}, \bibfield{editor}{\bibinfo{person}{Giuseppe Castagna} {and}
  \bibinfo{person}{Andrew~D. Gordon}} (Eds.). \bibinfo{publisher}{{ACM}},
  \bibinfo{pages}{46--59}.
\newblock
\href{https://doi.org/10.1145/3009837.3009885}{doi:\nolinkurl{10.1145/3009837.3009885}}


\bibitem[Skalka and Smith(2004)]%
        {DBLP:conf/aplas/SkalkaS04}
\bibfield{author}{\bibinfo{person}{Christian Skalka} {and}
  \bibinfo{person}{Scott~F. Smith}.} \bibinfo{year}{2004}\natexlab{}.
\newblock \showarticletitle{History Effects and Verification}. In
  \bibinfo{booktitle}{\emph{Programming Languages and Systems: Second Asian
  Symposium, {APLAS} 2004, Taipei, Taiwan, November 4-6, 2004. Proceedings}}
  \emph{(\bibinfo{series}{Lecture Notes in Computer Science},
  Vol.~\bibinfo{volume}{3302})}, \bibfield{editor}{\bibinfo{person}{Wei{-}Ngan
  Chin}} (Ed.). \bibinfo{publisher}{Springer}, \bibinfo{pages}{107--128}.
\newblock
\href{https://doi.org/10.1007/978-3-540-30477-7\_8}{doi:\nolinkurl{10.1007/978-3-540-30477-7\_8}}


\bibitem[Skalka et~al\mbox{.}(2008)]%
        {DBLP:journals/jfp/SkalkaSH08}
\bibfield{author}{\bibinfo{person}{Christian Skalka}, \bibinfo{person}{Scott~F.
  Smith}, {and} \bibinfo{person}{David~Van Horn}.}
  \bibinfo{year}{2008}\natexlab{}.
\newblock \showarticletitle{Types and trace effects of higher order programs}.
\newblock \bibinfo{journal}{\emph{J. Funct. Program.}} \bibinfo{volume}{18},
  \bibinfo{number}{2} (\bibinfo{year}{2008}), \bibinfo{pages}{179--249}.
\newblock
\href{https://doi.org/10.1017/S0956796807006466}{doi:\nolinkurl{10.1017/S0956796807006466}}


\bibitem[sosy lab(2025)]%
        {github:benchexec}
\bibfield{author}{\bibinfo{person}{sosy lab}.} \bibinfo{year}{2025}\natexlab{}.
\newblock \bibinfo{title}{benchexec}.
\newblock
\urldef\tempurl%
\url{https://github.com/sosy-lab/benchexec}
\showURL{%
\tempurl}


\bibitem[Stephens et~al\mbox{.}(2021)]%
        {DBLP:conf/sp/StephensFMLD21}
\bibfield{author}{\bibinfo{person}{Jon Stephens}, \bibinfo{person}{Kostas
  Ferles}, \bibinfo{person}{Benjamin Mariano}, \bibinfo{person}{Shuvendu~K.
  Lahiri}, {and} \bibinfo{person}{Isil Dillig}.}
  \bibinfo{year}{2021}\natexlab{}.
\newblock \showarticletitle{SmartPulse: Automated Checking of Temporal
  Properties in Smart Contracts}. In \bibinfo{booktitle}{\emph{42nd {IEEE}
  Symposium on Security and Privacy, {SP} 2021, San Francisco, CA, USA, 24-27
  May 2021}}. \bibinfo{publisher}{{IEEE}}, \bibinfo{pages}{555--571}.
\newblock
\href{https://doi.org/10.1109/SP40001.2021.00085}{doi:\nolinkurl{10.1109/SP40001.2021.00085}}


\bibitem[Swamy et~al\mbox{.}(2013)]%
        {DBLP:conf/pldi/SwamyWSCL13}
\bibfield{author}{\bibinfo{person}{Nikhil Swamy}, \bibinfo{person}{Joel
  Weinberger}, \bibinfo{person}{Cole Schlesinger}, \bibinfo{person}{Juan Chen},
  {and} \bibinfo{person}{Benjamin Livshits}.} \bibinfo{year}{2013}\natexlab{}.
\newblock \showarticletitle{Verifying higher-order programs with the dijkstra
  monad}. In \bibinfo{booktitle}{\emph{{ACM} {SIGPLAN} Conference on
  Programming Language Design and Implementation, {PLDI} '13, Seattle, WA, USA,
  June 16-19, 2013}}, \bibfield{editor}{\bibinfo{person}{Hans{-}Juergen Boehm}
  {and} \bibinfo{person}{Cormac Flanagan}} (Eds.). \bibinfo{publisher}{{ACM}},
  \bibinfo{pages}{387--398}.
\newblock
\href{https://doi.org/10.1145/2491956.2491978}{doi:\nolinkurl{10.1145/2491956.2491978}}


\bibitem[Tate(2013)]%
        {DBLP:conf/popl/Tate13}
\bibfield{author}{\bibinfo{person}{Ross Tate}.}
  \bibinfo{year}{2013}\natexlab{}.
\newblock \showarticletitle{The sequential semantics of producer effect
  systems}. In \bibinfo{booktitle}{\emph{The 40th Annual {ACM} {SIGPLAN-SIGACT}
  Symposium on Principles of Programming Languages, {POPL} '13, Rome, Italy -
  January 23 - 25, 2013}}, \bibfield{editor}{\bibinfo{person}{Roberto
  Giacobazzi} {and} \bibinfo{person}{Radhia Cousot}} (Eds.).
  \bibinfo{publisher}{{ACM}}, \bibinfo{pages}{15--26}.
\newblock
\href{https://doi.org/10.1145/2429069.2429074}{doi:\nolinkurl{10.1145/2429069.2429074}}


\bibitem[Terauchi(2010)]%
        {DBLP:conf/popl/Terauchi10}
\bibfield{author}{\bibinfo{person}{Tachio Terauchi}.}
  \bibinfo{year}{2010}\natexlab{}.
\newblock \showarticletitle{Dependent types from counterexamples}. In
  \bibinfo{booktitle}{\emph{Proceedings of the 37th {ACM} {SIGPLAN-SIGACT}
  Symposium on Principles of Programming Languages, {POPL} 2010, Madrid, Spain,
  January 17-23, 2010}}, \bibfield{editor}{\bibinfo{person}{Manuel~V.
  Hermenegildo} {and} \bibinfo{person}{Jens Palsberg}} (Eds.).
  \bibinfo{publisher}{{ACM}}, \bibinfo{pages}{119--130}.
\newblock
\href{https://doi.org/10.1145/1706299.1706315}{doi:\nolinkurl{10.1145/1706299.1706315}}


\bibitem[Unno(2025)]%
        {github:coar}
\bibfield{author}{\bibinfo{person}{Hiroshi Unno}.}
  \bibinfo{year}{2025}\natexlab{}.
\newblock \bibinfo{title}{CoAR}.
\newblock
\urldef\tempurl%
\url{https://github.com/hiroshi-unno/coar/}
\showURL{%
\tempurl}


\bibitem[Vardi(1987)]%
        {DBLP:conf/lics/Vardi87}
\bibfield{author}{\bibinfo{person}{Moshe~Y. Vardi}.}
  \bibinfo{year}{1987}\natexlab{}.
\newblock \showarticletitle{Verification of Concurrent Programs: The
  Automata-Theoretic Framework}. In \bibinfo{booktitle}{\emph{Proceedings of
  the Symposium on Logic in Computer Science {(LICS} '87), Ithaca, New York,
  USA, June 22-25, 1987}}. \bibinfo{publisher}{{IEEE} Computer Society},
  \bibinfo{pages}{167--176}.
\newblock


\bibitem[Vazou et~al\mbox{.}(2014)]%
        {DBLP:conf/haskell/VazouSJ14}
\bibfield{author}{\bibinfo{person}{Niki Vazou}, \bibinfo{person}{Eric~L.
  Seidel}, {and} \bibinfo{person}{Ranjit Jhala}.}
  \bibinfo{year}{2014}\natexlab{}.
\newblock \showarticletitle{LiquidHaskell: experience with refinement types in
  the real world}. In \bibinfo{booktitle}{\emph{Proceedings of the 2014 {ACM}
  {SIGPLAN} symposium on Haskell, Gothenburg, Sweden, September 4-5, 2014}},
  \bibfield{editor}{\bibinfo{person}{Wouter Swierstra}} (Ed.).
  \bibinfo{publisher}{{ACM}}, \bibinfo{pages}{39--51}.
\newblock
\href{https://doi.org/10.1145/2633357.2633366}{doi:\nolinkurl{10.1145/2633357.2633366}}


\bibitem[Xi and Pfenning(1999)]%
        {DBLP:conf/popl/XiP99}
\bibfield{author}{\bibinfo{person}{Hongwei Xi} {and} \bibinfo{person}{Frank
  Pfenning}.} \bibinfo{year}{1999}\natexlab{}.
\newblock \showarticletitle{Dependent Types in Practical Programming}. In
  \bibinfo{booktitle}{\emph{{POPL} '99, Proceedings of the 26th {ACM}
  {SIGPLAN-SIGACT} Symposium on Principles of Programming Languages, San
  Antonio, TX, USA, January 20-22, 1999}},
  \bibfield{editor}{\bibinfo{person}{Andrew~W. Appel} {and}
  \bibinfo{person}{Alex Aiken}} (Eds.). \bibinfo{publisher}{{ACM}},
  \bibinfo{pages}{214--227}.
\newblock
\href{https://doi.org/10.1145/292540.292560}{doi:\nolinkurl{10.1145/292540.292560}}


\bibitem[Yamada et~al\mbox{.}(2025)]%
        {DBLP:conf/esop/YamadaKSS25}
\bibfield{author}{\bibinfo{person}{Risa Yamada}, \bibinfo{person}{Naoki
  Kobayashi}, \bibinfo{person}{Ken Sakayori}, {and} \bibinfo{person}{Ryosuke
  Sato}.} \bibinfo{year}{2025}\natexlab{}.
\newblock \showarticletitle{On the Relationship between Dijkstra Monads and
  Higher-Order Fixpoint Logic}. In \bibinfo{booktitle}{\emph{Programming
  Languages and Systems - 34th European Symposium on Programming, {ESOP} 2025,
  Held as Part of the International Joint Conferences on Theory and Practice of
  Software, {ETAPS} 2025, Hamilton, ON, Canada, May 3-8, 2025, Proceedings,
  Part {II}}} \emph{(\bibinfo{series}{Lecture Notes in Computer Science},
  Vol.~\bibinfo{volume}{15695})}, \bibfield{editor}{\bibinfo{person}{Viktor
  Vafeiadis}} (Ed.). \bibinfo{publisher}{Springer}, \bibinfo{pages}{402--428}.
\newblock
\href{https://doi.org/10.1007/978-3-031-91121-7\_16}{doi:\nolinkurl{10.1007/978-3-031-91121-7\_16}}


\bibitem[Zhou et~al\mbox{.}(2024)]%
        {DBLP:journals/pacmpl/ZhouYDJ24}
\bibfield{author}{\bibinfo{person}{Zhe Zhou}, \bibinfo{person}{Qianchuan Ye},
  \bibinfo{person}{Benjamin Delaware}, {and} \bibinfo{person}{Suresh
  Jagannathan}.} \bibinfo{year}{2024}\natexlab{}.
\newblock \showarticletitle{A {HAT} Trick: Automatically Verifying
  Representation Invariants using Symbolic Finite Automata}.
\newblock \bibinfo{journal}{\emph{Proc. {ACM} Program. Lang.}}
  \bibinfo{volume}{8}, \bibinfo{number}{{PLDI}} (\bibinfo{year}{2024}),
  \bibinfo{pages}{1387--1411}.
\newblock
\href{https://doi.org/10.1145/3656433}{doi:\nolinkurl{10.1145/3656433}}


\end{thebibliography}

\ifdefined\extdversion

\appendix
\vfill
\pagebreak

\section{Soundness of Accumulative Dependent Type and Effect System}
\label{sec:soundness}

In this section we prove \cref{thm:soundness}. We start from an instantiation of our type system that satisfies the assumption specified in \cref{sec:types-and-effects}. That is, we assume a lattice of base refinement types
$\langle \bdom, \bord, \bbot, \btop, \bjoin, \bmeet \rangle$ and a lattice of effects $\langle \effdom, \efford, \effjoin, \effmeet, \effbot, \efftop \rangle $ with their concretization functions and abstract domain operations.

As outlined in \cref{sec:types-and-effects}, we fist establish a connection between abstract typing derivations and derivations in a \emph{concretized} version of the type system.

We obtain this concretized version by instantiating the abstract domain of base refinement types with the powerset lattice $\powerset(\cvalues \times \cenvs)$ and the concretization function as identity. We likewise instantiate the abstract domain of dependent effects with the powerset lattice $\powerset(\ceffs \times \cenvs)$. All abstract domain operations are defined by the most precise operator that satisfies the respective soundness condition. For example, we have for all $\bval \subseteq \cvalues \times \cenvs$ and concrete typing environments $\tenv$:
\[\tstrengthen{\bval}{\tenv} \;\Def=\; \bval \cap (\cvalues \times \tgamma(\tenv))\enspace.\]
To distinguish the operators $\{\nu = c\}$ for constructing base refinement types from constant values $c$ for the two versions of the type system, we annotate it with their respective domain, writing $\{\nu = c\}_\bdom$ for the abstract version and $\{\nu = c\}_{\powerset(\cvalues \times \cenvs)}$ for the concrete one.

For a type $\tval$ in the abstract type system, we define its concretization $\tval^\gamma$ recursively as follows:
\begin{align*}
\bval^\gamma =\; & \bgamma(\bval)\\
(\tfun{x}{\tval_1}{\eff_1}{\tval_2}{\eff_2})^\gamma =\; &
\tfun{x}{\tval_1^\gamma}{\effgamma(\eff_1)}{\tval_2^\gamma}{\effgamma(\eff_2)}\\
(\texists{x}{\tval_1}{\tval_2})^\gamma =\; & \texists{x}{\tval_1^\gamma}{\tval_2^\gamma}\enspace.
\end{align*}
For an abstract typing environment $\tenv$ we denote by $\tenv^\gamma$ the concrete typing environment obtained from $\tenv$ by applying the above concretization pointwise to each binding in $\tenv$. For consistency of notation, we use the short-hand $\eff^\gamma$ to denote the concrete effect $\effgamma(\eff)$ associated with an abstract effect $\eff \in \effdom$.

The following three lemmas establish that (sub)typing derivations in the abstract instantiation of the type system can be replayed in the concrete one.

\begin{lemma}
  If $\wf(\tenv)$, then $\wf(\tenv^\gamma)$.
\end{lemma}

\begin{proof}
  Straightforward by induction on the length of $\tenv$ and the structure of the types bound in $\tenv$.
\end{proof}

\begin{lemma}
  \label{lem:abstract-subtyping-gives-concrete-subtyping}
  If $\subtjudge{\tenv}{\tval}{\tval'}$, then $\subtjudge{\tenv^\gamma}{\tval^\gamma}{\tval'^\gamma}$.
\end{lemma}

\begin{proof}
  Straightforward by induction on the derivation of $\subtjudge{\tenv}{\tval}{\tval'}$ and case analysis on the last rule applied in the derivation. For the base case of rule \ruleref{s-base} we use the soundness condition on $\tstrengthen{\bval}{\tenv}$ to establish $\tstrengthen{\bval^\gamma}{\tenv^\gamma} \subseteq \bgamma(\tstrengthen{\bval}{\tenv})$.
\end{proof}

\begin{lemma}
  \label{lem:abstract-typing-gives-concrete-typing}
  If $\tjudge{\tenv;\eff}{\term}{\tval}{\eff'}$, then $\tjudge{\tenv^\gamma;\eff^\gamma}{\term}{\tval^\gamma}{\eff'^\gamma}$.    
\end{lemma}

\begin{proof}
  The proof goes by induction on the derivation of $\tjudge{\tenv;\eff}{\term}{\tval}{\eff'}$. We do case analysis on the last typing rule that has been applied in the derivation.

  \begin{description}
  \item[Case \ruleref{t-var}] We have $\term = x$ for some $x$ such that $\tenv(x)=\tval$. We must also have $\eff=\eff'$. From the definition of $\tenv^\gamma$ it follows that $\tenv^\gamma(x)=\tval^\gamma$ and, hence, $\tjudge{\tenv^\gamma;\eff^\gamma}{x}{\tval^\gamma}{\eff^\gamma}$ using rule \ruleref{t-var}.
  \item[Case \ruleref{t-const}] We have $\term = c$ for some $c$ such that $\tval=\{\nu = c\}_\bdom$. Moreover, we must have $\eff'=\eff$. Using rule \ruleref{t-const} we infer $\tjudge{\tenv^\gamma;\eff^\gamma}{x}{\{\nu = x\}_{\powerset(\cvalues \times \cenvs)}}{\eff^\gamma}$. By assumption on $\{\nu = x\}_{\bdom}$ we have $\tval^\gamma = \bgamma(\{\nu = x\}_{\bdom}) \supseteq \{c\} \times \cenvs = \{\nu = x\}_{\powerset(\cvalues \times \cenvs)}$. It follows that $\tstrengthen{\{\nu = x\}_{\powerset(\cvalues \times \cenvs)}}{\tenv^\gamma} \subseteq \tval^\gamma$. As $\tstrengthen{\eff^\gamma}{\tenv} \subseteq \eff^\gamma$ holds trivially,  we can use rule \ruleref{t-weaken} to derive $\tjudge{\tenv^\gamma;\eff^\gamma}{x}{\tval^\gamma}{\eff^\gamma}$.
  \item[Case \ruleref{t-ev}] We have $e = \termkw{ev}~\term_1$ for some $\term_1$ and $\tval=\{\nu = \unitval\}_\bdom$. Moreover, there exists $\eff_1$ and $\bval_1$ such that $\tjudge{\tenv;\eff}{\term_1}{\bval_1}{\eff_1}$ and $\eff' = \eff_1 \odot \bval_1$. By induction hypothesis we have $\tjudge{\tenv^\gamma;\eff^\gamma}{\term_1}{\bval_1^\gamma}{\eff_1^\gamma}$. Thus, using rule \ruleref{t-ev} we can derive $\tjudge{\tenv^\gamma;\eff^\gamma}{\termkw{ev}~\term_1}{\{\nu = \unitval\}_{\powerset(\cvalues \times \cenvs)}}{(\eff^\gamma \odot \bval_1^\gamma)}$.
    We have by assumption that $\tval^\gamma = \bgamma(\{\nu = \unitval\}_\bdom) \supseteq \{\unitval\} \times \cenvs = \{\nu = \unitval\}_{\powerset(\cvalues \times \cenvs)}$. Moreover, we trivially have $\tstrengthen{\eff^\gamma}{\tenv^\gamma} \subseteq \eff^\gamma$. Finally, by assumption on $\odot$ we have $\eff'^\gamma = \bgamma(\eff_1 \odot \bval_1) \supseteq \effgamma(\eff_1) \odot \bgamma(\bval_1) = \eff_1^\gamma \odot \bval_1^\gamma$. It follows that $\tstrengthen{(\eff_1^\gamma \odot \bval_1^\gamma)}{\tenv^\gamma} \subseteq \eff'^\gamma$. Using rule \ruleref{t-weaken} we can thus derive $\tjudge{\tenv^\gamma;\eff^\gamma}{\term}{\tval^\gamma}{\eff'^\gamma}$.
  \item[Case \ruleref{t-app}]
    We have $\term = \term_1~\term_2$ for some $\term_1$ and $\term_2$.
    Moreover, there exist $\effp$, $\eff_1$, and $\eff_2$ as well as $\tval'$, $\tval_1$ and $\tval_2$ such that $(\tandeff{\tval}{\eff'}) = \texists{x}{\tval_2}{(\tandeff{\tval'}{\effp})}$,
    $\tjudge{\tenv;\eff}{e_1}{\tval_1}{\eff_1}$, $\tjudge{\tenv;\eff_1}{e_2}{\tval_2}{\eff_2}$, and $\tval_1 = \tfun{x}{\tval_2}{\eff_2}{\tval'}{\effp}$. By induction hypothesis we have $\tjudge{\tenv^\gamma;\eff^\gamma}{e_1}{\tval_1^\gamma}{\eff_1^\gamma}$ and $\tjudge{\tenv^\gamma;\eff_1^\gamma}{e_2}{\tval_2^\gamma}{\eff_2^\gamma}$. Using rule \ruleref{t-app} we conclude $\tjudgesimp{\tenv^\gamma;\eff^\gamma}{\term}{\tandeff{\tval^\gamma}{(\texists{x}{\tval_2^\gamma}{\effp^\gamma})}}$. We trivially have $\subtjudge{\tenv^\gamma}{\tval^\gamma}{\tval^\gamma}$ and $\tstrengthen{\eff^\gamma}{\tenv^\gamma} \subseteq \eff^\gamma$. Moreover, by soundness of strengthening we have
      $(\texists{x}{\tval_2^\gamma}{\effp^\gamma}) \subseteq \effgamma(\texists{x}{\tval_2}{\effp}) = \eff'^\gamma$. It follows that $\tstrengthen{(\texists{x}{\tval_2^\gamma}{\effp^\gamma})}{\tenv^\gamma} \subseteq \eff'^\gamma$. Hence, using rule \ruleref{t-weaken} we derive $\tjudge{\tenv^\gamma;\eff^\gamma}{\term}{\tval^\gamma}{\eff'^\gamma}$.      
  \item[Case \ruleref{t-abs}]
    We have $\eff=\eff'$, $\term = (\lambda x.\, \term_1)$ for some $x$ and $\term_1$, and $\tau = \tfun{x}{\tval_2}{\eff_2}{\tval_1}{\eff_1}$ such that $\tjudge{\tenv,x: \tval_2; \eff_2}{\term_1}{\tval_1}{\eff_1}$. By induction hypothesis, we obtain $\tjudge{\tenv^\gamma,x: \tval_2^\gamma; \eff_2^\gamma}{\term_1}{\tval_1^\gamma}{\eff_1^\gamma}$. Using rule \ruleref{t-abs} we can immediately conclude $\tjudge{\tenv^\gamma;\eff^\gamma}{\term}{\tval^\gamma}{\eff'^\gamma}$.
  \item[Case \ruleref{t-cut}]
      There must exist $\cval$, $\tval'$, $\tval''$, $\effp$, and $x$ such that $\tjudge{\tenv;\eff}{\cval}{\tval'}{\eff}$, $x \notin \fv(\cval)$, $\tjudge{\tenv,x:\tval' ; \eff}{\term}{\tval''}{\effp}$, $\tval=\texists{x}{\tval'}{\tval''}$, and $\eff'=\texists{x}{\tval'}{\effp}$.
      By induction hypothesis, we have $\tjudge{\tenv^\gamma;\eff^\gamma}{\cval}{\tval'^\gamma}{\eff^\gamma}$ and $\tjudge{\tenv^\gamma,x:\tval'^\gamma ; \eff^\gamma}{\term}{\tval''^\gamma}{\effp^\gamma}$. Using rule \ruleref{t-cut} we can thus derive $\tjudgesimp{\tenv^\gamma;\eff^\gamma}{\term}{\texists{x}{\tval'^\gamma}{(\tandeff{\tval^\gamma}{\effp^\gamma})}}$. We trivially have $\subtjudge{\tenv^\gamma}{\tval^\gamma}{\tval^\gamma}$ and $\tstrengthen{\eff^\gamma}{\tenv^\gamma} \subseteq \eff^\gamma$. Moreover, by soundness of strengthening we have
      $(\texists{x}{\tval'^\gamma}{\effp^\gamma}) \subseteq \effgamma(\texists{x}{\tval'}{\effp}) = \eff'^\gamma$. It follows that $\tstrengthen{(\texists{x}{\tval'^\gamma}{\effp^\gamma})}{\tenv^\gamma} \subseteq \eff'^\gamma$. Hence, using rule \ruleref{t-weaken} we derive $\tjudge{\tenv^\gamma;\eff^\gamma}{\term}{\tval^\gamma}{\eff'^\gamma}$.      
    \item[Case \ruleref{t-weaken}]
      There must exist $\effp$, $\effp'$, and $\tval'$ such that $\tstrengthen{\eff}{\tenv} \efford \effp$, $\tjudge{\tenv ; \effp}{\term}{\tval'}{\effp'}$, $\subtjudge{\tenv}{\tval'}{\tval}$, and $\tstrengthen{\effp'}{\tenv} \efford \eff'$. By induction hypothesis, we have $\tjudge{\tenv^\gamma ; \effp^\gamma}{\term}{\tval'^\gamma}{\effp'^\gamma}$. Because of monotonicity of $\effgamma$, we have $\effgamma(\tstrengthen{\eff}{\tenv}) \subseteq \effgamma(\effp)=\effp^\gamma$.
      Because of the soundness assumption on effect strengthening, we also have $\tstrengthen{\eff^\gamma}{\tenv^\gamma} = \effgamma(\effp) \cap \tgamma(\tenv) \subseteq \effgamma(\tstrengthen{\eff}{\tenv})$. Thus, we obtain $\tstrengthen{\eff^\gamma}{\tenv^\gamma} \subseteq \effp^\gamma$. Using similar reasoning, we infer $\tstrengthen{\effp'^\gamma}{\tenv^\gamma} \subseteq \eff'^\gamma$. Finally, we apply \cref{lem:abstract-subtyping-gives-concrete-subtyping} to obtain $\subtjudge{\tenv^\gamma}{\tval'^\gamma}{\tval^\gamma}$. Then we can use rule \ruleref{t-weaken} to derive $\tjudge{\tenv^\gamma;\eff^\gamma}{\term}{\tval^\gamma}{\eff^\gamma}$. \qedhere
  \end{description}    
\end{proof}

We next show that the concrete instantiation of the type system satisfies progress and preservation. To ease notation, from here on meta variables like $\tval$, $\eff$, and $\tenv$ refer to types and effects of the concrete instantiation of the type system (unless specified otherwise).

We start with two technical lemmas that are needed to prove progress.

\begin{lemma}[Subtyping monotone]
  \label{lem:subtyping-monotone}
  If $\subtjudge{\tenv}{\tval}{\tval'}$, $\pair{\cval}{\cenv} \in \tval$, and $\cenv \in \tgamma(\tenv)$, then $\pair{\cval}{\cenv} \in \tval'$.
\end{lemma}

\begin{lemma}[Value typing]
  \label{lem:values}
  If $\tjudge{\tenv;\eff}{\cval}{\tval}{\eff'}$, then for all $\eff''$ with $\dom(\eff'') \subseteq \dom(\tenv)$, we have $\eff''$, $\tjudge{\tenv;\eff''}{\cval}{\tval}{\eff''}$. Moreover, for all $\cenv \in \tgamma(\tenv)$, $\pair{\cval}{\cenv} \in \tgamma(\tval)$.
\end{lemma}

\begin{proof}
  The proof goes by induction on the derivation of $\tjudge{\tenv;\eff}{\cval}{\tval}{\eff'}$. We do case analysis on the last typing rule that has been applied in the derivation.

\begin{description}
\item[Case \ruleref{t-var}, \ruleref{t-app}, and \ruleref{t-ev}] These rules cannot be the last rules that have been applied in the derivation since they do not apply to values.
\item[Case \ruleref{t-const}]
  We must have $\tval = \{\nu = \cval\} \in \bdom$ and $\eff'=\eff$. Let $\eff''$ be such that $\dom(\eff'') \subseteq \dom(\tenv)$. Using rule $\ruleref{t-const}$ we can immediately derive $\tjudge{\tenv;\eff''}{\cval}{\tval}{\eff''}$. Moreover, let $\cenv \in \tgamma(\tenv)$. By assumption on the operator $\{\nu = \cval\}$ we have $\{\nu = \cval\} \supseteq \set{\cval} \times \cenvs$. It follows that we have $\pair{\cval}{\cenv} \in \{\nu = \cval\} = \tval = \tgamma(\tval)$.
\item[Case \ruleref{t-abs}] We must have $\cval = \lambda x.\,\term$ for some $x$ and $\term$. Moreover, there must be $\tval_1$ and $\tval_2$, $\eff_1$ and $\eff_2$ such that $\tjudge{\tenv, x:\tval_2;\eff_2}{\term}{\tval_1}{\eff_1}$ and $\tval = (\tfun{x}{\tval_2}{\eff_2}{\tval_1}{\eff_1})$. Let $\eff''$ be such that $\dom(\eff'') \subseteq \dom(\tenv)$. Using rule \ruleref{t-abs} we can immediately derive $\tjudge{\tenv;\eff''}{\cval}{\tval}{\eff''}$. Moreover, let $\cenv \in \tgamma(\tenv)$. We have $\pair{\cval}{\cenv} \in \cvalues \times \cenvs = \tgamma(\tval)$.
\item[Case \ruleref{t-weaken}]
  There must exist $\effp$, $\effp'$, and $\tval'$ such that $\tstrengthen{\eff}{\tenv} \efford \effp$, $\tjudge{\tenv ; \effp}{\cval}{\tval'}{\effp'}$, $\subtjudge{\tenv}{\tval'}{\tval}$, and $\tstrengthen{\effp'}{\tenv} \efford \eff'$. Let $\eff''$ be such that $\dom(\eff'') \subseteq \dom(\tenv)$. By induction hypothesis, we conclude $\tjudge{\tenv ; \eff''}{\cval}{\tval'}{\eff''}$. By monotonicity of strengthening, we have $\tstrengthen{\eff''}{\tenv} \efford \eff''$. Using rule \ruleref{t-weaken} we can thus conclude $\tjudge{\tenv ; \eff''}{\cval}{\tval}{\eff''}$.
  Moreover, let $\cenv \in \tgamma(\tenv)$. By induction hypothesis, we have $\pair{\cval}{\cenv} \in \tgamma(\tval')$. Then by \cref{lem:subtyping-monotone} we have $\pair{\cval}{\cenv} \in \tgamma(\tval)$.  
\item[Case \ruleref{t-cut}] There must exist $\cval'$, $\tval'$, $\tval''$, $\effp$, and $x$ such that $\tjudge{\tenv;\eff}{\cval'}{\tval'}{\eff}$, $x \notin \fv(\cval)$, $\tjudge{\tenv,x:\tval' ; \eff}{\cval}{\tval''}{\effp}$, $\tval=\texists{x}{\tval'}{\tval''}$, and $\eff'=\texists{x}{\tval'}{\effp}$. Let $\eff''$ be such that $\dom(\eff'') \subseteq \dom(\tenv)$. Then also $\dom(\eff'') \subseteq \dom(\tenv,x:\tval')$. It follows by induction hypothesis that $\tjudge{\tenv,x:\tval' ; \eff''}{\cval}{\tval''}{\eff''}$. Using rule \ruleref{t-cut} we can thus derive $\tjudge{\tenv; \eff''}{\cval}{\tval}{\eff''}$. Moreover, let $\cenv \in \tgamma(\tenv)$. By induction hypothesis, we have $\pair{\cval'}{\cenv} \in \tgamma(\tval')$. Since $ \notin \dom(\tenv)$, it follows that $\cenv[x \mapsto \cval'] \in \tgamma(\tenv,x:\tval')$. Then again by induction hypothesis, $\pair{\cval}{\cenv[x \mapsto \cval']} \in \tgamma(\tval'')$. From the definition of $\tgamma$, we conclude $\pair{\cval}{\cenv} \in \tgamma(\texists{x}{\tval'}{\tval''}) = \tgamma(\tval)$. \qedhere
\end{description}
\end{proof}

We are now ready to show that the concrete instantiation of the type system satisfies progress. 

\begin{theorem}[Progress]
  \label{thm:progress}
  Let $\term$ be a closed term.
  If $\tjudge{\tenv;\eff}{\term}{\tval}{\eff'}$, then for all $\pair{\ceff}{\cenv} \in \tstrengthen{\eff}{\tenv}$, 
  $\term$ is a value and $\pair{\ceff}{\cenv} \in \eff'$ or there exist $\ceff'$ and $\term'$ such that $\step{\term,\ceff}{\term',\ceff'}$.
\end{theorem}

\begin{proof}
  The proof goes by induction on the derivation of $\tjudge{\tenv;\eff}{e}{\tval}{\eff'}$. We do case analysis on the last typing rule that has been applied in the derivation.
  
\begin{description}
\item[Case \ruleref{t-weaken}]
  We have $\tjudge{\tenv;\effp}{\term}{\tval'}{\effp'}$ for some $\effp$, $\tval'$, and $\effp'$ such that $\tstrengthen{\eff}{\tenv} \subseteq \effp$, $\tstrengthen{\effp'}{\tenv} \subseteq \eff'$, and $\subtjudge{\tenv}{\tval'}{\tval}$.
  Let $\pair{\ceff}{\cenv} \in \tstrengthen{\eff}{\tenv}$. Then also $\cenv \in \tgamma(\tenv)$.
  From $\tstrengthen{\eff}{\tenv} \subseteq \effp$ we then conclude $\pair{\ceff}{\cenv} \in \tstrengthen{\effp}{\tenv}$.
  It follows from the induction hypothesis that $e$ is a value and $\pair{\ceff}{\cenv} \in \effp'$ or there exist $\ceff'$ and $\term'$ such that $\step{\term,\ceff}{\term',\ceff'}$. In the second case we are done. In the first case, we use $\tstrengthen{\effp'}{\tenv} \subseteq \eff'$ and $\cenv \in \tgamma(\tenv)$ to conclude that $\pair{\ceff}{\cenv} \in \eff'$.

\item[Case \ruleref{t-cut}] We have $\tval=\texists{x}{\tval'}{\tval''}$ and $\eff=\texists{x}{\tval'}{\eff''}$ such that  $\tjudge{\tenv;\eff}{\cval}{\tval'}{\eff}$, and
  $\tjudge{\tenv,x:\tval';\eff}{\term}{\tval''}{\eff''}$. Therefore, we must also have $\wf(\tenv,x:\tval')$.
  Let $\pair{\ceff}{\cenv} \in \tstrengthen{\eff}{\tenv}$. Then also $\cenv \in \tgamma(\tenv)$. 
  By \cref{lem:values}, we have $\pair{\cval}{\cenv} \in \tgamma(\tval')$.
  It follows from the definition of $\tgamma$ and $\wf(\tenv,x:\tval')$ that $\cenv[x \mapsto \cval] \in \tgamma(\tenv,x:\tval')$. From the monotonicity of strengthening, it follows that $\pair{\cval}{\cenv} \in \tstrengthen{\eff}{\varnothing} = \eff$. Moreover, $\dom(\eff) \subseteq \dom(\tenv)$ and $x \notin \dom(\tenv)$ implies $x \notin \dom(\eff)$. Therefore, $\pair{\ceff}{\cenv[x\mapsto\cval]} \in \eff$. By the definition of strengthening we now conclude $\pair{\ceff}{\cenv[x\mapsto\cval]} \in \tstrengthen{\eff}{\tenv,x:\tval'})$.
  From the induction hypothesis, it then follows that either $\term$ is a value and $\pair{\ceff}{\cenv[x\mapsto\cval]} \in \eff''$ or there exists $\term'$ and $\ceff'$, such that $\step{\term,\ceff}{\term',\ceff'}$. In the second case we are done. In the first case, we note that $\pair{\ceff}{\cenv[x\mapsto\cval]} \in \eff''$ and $\pair{\cval}{\cenv} \in \tgamma(\tval')$ implies $\pair{\ceff}{\cenv} \in (\texists{x}{\tval'}{\eff''}) = \eff'$.

\item[Case \ruleref{t-const} and \ruleref{t-abs}] In both cases $\term$ is a value and $\eff=\eff'$. So the claim follows immediately.




\item[Case \ruleref{t-var}] By assumption, $\term$ is closed. So this rule cannot be the last rule used in the typing derivation.
  
\item[Case \ruleref{t-ev}]
  We have $e = \termkw{ev}~\term_1$ for some $\term_1$. Moreover, there exists $\eff_1$ and $\bval_1$ such that $\tjudge{\tenv;\eff}{\term_1}{\bval_1}{\eff_1}$ and $\eff' = \eff_1 \odot \bval_1$.

  Suppose $\term_1$ is not a value. Let $\pair{\ceff}{\cenv} \in \tstrengthen{\eff}{\tenv}$ and $\cenv \in \tgamma(\tenv)$. By induction hypothesis, there must exist $\ceff'$ and $\term_1'$ such that $\step{\term_1,\ceff}{\term_1',\ceff'}$. Thus, we also have $\step{\term,\ceff}{\term',\ceff'}$ for $\term' = \termkw{ev}~\term_1'$ by rule \ruleref{e-context}.

  Suppose on the other hand that $\term_1$ is a value. Let $\pair{\ceff}{\cenv} \in \tstrengthen{\eff}{\tenv}$. We have $\step{\term,\ceff}{\term',\ceff'}$ for $\term' = \unitval$ and $\ceff' = \ceff \cdot \term_1$ by rule \ruleref{e-ev}.

\item[Case \ruleref{t-app}]
  We have $\term = \term_1~\term_2$ for some $\term_1$ and $\term_2$.
  Moreover, there exist $\effp$, $\eff_1$ and $\eff_2$ as well as $\tval'$, $\tval_1$ and $\tval_2$ such that $\tjudge{\tenv;\eff}{e_1}{\tval_1}{\eff_1}$, $\tjudge{\tenv;\eff_1}{e_2}{\tval_2}{\eff_2}$, and $\tau_1 = \tfun{x}{\tval_2}{\eff_2}{\tval'}{\effp}$.

  If $\term_1$ is not a value, then by induction hypothesis, for all $\pair{\ceff}{\cenv} \in \tstrengthen{\eff}{\tenv}$
 there exist $\ceff'$ and $\term_1'$ such that $\step{\term_1,\ceff}{\term_1',\ceff'}$.
  It follows that $\step{\term,\ceff}{\term',\ceff'}$ for $\term' = \term_1' \; \term_2$ by rule \ruleref{e-context}.

  The case where $\term_1$ is a value but $\term_2$ is not is similar to the previous case.

  Thus, let us assume that both $\term_1$ and $\term_2$ are values.
  Since $\tval_1 = \tfun{x}{\tval_2}{\eff_2}{\tval'}{\effp}$ and $e_1$ is a value, we must have $e_1 = \lambda x.\,\term$ for some $x$ and $\term$.
  It follows for all $\pair{\ceff}{\cenv} \in \tstrengthen{\eff}{\tenv}$ that $\step{\term,\ceff}{\term',\ceff'}$ for $\term' = \term[\term_2/x]$ and $\ceff'=\ceff$ by rule \ruleref{e-app}.\qedhere
\end{description}

\end{proof}

We next turn to proving preservation. Again, we start with some technical lemmas.

\begin{lemma}
  \label{lem:swap}
  If $\tjudge{\tenv,x:\tval_x,y:\tval_y;\eff}{e}{\tval}{\eff'}$ and $\wf(\tenv,y:\tval_y)$, then $\tjudge{\tenv,y:\tval_y,x : \tval_x;\eff}{e}{\tval}{\eff'}$.
\end{lemma}

\begin{proof}
  The proof goes by induction on the derivation of $\tjudge{\tenv,x:\tval_x,y:\tval_y;\eff}{e}{\tval}{\eff'}$ using the fact that $\wf(\tenv,x:\tval_x,y:\tval_y)$ and $\wf(\tenv,y:\tval_y)$ implies (1) $\wf(\tenv,y:\tval_y,x:\tval_x)$, (2) $\tgamma(\tenv,x:\tval_x,y:\tval_y)=\tgamma(\tenv,y:\tval_y,x:\tval_x)$, and (3) for all $z \in \dom(\tenv,x:\tval_x,y:\tval_y) = \dom(\tenv,y:\tval_y,x:\tval_x)$, $(\tenv,y:\tval_y,x:\tval_x)(z) = (\tenv,x:\tval_x,y:\tval_y)(z)$.
\end{proof}

\begin{lemma}
  \label{lem:strengthen-env-sub}
  If $\subtjudge{\tenv}{\tval}{\tval'}$ and $\wf(\tenv,x:\tval_x)$, then $\subtjudge{\tenv,x : \tval_x}{\tval}{\tval'}$.
\end{lemma}

\begin{proof}
  The proof goes by induction on the derivation of $\subtjudge{\tenv}{\tval}{\tval'}$, using the fact that for all $\bval \subseteq \cvalues \times \cenvs$, we have
  \begin{align*}
   & \tstrengthen{\bval}{\tenv,x:\tval_x}\\
   = \; & \bval \cap \tgamma(\tenv,x:\tval_x)\\
   = \; & \bval \cap \tgamma(\tenv) \cap \pset{\cenv}{\exists \cval.\, \pair{\cval}{\cenv} \in \tgamma(\tval_x)}\\
   \subseteq \; & \bval \cap \tgamma(\tenv)\\
 = \; & \tstrengthen{\bval}{\tenv}
  \end{align*}
  and, similarly, that for all $\eff \subseteq \ceffs \times \cenvs$, we have $\tstrengthen{\eff}{\tenv,x:\tval_x} \subseteq \tstrengthen{\eff}{\tenv}$.
\end{proof}

\begin{lemma}
  \label{lem:strengthen-env}
  If $\tjudge{\tenv;\eff}{e}{\tval}{\eff'}$ and $\wf(\tenv,x:\tval_x)$, then $\tjudge{\tenv,x : \tval_x;\eff}{e}{\tval}{\eff'}$.
\end{lemma}

\begin{proof}
  We prove the claim by induction on the derivation of $\tjudge{\tenv; \eff}{\term}{\tval}{\eff'}$. We proceed by case analysis on the last rule that is applied in the derivation. We only show some of the more interesting cases.

\begin{description}
\item[Case \ruleref{t-var}] We have $\eff=\eff'$ and $\term = y$ for some $y \in \vars$ such that $\tenv(y)=\tval$. Since $\wf(\tenv,x:\tval_x)$, we must have $x \neq y$. Hence, $(\tenv,x:\tval_x)(y) = \tval$. Using rule \ruleref{t-var}, we then derive $\tjudge{\tenv,x : \tval_x;\eff}{\term}{\tval}{\eff'}$.
  
\item[Case \ruleref{t-abs}] We have $\eff=\eff'$, $\term = (\lambda y.\, \term_1)$ for some $y$ and $\term_1$, and $\tau = \tfun{y}{\tval_2}{\eff_2}{\tval_1}{\eff_1}$ such that $\tjudge{\tenv,y: \tval_2; \eff_2}{\term_1}{\tval_1}{\eff_1}$. Without loss of generality we have $x \neq y$. Since $\wf(\tenv,x:\tval_x)$ and $\wf(\tenv,y:\tval_2)$, we also have $\wf(\tenv,y:\tval_2,x:\tval_x)$. By induction hypothesis, we obtain $\tjudge{\tenv,y: \tval_2,x:\tval_x; \eff_2}{\term_1}{\tval_1}{\eff_1}$. Using \cref{lem:swap} we infer $\tjudge{\tenv,x:\tval_xy: \tval_2; \eff_2}{\term_1}{\tval_1}{\eff_1}$. Thus, we derive $\tjudge{\tenv,x : \tval_x;\eff}{e}{\tval}{\eff'}$ using rule \ruleref{t-abs}.

\item[Case \ruleref{t-weaken}]
  We have $\tjudge{\tenv;\effp}{\term}{\tval'}{\effp'}$ for some $\effp$, $\tval'$, and $\effp'$ such that $\tstrengthen{\eff}{\tenv} \subseteq \effp$, $\tstrengthen{\effp'}{\tenv} \subseteq \eff'$, and $\subtjudge{\tenv}{\tval'}{\tval}$. By induction hypothesis, we have $\tjudge{\tenv,x:\tval_x;\effp}{\term}{\tval'}{\effp'}$. Moreover, using \cref{lem:strengthen-env-sub}, we infer $\subtjudge{\tenv,x:\tval_x}{\tval'}{\tval}$. Next, we use $\tstrengthen{\eff}{\tenv} \subseteq \effp$ and $\tstrengthen{\eff}{\tenv,x:\tval_x} \subseteq \tstrengthen{\eff}{\tenv}$ to infer $\tstrengthen{\eff}{\tenv,x:\tval_x} \subseteq \effp$. Using similar reasoning, we infer $\tstrengthen{\effp'}{\tenv,x:\tval_x} \subseteq \eff'$. Thus, we can apply rule \ruleref{t-weaken} to derive the desired $\tjudge{\tenv,x : \tval_x;\eff}{\term}{\tval}{\eff'}$.\qedhere
\end{description}

\end{proof}

\begin{lemma}[Substitution]
  \label{lem:substitution}
  If $\tjudge{\tenv;\eff_v}{v}{\tval_v}{\eff'_v}$ and $\tjudge{\tenv,x : \tval_v,\tenv';\eff}{\term}{\tval}{\eff'}$, then $\tjudge{\tenv,x:\tval_v,\tenv';\eff}{\term[v/x]}{\tval}{\eff'}$.
\end{lemma}

\begin{proof}
  We prove the claim by induction on the derivation of $\tjudge{\tenv,x : \tval_v,\tenv'; \eff}{\term}{\tval}{\eff'}$. We proceed by case analysis on the last rule that is applied in the derivation. We only show some of the cases. The omitted cases are similar to the case for \ruleref{t-abs} but follow simpler reasoning.

\begin{description}
\item[Case \ruleref{t-const}]
  We have $\eff=\eff'$ and $\term = c$ for some constant value $c$. It follows $\term[v/x]=c$. Hence, $\tjudge{\tenv,x:\tval_v,\tenv';\eff}{\term[v/x]}{\tval}{\eff'}$.

\item[Case \ruleref{t-var}] We have $\eff=\eff'$ and $\term = y$ for some $y \in \vars$ such that $(\tenv,x \mapsto \tval_v,\tenv')(y)=\tval$. If $x \neq y$, then $\term[v/x]=y$. Hence, we immediately obtain $\tjudge{\tenv,x:\tval_v,\tenv';\eff}{\term[v/x]}{\tval}{\eff'}$.

  On the other hand, if $x=y$, then $\tval = \tval_v$ and $e[v/x]=v$. Using $\tjudge{\tenv;\eff_v}{v}{\tval_v}{\eff'_v}$ and \cref{lem:values}, we first infer $\tjudge{\tenv;\eff}{v}{\tval_v}{\eff}$. By repeatedly applying \cref{lem:strengthen-env}, we can then derive the desired $\tjudge{\tenv,x:\tval_v,\tenv';\eff}{v}{\tval_v}{\eff}$.

\item[Case \ruleref{t-abs}] We have $\eff=\eff'$, $\term = (\lambda y.\, \term_1)$ for some $y$ and $\term_1$, and $\tau = \tfun{y}{\tval_2}{\eff_2}{\tval_1}{\eff_1}$ such that $\tjudge{\tenv,x: \tval_v,\tenv',y: \tval_2; \eff_2}{\term_1}{\tval_1}{\eff_1}$. By induction hypothesis, we obtain $\tjudgesimp{\tenv,x:\tval_v,\tenv',y : \tval_2; \eff_2}{\term_1[v/x]}{\tandeff{\tval_1}{\eff_1}}$. Thus, using rule~\ruleref{t-abs} and $x \neq y$, we derive $\tjudgesimp{\tenv,x:\tval_v,\tenv'; \eff_2}{(\lambda y.\, \term_1)[v/x]}{\tandeff{\tval}{\eff}}$. \qedhere
  
\end{description}
\end{proof}

We now prove that the concrete instantiation of the type system satisfies progress.

\begin{theorem}[Preservation]
  \label{thm:preservation}
  If $\step{e,\ceff}{e',\ceff'}$ and $\tjudge{\eff}{e}{\tval}{\eff'}$ where
$\pair{\ceff}{\cenv} \in \eff$ for some $\cenv$. Then there exists $\eff''$ such that $\tjudge{\eff''}{e'}{\tval}{\eff'}$ and $\pair{\ceff'}{\cenv} \in \eff''$.
\end{theorem}

\begin{proof}
We prove the claim by induction on the derivation of $\tjudge{\eff}{\term}{\tval}{\eff'}$. We only discuss the cases for rules~\ruleref{t-ev} and \ruleref{t-app} as they are the most involved.

\begin{description}
\item[Case \ruleref{t-ev}] We have $\term = \termkw{ev}\; \term_1$ for some $\term_1$ and $\tval=\{\nu = \unitval\}$. Moreover, there exists $\bval$ and $\eff_1$ such that $\tjudge{\eff}{e_1}{\bval}{\eff_1}$ and $\eff' = \eff_1 \odot \bval_1$.

  The case where $\term_1$ is not a value follows from the induction hypothesis.

  If $\term_1$ is a value, then $e'=\unitval$ and $\ceff' = \ceff \cdot \term_1$.
  Moreover, by \cref{lem:values} we can assume that $\eff_1=\eff$ and $\pair{e_1}{\cenv} \in \bval$.
  From $\pair{\term_1}{\cenv} \in \tgamma(\tval_1)$, $\pair{\ceff}{\cenv} \in \eff_1$, and the requirement on $\odot$, we conclude that $\pair{\ceff'}{\cenv} \in \eff'$. The claim then follows for $\eff'' = \eff'$.

\item[Case \ruleref{t-app}.] We have $\term = \term_1\;\term_2$ for some $\term_1$ and $\term_2$. Moreover, there exist $\eff_1$, $\eff_2$, and $\effp$ as well as $\tval_1$, $\tval_2$, and $\tval'$ such that $\tjudge{\eff}{\term_1}{\tval_1}{\eff_1}$, $\tjudge{\eff_1}{\term_2}{\tval_2}{\eff_2}$, $\tval_1 = \tfun{x}{\tval_2}{\eff_2}{\tval'}{\effp}$, and $\tandeff{\tval}{\eff'}=\texists{x}{\tau_2}{(\tandeff{\tval'}{\effp})}$.
  
  If $\term_1$ is not a value, then we must have $\term' = \term_1' \; \term_2$ for some $\term_1'$ such that $\step{\term_1,\ceff_1}{\term_1',\ceff_1'}$ for some $\ceff_1'$. It follows by the induction hypothesis that there exists $\eff''$ such that $\tjudge{\eff''}{\term_1'}{\tval_1}{\eff_1}$ and $\pair{\ceff_1'}{\cenv} \in \eff''$. Thus, using rule \ruleref{t-app} we conclude $\tjudge{\eff''}{\term'}{\tval}{\eff'}$.

  The case where $\term_1$ is a value but $\term_2$ is not is similar to the previous case.

  Thus, let us assume that both $\term_1$ and $\term_2$ are values.
  Since $\tval_1 = \tfun{x}{\tval_2}{\eff_2}{\tval'}{\effp}$ and $\term_1$ is a value, we must have $\term_1 = \lambda x.\,e$ for some $x$ and $e$.
  It follows that $\term' = \term[\term_2/x]$ and $\ceff'=\ceff$.
  Moreover, from \cref{lem:values} it follows that we may assume $\eff = \eff_1 = \eff_2$.
  By rule \ruleref{t-abs}, we must have $\tjudge{x : \tval_2; \eff_2}{\term}{\tval'}{\effp}$.
  By \cref{lem:substitution} we have $\tjudge{x : \tval_2; \eff_2}{\term[e_2/x]}{\tval'}{\effp}$. Thus, using rule~\ruleref{t-cut}, we infer $\tjudge{\eff_2}{\term'}{\tval}{\eff'}$.
  Since $\ceff' = \ceff$, $\pair{\ceff}{\cenv} \in \eff$ and $\eff = \eff_1 = \eff_2$,
  the claim then follows by choosing $\eff'' = \eff$. \qedhere

\end{description}

\end{proof}

Finally, we are ready to prove \cref{thm:soundness}.

\begin{proof}[Proof of \cref{thm:soundness}]
  Assume that $\tjudge{\eff}{\term}{\tval}{\eff'}$ holds in the abstract type system and let $\pair{\ceff}{\cenv} \in \bgamma(\eff)$ and $\pair{\term'}{\ceff'}$ such that $\pair{\term}{\ceff} \leadsto \pair{\term'}{\ceff'}$.

  Let $\pair{\term}{\ceff} = \pair{\term_0}{\ceff_0} \rightarrow \pair{\term_1}{\ceff_1} \rightarrow \dots \rightarrow \pair{\term_n}{\ceff_n} = \pair{\term'}{\ceff'}$ be a sequence of reduction steps used to obtain $\pair{\term'}{\ceff'}$ from $\pair{\term}{\ceff}$ for $n \geq 0$.

  We show by induction on $i$, $0 \leq i \leq n$ that there exist $\eff_i$ such that $\tjudge{\eff_i}{\term}{\tval^\gamma}{\eff'^\gamma}$ and $\pair{\ceff_i}{\cenv} \in \eff_i$.

  If $i = 0$, we can use \cref{lem:abstract-typing-gives-concrete-typing} to conclude that $\tjudge{\eff^\gamma}{\term_0}{\tval^\gamma}{\eff'^\gamma}$. The claim then follows for $\eff_0=\eff^\gamma$.

  If $0 < i \leq n$, then by induction hypothesis we have $\tjudge{\eff_{i-1}}{\term_{i-1}}{\tval^\gamma}{\eff'^\gamma}$ for some $\eff_{i-1}$ such that $\pair{\ceff_{i-1}}{\cenv} \in \eff_{i-1}$. Using \cref{thm:preservation}, we conclude that there exists $\eff_i$ such that $\tjudge{\eff_{i}}{\term_i}{\tval^\gamma}{\eff'^\gamma}$ and $\pair{\ceff_{i}}{\cenv} \in \eff_{i}$.

  We thus have $\tjudge{\eff_n}{\term'}{\tval^\gamma}{\eff'^\gamma}$ and $\pair{\ceff'}{\cenv} \in \eff_n$. By assumption there is no $\pair{\term''}{\cenv''}$ such that $\step{\term',\ceff'}{\term'',\ceff''}$. Using \cref{thm:progress} we thus conclude that $\term'$ must be a value and $\pair{\ceff'}{\cenv} \in \eff'^\gamma$.
  
\end{proof}




\begin{figure}
  \begin{mathpar}
    \inferHtop{s-bot[d]}
    {\tvald \neq \top^d}
    {\bot^d <: \tvald}
    \and
    \inferHtop{s-base[d]}
    {\beta_1 \sqsubseteq^b \beta_2}
    {\beta_1 <: \beta_2}
    \and
    \inferHtop{s-fun[d]}
    {\tvald_{x}' <: \tvald_{x} \\ \tvald[x : \tvald_{x}'] <: \tvald'[x : \tvald_{x}']}
    {x:\tvald_x \to \tvald \;<:\; x:\tvald_x' \to \tvald'}
    \and
    \inferHtop{s-pair[d]}
    {\tvald_1 <: \tvald_1' \\ \tvald_2 <: \tvald_2'}
    {(\tvald_1 \times \tvald_2) <: (\tvald_1' \times \tvald_2')}
    \and
    \inferH{t-const[d]}
    {[\nu = c]^d[\tenvd] <: \tvald}
    {\tenvd \vdash c : \tvald}
    \and
    \inferH{t-var[d]}
    {\tenvd(x)[\nu=x][\tenvd] <: \tvald[v=x][\tenvd]}
    {\tenvd \vdash x : \tvald}
    \and
    \inferHtop{t-abs[d]}
    {\tenvd, x:\tvald_x \vdash e_i : \tvald_i \\ (x:\tvald_x \to \tvald_i) <: \ttabled}
    {\tenvd \vdash \lambda x.e_i : \ttabled}
    \and
    \inferHtop{t-app[d]}
    {\tenvd \vdash e_i: \tvald_i \\ \tenvd \vdash e_j:\tvald_j \\ \tvald_i <: (x:\tvald_j \to \tvald)}
    {\tenvd \vdash e_i~e_j: \tvald}
    \and
    \inferHtop{t-pair[d]}
    {\tenvd \vdash e_1:\tvald_1 \\ \tenvd \vdash e_2:\tvald_2}
    {\tenvd \vdash \langle e_1, e_2 \rangle:\tvald_1 \times \tvald_2}
    \and
    \inferHtop{t-proj[d]}
    {\tenvd \vdash e : \tvald_1 \times \tvald_2}
    {\tenvd \vdash \#_i(e):\tvald_i, i=1,2}
    \and
    \inferHtop{t-extension-op[d]}
    {\tenvd \vdash e_1:\eff \\ e_2:\bval \\ \eff' = \eff \odot \bval}
    {\tenvd \vdash e_1 \eop e_2 : \eff'}
  \end{mathpar}
  \caption{\label{fig:drift-type-rules} Data flow refinement type system}  
\end{figure}

Next we will describe how to translate programs with effects into monadified programs without effects. This then allows us to instantiate the inference algorithm to infer effect summaries.

\section{Soundness of Type and Effect Inference}
\label{sec:inference-soundness-top}

We now provide an extended discussion on the type and effect inference algorithm introduced in Sec. \ref{sec:inference}
At a high level, there are four steps:
(i) translate the program so that prefix event traces $\pi$ are symbolically represented at the syntactic level,
(ii) extend refinement \emph{types} and type inference to support \emph{event sequences} using our novel effect abstract domain, 
(iii) the types of those event sequences then correspond to \emph{effects} in our type and effect system in \cref{sec:types-and-effects} and then
(iv) ensure that at every program location  \circled{i}, the computed summary associates $\bot$ with every accepting state of the SAA that encodes the property of interest. We now describe these steps; for lack of space details are available in \cref{apx:inference}.


\subsection{Dataflow refinement type system with event sequences as program values}

In our instantiation of the \drift{} type system, we treat event sequences as values that can be manipulated directly by the program. The only primitive operator defined on event sequences is $e_1 \eop e_2$ where $e_1$ is expected to evaluate to an event sequence $\ceff_1$ and $e_2$ to a value $\cval_2$. The result of the operation is the concatenated event sequence $\ceff_1 \cdot \cval_2$. We additionally have the constant expression $\epsilon$ denoting the empty event sequence. We also have a built-in pair constructor $\pair{e_1}{e_2}$ and projection operators $\#_1(e)$ and $\#_2(e)$ on pairs.
In order to support inference of effects through dataflow refinement inference, we introduce a set of primitive values in the language: sequences of events.
Each event is a primitive value already tracked by the currently support basic refinement types of the language. 
We additionally support primitives to construct sequences of events: an empty sequence and an append operator.
Finally, we extend the type system to support \emph{effects} that represent event sequences.
The instantiated \drift{} types are then:
\[  \tvald ~::=~ \bot^d ~|~ \top^d ~|~ \bval ~|~ \eff ~|~ x:\tvald_1 \to \tvald_2 ~|~ \tvald_1 \times \tvald_2\enspace.\]
The types $\top^d$ and $\bot^d$ are the extrema of the type lattice. $\top^d$ can be thought of as representing a type error and $\bot^d$ represents the empty set of values, indicating that evaluation of an expression never returns, respectively, that the expression is unreachable in the program. $\bval$ is an element of $\bdom$, $\eff$ is an element of $\effdom$, $x:\tvald_1 \to \tvald_2$ is a dependent function type and $\tvald_1 \times \tvald_2$ a pair type.

It comes equipped with type operations $\tvald[x=y]$, $\tvald[\tenvd]$ that strengthen a type $\tvald$ with constraints corresponding to the equality between two variables, and the typing environment respectively.
These operators are defined similarly for base types to the strengthening operator $\bval[\tenv]$ in our type and effect system, the latter operator recursively pushes the strengthening of compound types to each constituent type.
The type constructor $[\nu=c]^d$ returns a type abstracting constant value $c$.
It con by enforcing an equality with the type variable $\nu$.  
Lastly, we assume a built-in primitive concatenation operator $\eop$ that extends an event prefix with a new event value.

The inferred \emph{types of encoded event sequences} correspond to \emph{effects} of the events in the original program.
A challenge in connecting the type inference result with our type and effects system is that the inference algorithm has been proven sound with respect to a bespoke dataflow semantics of functional program rather than a standard operational semantics like the one underlying our system. To bridge this gap, we relate the two type systems at the abstract level by showing that, from the typing derivation for a translated program produced by the soundness proof of \cite{DBLP:journals/pacmpl/PavlinovicSW21}, one can reconstruct a typing derivation in the types and effects system for the original effectful program. The key soundness \cref{thm:inference-sound} and its proof are in \cref{sec:inference-soundness}. The overall soundness then follows from \cref{thm:soundness}.

\subsection{Translation from effectful programs to programs with sequences}
\label{sec:tuple-translation}


\Cref{subfig:tr-term} defines the translation function from effectful programs into a functional language where $\termkw{ev}$ expressions are absent. This function preserves the operational semantics while ensuring that the events emitted during the program's execution are carried through the computation. 

\begin{figure}[t]
  \begin{subfigure}[t]{1\textwidth}\centering
  \[\arraycolsep=0pt\def\arraystretch{1.3}
    \begin{array}[t]{l}
      
      \tr\den{c}(\pivar) \Def= \langle c, \pivar \rangle
      \qquad \qquad \tr\den{x}(\pivar) \Def= \langle x, \pivar \rangle
      \qquad \qquad\tr\den{\lambda x.\,e}(\pivar) \Def= \langle \lambda x. \lambda y.(\tr \den{e}(y)), ~\pivar \rangle \\[3pt]
      \begin{array}[t]{lcr}
        \begin{array}[t]{l}
          \tr\den{e_1 ~ e_2}(\pivar) \Def= \\
          \;\;\;\;\begin{array}[t]{l}
            \termkw{let} ~\pair{x_1}{y_1} = \#_{1,2}(\tr\den{e_1}(\pivar))~\termkw{in}\\
            \termkw{let} ~\pair{x_2}{y_2} = \#_{1,2}(\tr\den{e_2}(y_1)) ~\termkw{in}\\
            x_1~x_2~y_2
          \end{array}
        \end{array}& \qquad \qquad \qquad \qquad \qquad &
        \begin{array}[t]{l}
          \tr\den{\termkw{ev}~e}(\pivar) \Def= \\
          \;\;\;\;\begin{array}[t]{l}
            \termkw{let} ~(x, y) = \#_{1,2}(\tr\den{e}(\pivar)) ~\termkw{in}\\
            \langle \unitval, y \eop x\rangle
          \end{array}
        \end{array}
      \end{array}
    \end{array}\]
  \caption{Forward term tranformation}
  \label{subfig:tr-term}
\end{subfigure}
\begin{subfigure}[t]{1\textwidth}\centering
\[\arraycolsep=5pt\def\arraystretch{1.3}
  \begin{array}{l}
    \begin{array}{l}
      \ttr(\bot^d) \Def= \bbot \quad \ttr(\top^d) \Def= \btop \quad \ttr(\bval) \Def= \bval \\[3pt]
      \ttr((x:\tvald_x) \to (\pivar_x:\eff_x) \to \tvald)) \Def= x:\tandeff{\ttr(\tvald_x)}{\eff_x} \to \tetr(\tvald) \qquad \quad
      \tetr(\pair{\tvald}{\eff}) \Def= \tandeff{\ttr(\tvald)}{\eff}
    \end{array}
  \end{array}
\]
\caption{Backward type and type/effect translation}
\label{subfig:tr-type}
\end{subfigure}
\caption{Tuple-encoding-based translation}
\label{fig:tr-tuple-encode}
\end{figure}

The translation functions $\tr\den{e}(\pivar)$ take two arguments: the effectful expression $e$ in the source language and an expression $\pivar$ in the target language that evaluates to the effect prefix produced by the context of $e$. 
The transformation follows the call-by-value, left-to-right evaluation order of source language's operational semantics.
For constant $\textrm{\lstinline|c|}$ and variable $\textrm{\lstinline|x|}$ expressions, we pair them with the event prefix observed up to the current evaluation context.
Lambda abstractions $\lambda x.\,e$ go through a syntactic transformation and are then paired with the event prefix in the evaluation environment. 
The translated function expressions expects an additional parameter $y$ representing the event prefix observed at the call site.
Then, as expected, the translation of the function body considers $y$ as the new event prefix.
Thus, a translated function always returns a pair, where the second component represents the event sequence produced after evaluation of a function call.
The translation of application terms $\tr\den{e_1~e_2}(\pivar)$ ensures the strict-evaluation semantics for our source language. We here abbreviate the sequence of let bindings, i.e.,
$
\termkw{let} ~x = e ~\termkw{in}
~\termkw{let} ~x_1 = \#_1~x ~\termkw{in}
~\termkw{let} ~x_1 = \#_2~x ~\termkw{in} ~e'
$,
with $\termkw{let}~ \pair{x_1}{x_2} = \#_{1,2}~e ~\termkw{in} ~e'$.
The last expression we consider is the event expression.
Translation $\tr\den{\termkw{ev}~e}(\pivar)$ follows the order of evaluation by first converting $e$ in the context of the current event prefix and then capturing the event sequence $y$ associated with its result value $x$. The result is the pair consisting of the unit value and the extended event sequence $y \eop x$.


\subsection{Soundness}
\label{sec:inference-soundness}

We prove the theorem that inference of type and effect via program translation is sound.
Intuitively, the theorem states that if we can obtain a \drift{} typing derivation for a translated term, then we can construct a derivation for the typing judgment in the type and effects systems.
The construction uses backward translation functions $\ttr$ and $\tetr$, defined in \cref{subfig:tr-type}, that embed types in the translated program back to types and type/effect pairs, respectively.

\begin{theorem}
  \label{thm:inference-sound}
  If $\tenvd, y:\eff \tjdgd \tr\den{e}(y): \tvald$, then $\ttr(\tenvd);\eff \tejdg e:\tetr(\tvald)$. 
\end{theorem}


The proof of \cref{thm:inference-sound} builds on the structure of the translated program and relates backward translatable typing environments and types. 
We start by showing that base refinement types have the same type semantics, and that strengthening a base type with respect to the typing environment in the target type system gives us a semantically equal type after strengthening with backward translated typing environment in the type and effects system. 
Then we show that if we have a \drift{} subtyping derivation for backward translatable types and typing environments, $\tvald\tstren{\tenvd} \subtyped \tvald'\tstren{\tenvd}$, then we can obtain a subtyping derivation $\ttr(\tenvd) \tejdg \ttr(\tvald) \subtype \ttr(\tvald')$ in the type and effects system.
The proof proceeds by structural induction on the source expression $e$. The details follow below.


We assume that strengthening operator is extended with the following clauses:
\[\tvald\tstrend{x \leftarrow \eff} = \tvald \qquad \tvald\tstrend{x \leftarrow (\tvald_1 \times \tvald_2)} = \tvald \qquad
  (\tvald_1 \times \tvald_2)\tstrend{x \leftarrow \tvald'} = (\tvald_1\tstrend{x \leftarrow \tvald'} \times \tvald_2\tstrend{x \leftarrow \tvald'})
\]

\begin{proposition}
  \label{prop:base-equality}
  If $\beta^d \in \mathcal{R}^t_X$ with scope $X$, $\beta \in \bdom$ and $X = dom(\beta)$ then $\beta^d = \beta$ iff ${\gamma}^d(\beta^d) \cup (\cvalues \times Env) = \gamma(\beta)$
\end{proposition}
\begin{proof}
 The proof is immediate from the definition of types. 
\end{proof}
A similar result can be shown for base types $\eff^d$ representing event sequences. Hereafter, $\bval$ (and $\eff$) denotes base type (and effect) both in the effect-free language and in our language that agrees on the scope and represent the same values (event sequences). 

Recall the two simultaneously inductive backwards translation functions that embed types in the translated program back to types and type and effects respectively. 
\[\arraycolsep=5pt\def\arraystretch{1.3}
  \begin{array}{l}
    \begin{array}{l}
      \ttr(\bot^d) \Def= \bbot \qquad \ttr(\top^d) \Def= \btop \qquad \ttr(\bval) \Def= \bval \\
      \ttr((x:\tvald_x) \to (\pivar_x:\eff_x) \to \tvald)) \Def= x:\tandeff{\ttr(\tvald_x)}{\eff_x} \to \tetr(\tvald) \\
      \tetr(\pair{\tvald}{\eff}) \Def= \tandeff{\ttr(\tval)}{\eff}
    \end{array}
  \end{array}
\]
For all the other cases both $\ttr$ and $\tetr$ are undefined. We lift this backwards type translation function to typing environments
\[
  \ttr(\emptyset) = \emptyset \qquad \ttr(\tenvd, x:\tvald) = \ttr(\tenvd), x: \ttr(\tvald) \\
\]

Note that even the backward translation functions are partially defined, they are sufficient to state the following lemmas and theorem that must hold only for the typing derivations of translated programs.
It's immediate to see that, given a closed program, the typing judgments of the corresponding translated terms relate typing contexts on which the backward translation is defined.
That does not impose a restriction on the typing derivation of the subexpressions of the translated terms.  
In the remaining part of this section, when referring to a typing judgment or subtyping relation, we consider only typing environments that are in the domain of $\ttr$. Henceforth, we refer to them as backward translatable type or typing environment.  

The following lemma states that strengthening a base type with respect to a backward translatable typing environment is the same as strengthening it with the translated typing environment in the type and effect system

\begin{lemma}
  \label{lem:type-strength}
  For a basic refinement type $\bval$, $\tenvd$ on $\bval\tstrend{\tenvd} = \bval[\ttr(\tenvd)]$  
\end{lemma}
\begin{proof}
  The proof goes by induction on the length of the environment $\tenvd$.

  \begin{description}
  \item[Case Base] We have $\tenvd = \emptyset$. The proof is immediate from the base definition of $\ttr$ function on the empty typing environment.
    
  \item[Case Induction] We assume the induction hypothesis $\beta\tstrend{\tenvd} = \bval[\ttr(\tenvd)]$ and we set to prove that $\bval\tstrend{\tenvd, x: \tvald} = \bval[\ttr(\tenvd, x: \tvald)]$. We proceed by case analysis on the structure of a backwards translatable type $\tvald$.
    
    \textit{Case $\tvald = \bot^d$}. By the definition we know that $\bval\tstrend{\tenvd, x: \bot^d} = \bval\tstrend{\tenvd} \bmeet \beta[x \leftarrow \bot^d]$. From the definition of type strengthening operator on types, and the definition of meet we get $\bval\tstrend{\tenvd, x: \bot^d} = \bbot$. Consider now the right side of the equation that after expanding the backwards type translation on typing environment is $\bval[\ttr(\tenvd), x: \ttr(\bot^d)]$. Using the translation definition for type $\bot^d$, and then using the definition of the base type strengthening with an environment we get $\beta[{\ttr(\tenvd), x: \bbot}] = \bbot$, thus concluding the proof of this case.
    
  \textit{Case $\tvald = \top^d$}. Following a similar approach, we expand the strengthening operator on the left side and get $\bval\tstrend{\tenvd, x: \top^d} = \bval\tstrend{\tenvd}$. By the definitions of the strengthening operator and the backwards translation of type $\top^d$, and by the fact that the concretization function is top-strict we obtain that $\beta[\ttr(\tenvd), x: \btop] = \beta[\ttr(\tenvd)]$. We finalize the proof of this case by applying the induction hypothesis.  

  \textit{Case $\tvald = \bval'$}. Expanding both sides of the equations and then applying the respective definitions we get the following proof obligation: $\bval\tstrend{\tenvd} \bmeet \bval[x \leftarrow \bval'] = \beta[\ttr(\tenvd), x: \bval']$.
  Using the definitions and the induction hypothesis it's easy to see that they agree in terms of their semantics given by their respective concretization functions because both types represent the most precise abstraction of the same set of concrete values and the infimum is unique. We use proposition \ref{prop:base-equality} we conclude the proof

  \textit{Case $\tvald = (x: \tvald_x) \to (\pivar: \eff) \to \tvald$}. The proof is immediate because the strengthening of a base type with a dependency variable bound to a function is similar to an identity operation. The argument for this case is the same as Case T-Top.
  
  \end{description}
\end{proof}

A similar result can be proved for strengthening a base type $\eff$ representing effect sequences. 

\begin{lemma}
  \label{lem:eff-strength}
  For an effect $\eff$, $\eff\tstrend{\tenvd} = \eff[\ttr(\tenvd)]$  
\end{lemma}
\begin{proof}
  Because $\eff$ are base refinement types drawn from the domain $\effdom$, the proof is similar to the proof for lemma \ref{lem:type-strength}.
\end{proof}

We ease the notation by dropping the superscript for the strengthening operator when it is clear from the context in which type system it is applied.

In the next lemma we show how that if we have a subtyping derivations in the target type system for backward translatable types and typing environments then we can obtain a derivation for subtyping in the type and effect system. As expected, the same result holds for effects with respect to their order relation.

\begin{lemma}
  \label{lem:subtype-sound}
  For all backward translatable $\tenvd, \tvald, \tvald'$, if $\tvald\tstren{\tenvd} \subtyped \tvald'\tstren{\tenvd}$ then $\ttr(\tenvd) \tejdg \ttr(\tvald) \subtype \ttr(\tvald')$
\end{lemma}
\begin{proof}
  The proof is by simultaneous induction over the depth of $\tvald$ and $\tvald'$. We case split on subtyping rules that apply to backward translatable types. 

  \begin{description}
  \item[Case \ruleref{s-bot[d]}] We have $\tvald = \bot^d$ and $\tvald'$ is any type, where the following must hold $\tvald' \ne \top^d$ . From the definition of translation of $\bot^d$ and the definition of strengthening operator, we get that that $\tvald\tstren{\tenvd} = \bot^d$. Then, by the definition of translation we get $\ttr(\tvald) = \bbot$. 
    Next, we consider the type structure of $\tvald'$. If $\tvald' = \bot^d$ or $\tvald' = \bval'$, then the proof is immediate as the ordering holds in the domain of basic refinement types and we can apply \ruleref{s-base}. 
    If it's a function type case, we must have $\tvald' = (x:\tvald_1) \to (\pivar:\eff_1) \to \tvald_2$, with its type translated backward to $\ttr(\tvald') = x:\tandeff{\ttr(\tvald_1)}{\eff} \to \tetr(\tvald_2)$. The rule is immediate following the premise that the function in the type and effect system is different from $\top$.
  \item[Case \ruleref{s-base[d]}] Follows immediately from lemma \ref{lem:type-strength}, reductive property of strengthening and rule \ruleref{s-base} 
    \item[Case \ruleref{s-fun[d]}] We can have $\tvald = (x:\tvald_1) \to (\pivar:\eff_1) \to \tvald_2$ and $\tvald' = (x:\tvald_1') \to (\pivar:\eff_1') \to \tvald_2'$. From subtyping rule premises we get the proof that input type and effect are contravariant, and output type is covariant. We must have $\tvald_1' \subtyped \tvald_1$, $\eff_1'\tstren{\tenvd,x:\tvald_1'} \subtyped \eff_1\tstren{\tenvd,x:\tvald_1'}$ and $\tvald_2\tstren{\tenvd, x:\tvald_1'} \subtyped \tvald_2'\tstren{\tenvd, x: \tvald_1}$. We know that $\tvald_2$ and $\tvald_2'$ are pair types and we can use the definition of strengthening that is applied component wise. Finally, we use the induction hypothesis and rule \ruleref{s-fun} to conclude the proof.
  \end{description}
\end{proof}

\begin{lemma}
  \label{lem:subeff-sound}
  For all backward translatable $\tenvd$, and effects $\eff, \eff'$, if $\eff\tstren{\tenvd} \subtyped \eff'\tstren{\tenvd}$ then $\eff[\ttr(\tenvd)] \efford \eff'[\ttr(\tenvd)]$
\end{lemma}
\begin{proof}
  The proof is similar to the proof for \ref{lem:subtype-sound}
\end{proof}

We are now ready to prove the theorem \ref{thm:inference-sound}
\begin{proof}
  The proof goes by structural induction. We show that we can inductively construct a derivation tree for a type and effect judgment of term $e$ from the derivation trees of its subterms. 
  We do this by considering each of the possible forms $e$ can have.
  We use $\tenv=\ttr(\tenvd)$ in the following. 
  \\  
  \begin{description} 
  \item[Case Const] We have term $e = c$ for some constant $c$ such that, given a term $\pivar$ with type $\eff$, the translated term is $\tr[\![c]\!](\pivar) \Def= \langle c, \pivar \rangle$ and $\tenvd,\pivar:\eff \tjdgd \langle c, \pivar \rangle: \tvald$.
    By the fact that the translated term is a pair we must have a proof for \ruleref{t-pair[d]}, and $\tvald = \tvald_1 \times \eff'$
    Using the typing rule \ruleref{t-proj[d]} we retrieve the proofs for individual types $\tenvd,\pivar:\eff \tjdgd c:\tvald_1$ and $\tenvd,\pivar:\eff \tjdgd \pivar:\eff'$ of each pair component.

    By rule \ruleref{t-const[d]} we know that the most precise type of $c$ is $\tstrend{\nu=c}\tstren{\tenvd, \pivar: \eff}$ and that $\tvald_1$ must be a base refinement type. Let $\tval = \ttr(\tvald_1)$. By the definition of the backwards translation for types, $\tval \in \bdom$.    
    From the definition of type semantics for constants we know that both $\tstrend{\nu = c}$ and $\{\nu = c\}_\bdom$ are the most precise approximations of $\{c\} \times Env$, therefore they must be equal.
    Additionally, we know that $\tstrend{\nu=c}\tstren{\tenvd, \pivar: \eff} = \tstrend{\nu=c}\tstren{\tenvd}$ because by definition the strengthening with a type information for effects leaves unchanged the type that is strengthened, and $\tstrend{\nu = c}\tstren{\tenvd} \subtyped \tvald_1$ from \ruleref{t-const[d]}. Then by lemma \ref{lem:subtype-sound} and rule \ruleref{s-base} we obtain that $\tenv \vdash \{\nu = c\}_\bdom \subtype \tval$. 
    
    Next, we move to derive the effect ordering. From the rule \ruleref{t-var[d]} we know that the following must hold $\eff\tstrend{\nu=\pivar}\tstrend{\tenvd, \pivar:\eff} \subtype \eff'\tstrend{\nu=\pivar}\tstrend{\tenvd, \pivar:\eff}$, that after eliminating the strengthening operation with $\eff$ is $\eff\tstrend{\nu=\pivar}\tstrend{\tenvd} \subtype \eff'\tstrend{\nu=\pivar}\tstrend{\tenvd}$.  
    By lemma \ref{lem:subeff-sound} together with the commutativity and reductive properties of the strengthening operator, we get $\eff\tstren{\tenv} \efford \eff\tstren{\nu=\pivar}\tstren{\tenv} \efford \eff'\tstren{\nu=\pivar}\tstren{\tenv} \efford \eff'$, giving us $\eff\tstren{\tenvd} \efford \eff'$. 

    By the reductive property of the strengthening operation we also have $\eff[\tenv] \efford \eff$
    
    We conclude the proof for $\tenv;\eff \vdash c: \tandeff{\tval}{\eff'}$ by using \ruleref{t-weaken} together with \ruleref{t-const} for constant $c$ instantiated with typing environment $\tenv$ and effect $\eff$, and using the proofs derived for $\eff[\tenv] \efford \eff$, $\tenv \vdash \{\nu = c\}_\bdom \subtype \tval$ and $\eff[\tenv] \efford \eff'$.   
    \\
    
  \item[Case Var] We have term $e = x$ for some $x$ such that, given a term $\pivar$ with type $\eff$, the translated term is $\tr[\![x]\!](\pivar) \Def= \langle x, \pivar \rangle$ and $\tenvd, \pivar:\eff \tjdgd \langle x, \pivar \rangle: \tvald$.
    By the fact that the translated term is a pair we must have a proof for \ruleref{t-pair[d]}, and $\tvald = \tvald_1 \times \eff'$
    We use again \ruleref{t-proj[d]} for getting $\tenvd,\pivar:\eff \tjdgd x : \tvald_1$ and $\tenvd,\pivar:\eff \tjdgd \pivar: \eff'$.
    
    From the rule \ruleref{t-var[d]} we know that the premise must hold $\tenvd(x)\tstren{\nu=x}\tstren{\tenvd, \pivar:\eff} \subtyped \tvald_1\tstren{\nu=x}\tstren{\tenvd, \pivar:\eff}$.
    We start by eliminating the strengthening with the $\eff$ type information in both sides of the subtyping relation.
    Let $\tval = \ttr(\tvald_1)$.
    Next, by lemma \ref{lem:subtype-sound} we get $\tenv \tejdg \tenv(x)\tstren{\nu=x} \subtype \tval[\nu=x]$. We have $\tenv(x)\tstren{\nu=x} = \tenv(x)$ because $x$ is a fresh variable and it doesn't restrict the values represented by the type variable $\nu$. By the reductive property of strengthening we have $\tenv \tejdg \tval[\nu=x] \subtype \tval$. Therefore, we get $\tenv \tejdg \tenv(x) \subtype \tval$. 
    We get the proof that $\eff\tstren{\tenv} \efford \eff'$ and $\eff\tstren{\tenv} \efford \eff$ the same we did for Case e-var.
    We use again \ruleref{t-weaken} with \ruleref{t-var} for variable $x$ instantiated with typing environment $\tenv$ and effect $\eff$, together with the proofs derived for
    $\eff[\tenv] \efford \eff$, $\tenv \tejdg \tenv(x) \subtype \tval$ and $\eff[\tenv] \efford \eff'$ to conclude the proof.
    \\
    
  \item[Case Abs] We have term $e = \lambda x. e_i$ such that, given a term $\pivar$ with type $\eff$, the translated term is $\tr[\![\lambda x. e_i]\!] = \langle \lambda x. \lambda \pivar_x. (\tr[\![e_i]\!](\pivar_x)), \pivar \rangle$, and $\tenvd, \pivar: \eff \tjdgd \langle \lambda x. \lambda \pivar_x. (\tr[\![e_i]\!](\pivar_x)), \pivar \rangle: \tvald$.
     From the rules \ruleref{t-pair[d]}, \ruleref{t-proj[d]} and \ruleref{t-abs[d]} used in derivation, and by induction hypothesis we get the derivations in the type and effect system for the following:
     \begin{itemize}
     \item $\tetr(\tvald) = \tandeff{\tval}{\eff'}$
     \item $\tenv \tejdg \lambda x.e_i : \tval$
     \item $\tenv, x: \tval_x;\; \eff_x \tejdg e_i : \tandeff{\tval_i}{\eff_i}$ 
     \item $\tenv \tejdg (x: \tandeff{\tval_x}{\eff_x} \to \tandeff{\tval_i}{\eff_i}) \subtype \tval$.
     \end{itemize}
     Then following the subtyping judgment it must be that $\tval = x: \tval_x'\&\eff_x' \to \tval_i'\&\eff_i'$. The proof follows immediately from using \ref{lem:subtype-sound} several times to construct the subtyping judgments in the type and effect system and lastly \ruleref{t-weaken} after constructing a derivation for the effect in the same spirit as for the constant and variable expressions.  
   \item[Case App] We have term $e = e_1~e_2$ such that, given a term $\pivar$ with type $\eff$, the translated term is after desugaring the $\termkw{let}$ constructs:
     \[\columnsep=5pt
       \begin{array}{l}
       (\lambda e_1'.(\\
        \;\;\lambda \pivar_1'.(\\
        \;\;\;\;\lambda e_2'.(\\
        \;\;\;\;\;\lambda \pivar_2'.(e_1'~e_2')~\pivar_2')\\
        \;\;\;\;\;(\#_2(\tr[\![e_2]\!](\pivar_1')))\\
        \;\;\;\;(\#_1(\tr[\![e_2]\!](\pivar_1'))))\\
        \;\;(\#_2(\tr[\![e_1])\!](\pivar))))\\
        \;(\#_1(\tr[\![e_1]\!](\pivar)))\\
     \end{array}\]
   We know that $\tenvd, \pivar:\eff \tjdgd \tr[\![e_1~e_2]\!](\pivar):\tvald$, and we must have used the rule \ruleref{t-app[d]}, therefore we know that $\tenvd, \pivar:\eff \tjdgd \tr[\![e_1]\!](\pivar): \tvald_1$ and $\tenvd, e_1': \tvald_1', \pivar_1':\eff_1 \tjdgd \tr[\![e_2]\!](\pivar_1'): \tvald_2$ and $\tvald_1 \subtyped x:\tvald_2 \to \tvald$.
   By induction hypothesis we know that $\tenvd, \pivar:\eff \tjdgd \tr[\![e_1]\!](\pivar): \tvald_1 \Rightarrow \tenv; \eff \tejdg e_1: \tetr(\tvald_1)$. Let $\tetr(\tvald_1) = \tval_1 \times \eff_1$.
   It is immediate to see that, if in the type and effect system the following holds $\tenv \vdash \tval_1 <: \tval_2$, then we can always prove that for any $x:\tval_x$, $\tenv, x:\tval_x$ by the monotonicity of strengthening operator.
   Using lemma \ref{lem:subtype-sound} we can construct from the subtyping derivation of $\lambda e_1'.\ldots$ we get from the premise of \ruleref{t-app[d]} a subtyping derivation. In our type and effect system the subtyping will hold in an empty typing environment. Because the strengthening operator is reductive, we can always introduce new type bindings in the context while maintaining the subtyping relation. We do that repeatedly until we can construct a proof for the subtyping with the same typing environment as used in for the expression $\#_1(\tr\den{e_1}(\pivar))$. Then we obtain the existential type by the reflexivity of subtyping followed by the derivation we can construct in the type and effect system using the introduction of a new type binding, and by the application of rule \ruleref{s-exists}.
   While long and cumbersome, the proof follows from applying the same strategy multiple times and then the rule \ruleref{t-app} to construct the typing derivation for the application. 
   \\
  \item[Case EV] The proof is immediate by using rules \ruleref{t-extension-op[d]} and the type operator $\odot$ provided by the abstract effect domain.
  \end{description}
  \end{proof}

\section{Trace-partitioning typing rules}
\label{sec:tp-details}

\paragraph{Original Drift type system} We present the original \drift{} subtyping and typing rules in Fig.~\ref{fig:drift-type-rules-ctx}.
We only present some of the interesting rule changes pertaining to callsite partitioning compared to Fig.~\ref{fig:drift-type-rules}.
The \drift{} type system, as described by \cite{DBLP:journals/pacmpl/PavlinovicSW21}, takes in as a parameter a finite set of \textit{abstract stacks} $\pstacks$, which represent abstractions of concrete call stacks.
These abstract stacks are equipped with an abstract concatenation operation $\pconcat : Loc \times \pstacks \to \pstacks$ that prepends a callsite location $i$ onto an abstract stack $\pstack$, denoted $i\pconcat\pstack$.
They further define function types, $x : \ttabled$, to be a mapping function from abstract stacks $\pstack \in \pstacks$ to dependent types.
Thus, $x : \ttabled$ essentially captures separate dependent function types $x : t_i \to t_o$ for every abstract stack $\pstack$.

Here, for a function type $t = x:\ttabled$, and $\pstack \in \pstacks$, $\ttabled(\pstack) = \pair{\tvald_i}{\tvald_o}$ for some $t_i$ and $t_o$.
Also, we denote by $t|_{\pstack}$ the function type $x : [\tentry{\pstack}{t_i}{t_o}]$ obtained from $t$ by restricting it to the call stack $\pstack$.
Some of the other typing operations are already defined earlier.

\begin{figure}
    \begin{mathpar}
      \inferH{s-fun[dcs]}
      {\ttabled_1(\pstack) = \pair{\tvald_{i1}}{\tvald_{o1}} \\ \ttabled_2(\pstack) = \pair{\tvald_{i2}}{\tvald_{o2}} \\\\ \tvald_{i2} <: \tvald_{i1} \\ \tvald_{o1}[x \mto \tvald_{i2}] <: \tvald_{o2}[x \mto \tvald_{i2}]}
      {x:\ttabled_1 \;<:\; x:\ttabled_2}\forall \pstack\in\pstacks
      \and
      \inferHtop{t-app[dcs]}
      {\tenvd,\pstack \vdash e_i : t_i \\ \tenvd,\pstack \vdash e_j : t_j \\ t_i <: x : [\tentry{i\pconcat\pstack}{t_j}{t}]}
      {\tenvd,\pstack \vdash e_i~e_j: \tvald}
      \and
      \inferHtop{t-abs[dcs]}
      {\tenvd_i = \tenvd.x : t_x \\ \tenvd_i,\pstack' \vdash e_i : t_i \\ x : [\tentry{\pstack'}{t_x}{t_i}] <: t|_{\pstack'}}
      {\tenvd,\pstack \vdash \lambda x.e_i : t} \forall \pstack' \in t
    \end{mathpar}
    \caption{\label{fig:drift-type-rules-ctx} Data flow refinement type system with callsite partitioning}  
  \end{figure}

In general, the typing judgement now takes the form $\tenvd, \pstack \vdash e : t$, i.e. every expression is also typed in the context of the abstract call stack.
The rule \ruleref{t-app[dcs]} for typing function applications $e_i~e_j$ requires that the type $t_i$ of $e_i$ must be a subtype of the function type $x : [\tentry{i\pconcat\pstack}{t_j}{t}]$ where $t_j$ is the type of the argument expression $e_j$ and $t$ is the result type of the function application.
Note that the rule extends the abstract stack $\pstack$ with the call site location $i$ identifying $e_i$.
The subtype relation then forces $t_i$ to have an appropriate entry for the abstract call stack $i\pconcat\pstack$.
The rule \ruleref{t-abs[dcs]} for typing lambda abstraction is as usual, except that it universally quantifies over all abstract stacks $\pstack'$ at which $t$ has been called.
The side condition $\forall\pstack' \in t$ implicitly constraints $t$ to be a function type.

\paragraph{Introducing if-then-else partitioning}
To extend the original drift type system with trace partitioning, we introduce a notion of \textit{abstract traces} $\ptraces$, which represents abstractions of concrete program traces, i.e. any subsequence (need not be contiguous) of program locations visited in a valid execution of a program.
Abstract traces, like abstract stacks are also equipped with an abstract concatenation operation $\pconcat : Loc \times \ptraces \to \ptraces$ that prepends a location $i$ onto an abstract trace $\ptrace$, denoted $i\pconcat\ptrace$.

We now say that every expression is typed in an \textit{abstract context}, a combination of an abstract stack and an abstract trace.
We redefine the typing judgement to take the form $\tenvd, (\pstack,\ptrace) \vdash e : r$, where $r\in \mathcal{R}$.
$\mathcal{R}$ is the set of \textit{trace-mappings} from abstract traces to refinement types.
Any expression may have nested if-then-else clauses, and according to whether those clauses evaluate to true or false, the expression may have different types.
A trace-mapping $r \in \mathcal{R}$, maps the different traces created by those clauses to different refinement types.
Thus, the typing judgement informs the possible types that an expression can have given the trace the program took up to the expression.
Further, we change the definition of output types for function types to be a trace-mapping $r$.
Note that we don't change input types to be trace-mappings, as any trace-mapping could be lifted to the function to have different abstract contexts.
Moreover, we also augment table types to include the environmental constraints the table type is valid in.

\begin{mathpar}
  t \in \cvalues ::= \bot^d ~\mid~ \top^d ~\mid~ \bval ~\mid~ x : \ttabled \qquad 
  \bval \in \mathcal{B} \qquad 
  r \in \mathcal{R} ::= \ptraces \rightarrow \mathcal{V} \\ 
  x : \ttabled \in \mathcal{L} ::= (\Sigma x \in Var. \pstacks\times\ptraces \rightarrow \cvalues \times \mathcal{R}) \times Env
\end{mathpar}

We define a function $\env(v) \in \cvalues \rightarrow (Var \rightarrow \cvalues)$, that extracts the environment $\cenvs$ from a refinement type $v$, as follows
\begin{equation*}
  \env(x) = \begin{cases}
    y : Var. \bot^d &x = \bot^d\\
    y : Var. \top^d &x = \top^d\\
    \cenv &x = \pair{x}{\cenv} \in \mathcal{B}\cup\mathcal{L}
  \end{cases}
\end{equation*}

\begin{figure}
  \begin{mathpar}
    \inferH{s-tr}
    {r_1(\ptrace)=t_1 \\ r_2(\ptrace)=t_2 \\ t_1 <: t_2}
    {r_1 <: r_2} \forall \ptrace \in \ptraces

    \inferH{t-cons[dcs+tr]}
    {[v=c]^d[\tenvd] <: t}
    {\tenvd,(\pstack,\ptrace) \vdash c : [\trentry{\ptrace}{t}]}

    \inferH{t-var[dcs+tr]}
    {\tenvd(x)[v=x]^d[\tenvd] <: t[v=x]^d[\tenvd]}
    {\tenvd,(\pstack,\ptrace) \vdash x : [\trentry{\ptrace}{t}]}

    \inferH{t-abs[dcs+tr]}
    {\tenvd_i = \tenvd.x : r_x = [\trentry{\ptrace'}{t_x}] \\ \tenvd_i,(\pstack',\ptrace') \vdash e_i : r_i \\ x : [\tentry{(\pstack',\ptrace')}{t_x}{r_i}] <: t|_{(\pstack',\ptrace')}}
    {\tenvd,(\pstack,\ptrace) \vdash \lambda x.e_i : [\trentry{\ptrace}{t}]} \forall (\pstack',\ptrace') \in t

    \inferH{t-app[dcs+tr]}
    {\tenvd,(\pstack,\ptrace) \vdash e_i : r_i \\ r_i(\ptrace_i)=t_i \\ \tenvd_i = \env(t_i) \\ \tenvd_i,(\pstack,\ptrace_i) \vdash e_j : r_j \\\\  \forall \ptrace_j\in r_j.r_j(\ptrace_j)=t_j. t_i <: x : [\tentry{(i\pconcat\pstack, \ptrace_j)}{t_j}{r_k}]. r_k <: r}
    {\tenvd,(\pstack,\ptrace) \vdash e_i~e_j: r} \forall \ptrace_i\in r_i

    \inferH{t-ite-true[dcs+tr]}
    {\tenvd,(\pstack,\ptrace) \vdash b <: [\trentry{\_}{bool}] \\ \tenvd_t = \tenvd[b=\textbf{true}] \\ \tenvd_t,(\pstack,\ptrace) \vdash e_i : r_i \\ \forall \ptrace_i\in r_i.r_i(\ptrace_i)=t_i. [\trentry{i\pconcat \ptrace_i}{t_i}] <: r }
    {\tenvd,(\pstack,\ptrace) \vdash \textbf{if}~b~\textbf{then}~e_i~\textbf{else}~e_j : r}
  
    \inferH{t-ite-false[dcs+tr]}
    {\tenvd,(\pstack,\ptrace) \vdash b <: [\trentry{\_}{bool}] \\ \tenvd_f = \tenvd[b=\textbf{false}] \\ \tenvd_f,(\pstack,\ptrace) \vdash e_j : r_j \\ \forall \ptrace_j\in r_j.r_j(\ptrace_j)=t_j. [\trentry{j\pconcat \ptrace_j}{t_j}] <: r }
    {\tenvd,(\pstack,\ptrace) \vdash \textbf{if}~b~\textbf{then}~e_i~\textbf{else}~e_j : r}

  \end{mathpar}
  \caption{\label{fig:drift-type-rules-tp} Data flow refinement type system with callsite and trace partitioning}  
\end{figure}

In figure~\ref{fig:drift-type-rules-tp} we give the typing rules for the extended type system, along with a subtyping relation for trace-mappings.
The subtyping rule \ruleref{s-tr} enforces subtyping over all entries in trace-mappings.
Compared to the rules in figure~\ref{fig:drift-type-rules}, the \ruleref{t-cons[dcs+tr]} and \ruleref{t-var[dcs+tr]} rules are largely the same, with the exception of the change in the typing judgement.
The trace-mapping only includes a single entry for the trace the respective expression is typed in, as these expressions don't create any new traces.
The rule \ruleref{t-abs[dcs+tr]} is also mostly similar to the \ruleref{t-abs[dcs]} rule.
The new rule also shows the changes in the table types.

The \ruleref{t-app[dcs+tr]} rule shows how abstract traces are threaded through for composite expressions.
The argument expression $e_j$ is typed under all the traces created by the function expression $e_i$, and the respective environment constraints captured using the refinement type $t_i$.
As for \ruleref{t-app[dcs]}, this rule similarly forces $t_i$ to have an entry for abstract contexts consisting of the new callsite $i$, and the abstract traces created by $e_i$ and then $e_j$.
Finally, all valid output mappings $r_k$ are subtypes of trace-mapping $r$.
This also captures the different traces created by the function body $e_i$ maps to.

Finally, the \ruleref{t-ite-true[dcs+tr]} and \ruleref{t-ite-false[dcs+tr]} rules show how new locations are concatenated to the abstract traces for if-then-else clauses.
Note that the conditional branches are typed under the same abstract context.
But the rule constraints the environments in which the respective conditional branches are evaluated.
The respective conditional branches might also be concatenating to the abstract trace $\ptrace$.
The final condition ensures that $r$ has entries for all the traces created by the conditional branch after concatenating them with the new location.
See that the two rules concatenate different locations to distinguish the two traces created by the clause.
Note that to simplify the presentation, we show a simplified version of these rules where the conditional clause doesn't crate any new traces.
One can take inspiration about how to handle traces created by the conditional branch from the \ruleref{t-app[dcs+tr]} to arrive at a more complete rule.

\section{CPS Translation of Example~\ref{ex:overview1}}
\label{apx:overview1-cps}

{\scriptsize
\begin{lstlisting}
let main prefx prefn = 
  let ev = fun k0 q acc evx ->
             if (q = 0) then k0 1 evx () 
             else if ((q = 1) && ((acc + evx) = 0)) then k0 2 acc () 
                  else if (q = 2) then k0 2 acc () 
                       else k0 q acc () in 
  let q1 = 0 in 
  let acc1 = 0 in 
  let f0 = fun k4 q3 acc3 busy ->
             let f1 = fun k6 q5 acc5 _main ->
                        let k8 q7 acc7 res3 =
                          let k7 q6 acc6 res2 =
                            k6 q6 acc6 res2 in 
                          res3 k7 q7 acc7 prefn in 
                        _main k8 q5 acc5 prefx in 
             let f2 = fun k9 q8 acc8 x ->
                        let f3 = fun k10 q9 acc9 n ->
                                   let x3 = () in 
                                   let k11 q10 acc10 x2 =
                                     let k13 q12 acc12 res5 =
                                       let k12 q11 acc11 res4 =
                                         let x1 = x3 ; res4 in  k10 q11 acc11 x1 in 
                                       res5 k12 q12 acc12 x in 
                                     busy k13 q10 acc10 n in 
                                   ev k11 q9 acc9 x in 
                        k9 q8 acc8 f3 in 
             let k5 q4 acc4 res1 =
               k4 q4 acc4 res1 in 
             f1 k5 q3 acc3 f2 in 
  let rec busy k14 q13 acc13 n =
    let f4 = fun k15 q14 acc14 t ->
               let x5 = 0 in 
               let x4 = n <= x5 in 
               let k16 q15 acc15 res6 =
                 k15 q15 acc15 res6 in 
               let k17 q16 acc16 res7 =
                 let x11 = - t in  let x10 = () in 
                                   let k21 q20 acc20 x9 =
                                     let x12 = 0 in 
                                     let x8 = x10 ; x12 in  k16 q20 acc20 x8 in 
                                   ev k21 q16 acc16 x11 in 
               let k18 q17 acc17 res8 =
                 let x7 = 1 in 
                 let x6 = n - x7 in  let k20 q19 acc19 res10 =
                                       let k19 q18 acc18 res9 =
                                         k16 q18 acc18 res9 in 
                                       res10 k19 q19 acc19 t in 
                                     busy k20 q17 acc17 x6 in 
               if x4 then k17 q14 acc14 x4 else k18 q14 acc14 x4 in 
    k14 q13 acc13 f4 in 
  let k3 q2 acc2 res0 =
    let k22 q acc x13 =
      let x15 = 2 in 
      let x14 = q = x15 in  assert(x14);x13 in 
    k22 q2 acc2 res0 in 
  f0 k3 q1 acc1 busy
\end{lstlisting}}


%


\section{Impact of Trace Partitioning}
\label{apx:tp-impact}

\renewcommand\humanCfgdirect[6]{$\langle tp\!:\!#3,  #5 \rangle$}
\renewcommand\humanCfgtrans[6]{$\langle tp\!:\!#3,  #5 \rangle$}

\begin{table}
	{\footnotesize\centering
		\renewcommand{\arraystretch}{\benchtablerowstretch}\setlength{\tabcolsep}{\benchtabletabcolsep}\footnotesize
		\begin{tabular}{l|crc|crc|crc}
			\toprule
			& \multicolumn{3}{c|}{\drift{} with trace part.} & \multicolumn{3}{c|}{\evdrift{} without trace part.} & \multicolumn{3}{c}{\evdrift{} with trace part.} \\ 
			{\bf Bench}  & {\bf Res} & {\bf CPU} & {\bf Config.} & {\bf Res} & {\bf CPU} & {\bf Config.} & {\bf Res} & {\bf CPU} & {\bf Config.} \\ 
			\midrule
			1. \texttt{\scriptsize all-ev-pos} & \Chk  & 1.6 & \humanCfgtrans{oopsla25july22}{1}{T}{T}{ls}{F} & \Chk  & 0.4 & \humanCfgdirect{oopsla25july22}{1}{F}{T}{ls}{T} & \Chk  & 0.4 & \humanCfgdirect{oopsla25july22}{1}{T}{T}{ls}{T} \\ 
2. \texttt{\scriptsize alt-inev} & \Unk  & 103.6 & \humanCfgtrans{oopsla25july22}{1}{T}{T}{ls}{F} & \Chk  & 5.4 & \humanCfgdirect{oopsla25july22}{1}{F}{T}{ls}{T} & \Chk  & 5.4 & \humanCfgdirect{oopsla25july22}{1}{T}{T}{ls}{T} \\ 
3. \texttt{\scriptsize auction} & \Unk  & 108.0 & \humanCfgtrans{oopsla25july22}{1}{T}{T}{ls}{F} & \Chk  & 4.1 & \humanCfgdirect{oopsla25july22}{1}{F}{T}{ls}{T} & \Chk  & 3.9 & \humanCfgdirect{oopsla25july22}{1}{T}{T}{ls}{F} \\ 
4. \texttt{\scriptsize binomial\_heap} & \Unk  & 238.1 & \humanCfgtrans{oopsla25july22}{1}{T}{T}{ls}{F} & \Chk  & 2.2 & \humanCfgdirect{oopsla25july22}{1}{F}{T}{ls}{F} & \Chk  & 2.2 & \humanCfgdirect{oopsla25july22}{1}{T}{T}{ls}{T} \\ 
5. \texttt{\scriptsize concurrent\_sum} & \Chk  & 11.8 & \humanCfgtrans{oopsla25july22}{1}{T}{T}{ls}{F} & \Chk  & 0.2 & \humanCfgdirect{oopsla25july22}{1}{F}{T}{ls}{T} & \Chk  & 0.8 & \humanCfgdirect{oopsla25july22}{1}{T}{T}{ls}{F} \\ 
6. \texttt{\scriptsize depend} & \Chk  & 0.0 & \humanCfgtrans{oopsla25july22}{1}{T}{T}{ls}{F} & \Chk  & 0.0 & \humanCfgdirect{oopsla25july22}{1}{F}{T}{ls}{T} & \Chk  & 0.0 & \humanCfgdirect{oopsla25july22}{1}{T}{T}{ls}{F} \\ 
7. \texttt{\scriptsize disj} & \Chk  & 137.1 & \humanCfgtrans{oopsla25july22}{1}{T}{F}{pg}{F} & \Chk  & 5.6 & \humanCfgdirect{oopsla25july22}{1}{F}{F}{pg}{F} & \Chk  & 10.3 & \humanCfgdirect{oopsla25july22}{1}{T}{F}{pg}{F} \\ 
8. \texttt{\scriptsize disj-gte} & \Chk  & 25.7 & \humanCfgtrans{oopsla25july22}{1}{T}{T}{ls}{F} & \Chk  & 2.8 & \humanCfgdirect{oopsla25july22}{1}{F}{T}{ls}{F} & \Chk  & 2.7 & \humanCfgdirect{oopsla25july22}{1}{T}{T}{ls}{T} \\ 
9. \texttt{\scriptsize disj-nondet} & \Chk  & 35.0 & \humanCfgtrans{oopsla25july22}{1}{T}{T}{ls}{F} & \Chk  & 2.9 & \humanCfgdirect{oopsla25july22}{1}{F}{T}{ls}{T} & \Chk  & 4.4 & \humanCfgdirect{oopsla25july22}{1}{T}{T}{ls}{T} \\ 
10. \texttt{\scriptsize higher-order} & \Chk  & 9.9 & \humanCfgtrans{oopsla25july22}{1}{T}{T}{ls}{F} & \Chk  & 0.4 & \humanCfgdirect{oopsla25july22}{1}{F}{T}{ls}{F} & \Chk  & 0.4 & \humanCfgdirect{oopsla25july22}{1}{T}{T}{ls}{F} \\ 
11. \texttt{\scriptsize intro-ord3} & \Chk  & 32.3 & \humanCfgtrans{oopsla25july22}{1}{T}{T}{ls}{F} & \Chk  & 6.0 & \humanCfgdirect{oopsla25july22}{1}{F}{T}{ls}{F} & \Chk  & 5.0 & \humanCfgdirect{oopsla25july22}{1}{T}{T}{ls}{F} \\ 
12. \texttt{\scriptsize last-ev-even} & \Chk  & 39.9 & \humanCfgtrans{oopsla25july22}{1}{T}{F}{pg}{F} & \Unk  & 1.0 & \humanCfgdirect{oopsla25july22}{1}{F}{T}{ls}{T} & \Chk  & 5.5 & \humanCfgdirect{oopsla25july22}{1}{T}{F}{pg}{T} \\ 
13. \texttt{\scriptsize lics18-amortized} & \Unk  & 587.4 & \humanCfgtrans{oopsla25july22}{1}{T}{T}{ls}{F} & \Chk  & 15.8 & \humanCfgdirect{oopsla25july22}{1}{F}{T}{ls}{T} & \Chk  & 15.8 & \humanCfgdirect{oopsla25july22}{1}{T}{T}{ls}{F} \\ 
14. \texttt{\scriptsize lics18-hoshrink} & \Unk  & 20.2 & \humanCfgtrans{oopsla25july22}{1}{T}{F}{pg}{F} & \Unk  & 1.0 & \humanCfgdirect{oopsla25july22}{1}{F}{T}{ls}{F} & \Unk  & 7.0 & \humanCfgdirect{oopsla25july22}{1}{T}{F}{pg}{T} \\ 
15. \texttt{\scriptsize lics18-web} & \TO   & 901.0 & \humanCfgtrans{oopsla25july22}{1}{T}{T}{ls}{F} & \Chk  & 23.8 & \humanCfgdirect{oopsla25july22}{1}{F}{T}{ls}{F} & \Chk  & 23.6 & \humanCfgdirect{oopsla25july22}{1}{T}{T}{ls}{F} \\ 
16. \texttt{\scriptsize market} & \TO   & 900.8 & \humanCfgtrans{oopsla25july22}{1}{T}{F}{pg}{F} & \Unk  & 35.4 & \humanCfgdirect{oopsla25july22}{1}{F}{F}{pg}{F} & \Unk  & 36.8 & \humanCfgdirect{oopsla25july22}{1}{T}{F}{pg}{F} \\ 
17. \texttt{\scriptsize max-min} & \TO   & 900.6 & \humanCfgtrans{oopsla25july22}{1}{T}{F}{pg}{F} & \Unk  & 28.2 & \humanCfgdirect{oopsla25july22}{1}{F}{F}{pg}{F} & \Chk  & 30.1 & \humanCfgdirect{oopsla25july22}{1}{T}{T}{ls}{T} \\ 
18. \texttt{\scriptsize monotonic} & \Chk  & 6.3 & \humanCfgtrans{oopsla25july22}{1}{T}{T}{ls}{F} & \Chk  & 0.5 & \humanCfgdirect{oopsla25july22}{1}{F}{T}{ls}{T} & \Chk  & 0.5 & \humanCfgdirect{oopsla25july22}{1}{T}{T}{ls}{T} \\ 
19. \texttt{\scriptsize nondet\_max} & \Chk  & 3.8 & \humanCfgtrans{oopsla25july22}{1}{T}{T}{ls}{F} & \Chk  & 1.0 & \humanCfgdirect{oopsla25july22}{1}{F}{T}{ls}{T} & \Chk  & 1.0 & \humanCfgdirect{oopsla25july22}{1}{T}{T}{ls}{T} \\ 
20. \texttt{\scriptsize order-irrel} & \Unk  & 36.4 & \humanCfgtrans{oopsla25july22}{1}{T}{F}{pg}{F} & \Unk  & 3.4 & \humanCfgdirect{oopsla25july22}{1}{F}{F}{pg}{F} & \Chk  & 3.5 & \humanCfgdirect{oopsla25july22}{1}{T}{F}{pg}{T} \\ 
21. \texttt{\scriptsize order-irrel-nondet} & \Unk  & 75.5 & \humanCfgtrans{oopsla25july22}{1}{T}{F}{pg}{F} & \Unk  & 10.3 & \humanCfgdirect{oopsla25july22}{1}{F}{F}{pg}{F} & \Chk  & 1.8 & \humanCfgdirect{oopsla25july22}{1}{T}{T}{ls}{T} \\ 
22. \texttt{\scriptsize overview1} & \Chk  & 2.3 & \humanCfgtrans{oopsla25july22}{1}{T}{T}{ls}{F} & \Chk  & 0.4 & \humanCfgdirect{oopsla25july22}{1}{F}{T}{ls}{T} & \Chk  & 0.5 & \humanCfgdirect{oopsla25july22}{1}{T}{T}{ls}{T} \\ 
23. \texttt{\scriptsize reentr} & \Chk  & 8.8 & \humanCfgtrans{oopsla25july22}{1}{T}{T}{ls}{F} & \Chk  & 0.2 & \humanCfgdirect{oopsla25july22}{1}{F}{T}{ls}{T} & \Chk  & 0.2 & \humanCfgdirect{oopsla25july22}{1}{T}{T}{ls}{F} \\ 
24. \texttt{\scriptsize resource-analysis} & \Chk  & 3.1 & \humanCfgtrans{oopsla25july22}{1}{T}{T}{ls}{F} & \Chk  & 0.3 & \humanCfgdirect{oopsla25july22}{1}{F}{T}{ls}{T} & \Chk  & 0.3 & \humanCfgdirect{oopsla25july22}{1}{T}{T}{ls}{F} \\ 
25. \texttt{\scriptsize sum-appendix} & \Chk  & 1.6 & \humanCfgtrans{oopsla25july22}{1}{T}{T}{ls}{F} & \Chk  & 0.0 & \humanCfgdirect{oopsla25july22}{1}{F}{T}{ls}{F} & \Chk  & 0.0 & \humanCfgdirect{oopsla25july22}{1}{T}{T}{ls}{F} \\ 
26. \texttt{\scriptsize sum-of-ev-even} & \Chk  & 6.2 & \humanCfgtrans{oopsla25july22}{1}{T}{F}{pg}{F} & \Chk  & 0.7 & \humanCfgdirect{oopsla25july22}{1}{F}{F}{pg}{F} & \Chk  & 0.7 & \humanCfgdirect{oopsla25july22}{1}{T}{F}{pg}{T} \\ 
27. \texttt{\scriptsize temperature} & \Chk  & 322.5 & \humanCfgtrans{oopsla25july22}{1}{T}{F}{pg}{F} & \Chk  & 12.0 & \humanCfgdirect{oopsla25july22}{1}{F}{F}{pg}{F} & \Chk  & 17.5 & \humanCfgdirect{oopsla25july22}{1}{T}{F}{pg}{F} \\ 

			\emph{geomean for \Chk's:} & & {\bf \expTPGMonDrift{}} & & & {\bf \expTPGMoffevDrift{}} & & & {\bf \expTPGMonevDrift{}} \\
			\bottomrule
		\end{tabular}
	}
\caption{\label{table:tp} Evaluating the impact of trace partitioning on \evdrift{}'s performance. The first set of columns in Tbl.~\ref{table:results} displayed \drift{}'s performance (via the tuple translation) \emph{without} trace partitioning, which had a geomean of \expTPGMoffDrift{}. In this table, column sets represent
(i) \drift{} \emph{with} trace partitioning with geomean \expTPGMonDrift,
(ii) \evdrift{} \emph{without} trace partitioning with geomean \expTPGMoffevDrift, and
(iii) \evdrift{} \emph{with} trace partitioning with geomean \expTPGMonevDrift. 
The colored highlighting shows how, even without trace partitioning, the abstract effect domain of \evdrift{} enables 7 benchmarks to be verified over drift, where as the additional improvement due to trace partitioning is a more modest 3 benchmarks.
Additionally, trace partitioning has a slowdown of \expTPSpeedupDrift{} for \drift{} and a slowdown of \expTPSpeedupevDrift{} for \evdrift{}.}
\end{table}

\paragraph{Analyzing the impact of trace-partitioning}
Since \evdrift{} uses our abstract effect domain \emph{and} trace partitioning, a natural question is: which feature provides the more substantial amount of improvement. We now compare \evdrift{} (and 1-context-sensitive \drift{} with the tuple translation) with and without trace partitioning to answer this question, and to quantify the speed overhead.
We summarize the results for these configurations in Table~\ref{table:tp}.
For \evdrift{}, the tool is able to verify two more benchmarks with trace partitioning.
Also note that trace partitioning helps \drift{} verify two additional benchmarks.
All these programs have if-then-else clauses that result in disjoint values for subsequent nodes.
However, \drift{} even with trace partitioning still does quite worse than \evdrift{}, as still verifies only 13 out of 26 benchmarks, compared to \evdrift{}'s 23 out of 26.

Moreover, for \evdrift{} there are several programs where the running times with trace partitioning are similar to the running times without trace partitioning.
This is because these benchmarks do not have nodes that follow an if-then-else clause, and hence there are no instances of multiplicatively analyzed nodes.
However, this doesn't happen for \drift{} due to the nature of the program translation.
Overall, trace partitioning slows \drift{} by \expTPSpeedupDrift{}$\times$, and slows \evdrift{} by \expTPSpeedupevDrift{}$\times$.

\section{Extended Evaluation}
\label{apx:extended}

\paragraph{Configurations} T=he following table lists all configurations we used for \drift{} and {\bf ev}-\drift{}:

\begin{center}
{\footnotesize
\begin{tabular}{|l|c|c|c|c|l|}
\hline
            & {\bf Trace} & {\bf Trace} & {\bf Thresh-} &           &  \\ 
{\bf Tool}  & {\bf Len.}  & {\bf Part.} & {\bf old}     & {\bf I/O} & {\bf Domain} \\ 
\hline
{\bf ev}-{\sc Drift} & 0 & false & true & true & Polka (Loose) \\
{\bf ev}-{\sc Drift} & 0 & false & false & true & PolkaGrid \\
{\bf ev}-{\sc Drift} & 0 & false & true & false & Polka (Loose) \\
{\bf ev}-{\sc Drift} & 0 & false & false & false & PolkaGrid \\
{\bf ev}-{\sc Drift} & 1 & true & true & true & Polka (Loose) \\
{\bf ev}-{\sc Drift} & 1 & true & false & true & PolkaGrid \\
{\bf ev}-{\sc Drift} & 1 & true & true & false & Polka (Loose) \\
{\bf ev}-{\sc Drift} & 1 & true & false & false & PolkaGrid \\
{\bf ev}-{\sc Drift} & 1 & false & true & true & Polka (Loose) \\
{\bf ev}-{\sc Drift} & 1 & false & false & true & PolkaGrid \\
{\bf ev}-{\sc Drift} & 1 & false & true & false & Polka (Loose) \\
{\bf ev}-{\sc Drift} & 1 & false & false & false & PolkaGrid \\
{\sc Drift} (tuple reduc.) & 0 & false & true & false & Polka (Loose) \\
{\sc Drift} (tuple reduc.) & 0 & false & false & false & PolkaGrid \\
{\sc Drift} (tuple reduc.) & 1 & true & true & false & Polka (Loose) \\
{\sc Drift} (tuple reduc.) & 1 & true & false & false & PolkaGrid \\
{\sc Drift} (tuple reduc.) & 1 & false & true & false & Polka (Loose) \\
{\sc Drift} (tuple reduc.) & 1 & false & false & false & PolkaGrid \\

\end{tabular}
}
\end{center}

\paragraph{Full experiments} The following table lists all results from running \rcaml, \drift\ and \evdrift\ on all configurations.

\renewcommand*{\hl}{\cellcolor{green!20}}

\renewcommand\humanCfgdirect[6]{$\langle tl\!:\!#2, tp\!:\!#3, th\!:\!#4, io\!:\!#6, #5 \rangle$}
\renewcommand\humanCfgtrans[6]{$\langle tl\!:\!#2, tp\!:\!#3, th\!:\!#4, io\!:\!#6, #5 \rangle$}

{\footnotesize
\renewcommand{\arraystretch}{\benchtablerowstretch}\setlength{\tabcolsep}{\benchtabletabcolsep}\footnotesize
\begin{longtable}{|l|l|c|c|r|}
\hline
{\bf Bench}  & {\bf Res} & {\bf CPU} & {\bf Config} \\ 
\hline
1. \texttt{\scriptsize all-ev-pos} \\
& \drift &      \humanCfgtrans{oopsla25july22}{1}{F}{F}{pg}{F} &      \Chk  &      6.26 \\
& \drift &      \humanCfgtrans{oopsla25july22}{1}{F}{T}{ls}{F} &      \Chk  &      1.53 \\
& \drift &      \humanCfgtrans{oopsla25july22}{0}{F}{T}{ls}{F} &      \Chk  &      0.61 \\
& \drift &      \humanCfgtrans{oopsla25july22}{1}{T}{F}{pg}{F} &      \Chk  &      6.89 \\
& \drift &      \humanCfgtrans{oopsla25july22}{1}{T}{T}{ls}{F} &      \Chk  &      1.60 \\
& \drift &      \humanCfgtrans{oopsla25july22}{0}{F}{F}{pg}{F} &      \Unk  &      2.20 \\
& \evdrift &      \humanCfgdirect{oopsla25july22}{1}{T}{T}{ls}{T} &      \Chk  &      0.38 \\
& \evdrift &      \humanCfgdirect{oopsla25july22}{0}{F}{F}{pg}{F} &      \Unk  &      0.29 \\
& \evdrift &      \humanCfgdirect{oopsla25july22}{1}{F}{T}{ls}{T} &      \Chk  &      0.36 \\
& \evdrift &      \humanCfgdirect{oopsla25july22}{1}{T}{F}{pg}{T} &      \Chk  &      1.02 \\
& \evdrift &      \humanCfgdirect{oopsla25july22}{1}{F}{F}{pg}{T} &      \Chk  &      0.86 \\
& \evdrift &      \humanCfgdirect{oopsla25july22}{0}{F}{F}{pg}{T} &      \Unk  &      0.32 \\
& \evdrift &      \humanCfgdirect{oopsla25july22}{1}{F}{F}{pg}{F} &      \Chk  &      0.87 \\
& \evdrift &      \humanCfgdirect{oopsla25july22}{1}{T}{T}{ls}{F} &      \Chk  &      0.39 \\
& \evdrift &      \humanCfgdirect{oopsla25july22}{0}{F}{T}{ls}{T} &      \Chk  &      0.18 \\
& \evdrift & \hl  \humanCfgdirect{oopsla25july22}{0}{F}{T}{ls}{F} & \hl  \Chk  & \hl  0.18 \\
& \evdrift &      \humanCfgdirect{oopsla25july22}{1}{T}{F}{pg}{F} &      \Chk  &      1.07 \\
& \evdrift &      \humanCfgdirect{oopsla25july22}{1}{F}{T}{ls}{F} &      \Chk  &      0.40 \\
\hline
2. \texttt{\scriptsize alt-inev} \\
& \drift &      \humanCfgtrans{oopsla25july22}{0}{F}{F}{pg}{F} &      \Unk  &      62.07 \\
& \drift &      \humanCfgtrans{oopsla25july22}{1}{T}{T}{ls}{F} &      \Unk  &      103.55 \\
& \drift &      \humanCfgtrans{oopsla25july22}{0}{F}{T}{ls}{F} &      \Unk  &      18.34 \\
& \drift &      \humanCfgtrans{oopsla25july22}{1}{F}{T}{ls}{F} &      \Unk  &      49.01 \\
& \drift &      \humanCfgtrans{oopsla25july22}{1}{T}{F}{pg}{F} &      \Unk  &      329.58 \\
& \drift &      \humanCfgtrans{oopsla25july22}{1}{F}{F}{pg}{F} &      \Unk  &      200.11 \\
& \evdrift &      \humanCfgdirect{oopsla25july22}{1}{T}{F}{pg}{F} &      \Chk  &      34.53 \\
& \evdrift &      \humanCfgdirect{oopsla25july22}{0}{F}{T}{ls}{F} &      \Chk  &      1.81 \\
& \evdrift &      \humanCfgdirect{oopsla25july22}{1}{F}{T}{ls}{F} &      \Chk  &      5.49 \\
& \evdrift & \hl  \humanCfgdirect{oopsla25july22}{0}{F}{T}{ls}{T} & \hl  \Chk  & \hl  1.78 \\
& \evdrift &      \humanCfgdirect{oopsla25july22}{1}{F}{F}{pg}{F} &      \Chk  &      33.70 \\
& \evdrift &      \humanCfgdirect{oopsla25july22}{1}{T}{T}{ls}{F} &      \Chk  &      5.50 \\
& \evdrift &      \humanCfgdirect{oopsla25july22}{1}{F}{F}{pg}{T} &      \Chk  &      33.82 \\
& \evdrift &      \humanCfgdirect{oopsla25july22}{0}{F}{F}{pg}{T} &      \Chk  &      8.23 \\
& \evdrift &      \humanCfgdirect{oopsla25july22}{1}{F}{T}{ls}{T} &      \Chk  &      5.42 \\
& \evdrift &      \humanCfgdirect{oopsla25july22}{1}{T}{F}{pg}{T} &      \Chk  &      34.46 \\
& \evdrift &      \humanCfgdirect{oopsla25july22}{0}{F}{F}{pg}{F} &      \Chk  &      8.31 \\
& \evdrift &      \humanCfgdirect{oopsla25july22}{1}{T}{T}{ls}{T} &      \Chk  &      5.40 \\
\hline
3. \texttt{\scriptsize auction} \\
& \drift &      \humanCfgtrans{oopsla25july22}{0}{F}{F}{pg}{F} &      \Unk  &      88.00 \\
& \drift &      \humanCfgtrans{oopsla25july22}{1}{T}{T}{ls}{F} &      \Unk  &      108.05 \\
& \drift &      \humanCfgtrans{oopsla25july22}{1}{F}{F}{pg}{F} &      \Unk  &      155.73 \\
& \drift &      \humanCfgtrans{oopsla25july22}{1}{T}{F}{pg}{F} &      \Unk  &      307.68 \\
& \drift &      \humanCfgtrans{oopsla25july22}{1}{F}{T}{ls}{F} &      \Unk  &      53.77 \\
& \drift &      \humanCfgtrans{oopsla25july22}{0}{F}{T}{ls}{F} &      \Unk  &      43.51 \\
& \evdrift &      \humanCfgdirect{oopsla25july22}{0}{F}{T}{ls}{T} &      \Chk  &      2.48 \\
& \evdrift &      \humanCfgdirect{oopsla25july22}{1}{F}{T}{ls}{F} &      \Chk  &      4.18 \\
& \evdrift & \hl  \humanCfgdirect{oopsla25july22}{0}{F}{T}{ls}{F} & \hl  \Chk  & \hl  2.46 \\
& \evdrift &      \humanCfgdirect{oopsla25july22}{1}{T}{F}{pg}{F} &      \Chk  &      7.64 \\
& \evdrift &      \humanCfgdirect{oopsla25july22}{1}{F}{F}{pg}{F} &      \Chk  &      7.58 \\
& \evdrift &      \humanCfgdirect{oopsla25july22}{1}{T}{T}{ls}{F} &      \Chk  &      3.90 \\
& \evdrift &      \humanCfgdirect{oopsla25july22}{1}{T}{F}{pg}{T} &      \Chk  &      7.00 \\
& \evdrift &      \humanCfgdirect{oopsla25july22}{1}{F}{T}{ls}{T} &      \Chk  &      4.12 \\
& \evdrift &      \humanCfgdirect{oopsla25july22}{0}{F}{F}{pg}{T} &      \Chk  &      3.41 \\
& \evdrift &      \humanCfgdirect{oopsla25july22}{1}{F}{F}{pg}{T} &      \Chk  &      7.56 \\
& \evdrift &      \humanCfgdirect{oopsla25july22}{1}{T}{T}{ls}{T} &      \Chk  &      3.91 \\
& \evdrift &      \humanCfgdirect{oopsla25july22}{0}{F}{F}{pg}{F} &      \Chk  &      3.27 \\
\hline
4. \texttt{\scriptsize binomial\_heap} \\
& \drift &      \humanCfgtrans{oopsla25july22}{1}{F}{F}{pg}{F} &      \Unk  &      572.02 \\
& \drift &      \humanCfgtrans{oopsla25july22}{0}{F}{T}{ls}{F} &      \Unk  &      107.17 \\
& \drift &      \humanCfgtrans{oopsla25july22}{1}{F}{T}{ls}{F} &      \Unk  &      198.54 \\
& \drift &      \humanCfgtrans{oopsla25july22}{1}{T}{F}{pg}{F} &      \Unk  &      669.60 \\
& \drift &      \humanCfgtrans{oopsla25july22}{1}{T}{T}{ls}{F} &      \Unk  &      238.11 \\
& \drift &      \humanCfgtrans{oopsla25july22}{0}{F}{F}{pg}{F} &      \Unk  &      446.67 \\
& \evdrift &      \humanCfgdirect{oopsla25july22}{0}{F}{F}{pg}{F} &      \Chk  &      2.46 \\
& \evdrift &      \humanCfgdirect{oopsla25july22}{1}{T}{T}{ls}{T} &      \Chk  &      2.20 \\
& \evdrift &      \humanCfgdirect{oopsla25july22}{1}{F}{T}{ls}{T} &      \Chk  &      2.22 \\
& \evdrift &      \humanCfgdirect{oopsla25july22}{1}{T}{F}{pg}{T} &      \Chk  &      3.38 \\
& \evdrift &      \humanCfgdirect{oopsla25july22}{1}{F}{F}{pg}{T} &      \Chk  &      3.32 \\
& \evdrift &      \humanCfgdirect{oopsla25july22}{0}{F}{F}{pg}{T} &      \Chk  &      2.44 \\
& \evdrift &      \humanCfgdirect{oopsla25july22}{1}{T}{T}{ls}{F} &      \Chk  &      2.22 \\
& \evdrift &      \humanCfgdirect{oopsla25july22}{1}{F}{F}{pg}{F} &      \Chk  &      3.39 \\
& \evdrift & \hl  \humanCfgdirect{oopsla25july22}{0}{F}{T}{ls}{T} & \hl  \Chk  & \hl  1.60 \\
& \evdrift &      \humanCfgdirect{oopsla25july22}{1}{T}{F}{pg}{F} &      \Chk  &      3.34 \\
& \evdrift &      \humanCfgdirect{oopsla25july22}{0}{F}{T}{ls}{F} &      \Chk  &      1.61 \\
& \evdrift &      \humanCfgdirect{oopsla25july22}{1}{F}{T}{ls}{F} &      \Chk  &      2.22 \\
\hline
5. \texttt{\scriptsize concurrent\_sum} \\
& \drift &      \humanCfgtrans{oopsla25july22}{1}{T}{T}{ls}{F} &      \Chk  &      11.83 \\
& \drift &      \humanCfgtrans{oopsla25july22}{0}{F}{F}{pg}{F} &      \Chk  &      3.45 \\
& \drift &      \humanCfgtrans{oopsla25july22}{1}{F}{F}{pg}{F} &      \Chk  &      8.45 \\
& \drift &      \humanCfgtrans{oopsla25july22}{1}{F}{T}{ls}{F} &      \Chk  &      3.39 \\
& \drift &      \humanCfgtrans{oopsla25july22}{0}{F}{T}{ls}{F} &      \Chk  &      1.51 \\
& \drift &      \humanCfgtrans{oopsla25july22}{1}{T}{F}{pg}{F} &      \Chk  &      25.44 \\
& \evdrift &      \humanCfgdirect{oopsla25july22}{1}{T}{T}{ls}{F} &      \Chk  &      0.77 \\
& \evdrift &      \humanCfgdirect{oopsla25july22}{1}{F}{F}{pg}{F} &      \Chk  &      0.46 \\
& \evdrift &      \humanCfgdirect{oopsla25july22}{0}{F}{T}{ls}{T} &      \Chk  &      0.17 \\
& \evdrift & \hl  \humanCfgdirect{oopsla25july22}{0}{F}{T}{ls}{F} & \hl  \Chk  & \hl  0.17 \\
& \evdrift &      \humanCfgdirect{oopsla25july22}{1}{T}{F}{pg}{F} &      \Chk  &      1.73 \\
& \evdrift &      \humanCfgdirect{oopsla25july22}{1}{F}{T}{ls}{F} &      \Chk  &      0.27 \\
& \evdrift &      \humanCfgdirect{oopsla25july22}{1}{T}{T}{ls}{T} &      \Chk  &      0.77 \\
& \evdrift &      \humanCfgdirect{oopsla25july22}{0}{F}{F}{pg}{F} &      \Chk  &      0.21 \\
& \evdrift &      \humanCfgdirect{oopsla25july22}{1}{F}{T}{ls}{T} &      \Chk  &      0.25 \\
& \evdrift &      \humanCfgdirect{oopsla25july22}{1}{T}{F}{pg}{T} &      \Chk  &      1.61 \\
& \evdrift &      \humanCfgdirect{oopsla25july22}{1}{F}{F}{pg}{T} &      \Chk  &      0.46 \\
& \evdrift &      \humanCfgdirect{oopsla25july22}{0}{F}{F}{pg}{T} &      \Chk  &      0.22 \\
\hline
6. \texttt{\scriptsize depend} \\
& \drift &      \humanCfgtrans{oopsla25july22}{0}{F}{T}{ls}{F} &      \Chk  &      0.04 \\
& \drift &      \humanCfgtrans{oopsla25july22}{1}{F}{T}{ls}{F} &      \Chk  &      0.04 \\
& \drift &      \humanCfgtrans{oopsla25july22}{1}{T}{F}{pg}{F} &      \Chk  &      0.10 \\
& \drift &      \humanCfgtrans{oopsla25july22}{1}{F}{F}{pg}{F} &      \Chk  &      0.10 \\
& \drift &      \humanCfgtrans{oopsla25july22}{0}{F}{F}{pg}{F} &      \Chk  &      0.10 \\
& \drift &      \humanCfgtrans{oopsla25july22}{1}{T}{T}{ls}{F} &      \Chk  &      0.04 \\
& \evdrift &      \humanCfgdirect{oopsla25july22}{1}{F}{F}{pg}{T} &      \Chk  &      0.03 \\
& \evdrift &      \humanCfgdirect{oopsla25july22}{0}{F}{F}{pg}{T} &      \Chk  &      0.03 \\
& \evdrift &      \humanCfgdirect{oopsla25july22}{1}{F}{T}{ls}{T} &      \Chk  &      0.02 \\
& \evdrift &      \humanCfgdirect{oopsla25july22}{1}{T}{F}{pg}{T} &      \Chk  &      0.03 \\
& \evdrift &      \humanCfgdirect{oopsla25july22}{0}{F}{F}{pg}{F} &      \Chk  &      0.03 \\
& \evdrift &      \humanCfgdirect{oopsla25july22}{1}{T}{T}{ls}{T} &      \Chk  &      0.02 \\
& \evdrift &      \humanCfgdirect{oopsla25july22}{1}{T}{F}{pg}{F} &      \Chk  &      0.03 \\
& \evdrift &      \humanCfgdirect{oopsla25july22}{0}{F}{T}{ls}{F} &      \Chk  &      0.02 \\
& \evdrift &      \humanCfgdirect{oopsla25july22}{1}{F}{T}{ls}{F} &      \Chk  &      0.02 \\
& \evdrift &      \humanCfgdirect{oopsla25july22}{0}{F}{T}{ls}{T} &      \Chk  &      0.02 \\
& \evdrift &      \humanCfgdirect{oopsla25july22}{1}{F}{F}{pg}{F} &      \Chk  &      0.03 \\
& \evdrift & \hl  \humanCfgdirect{oopsla25july22}{1}{T}{T}{ls}{F} & \hl  \Chk  & \hl  0.02 \\
\hline
7. \texttt{\scriptsize disj} \\
& \drift &      \humanCfgtrans{oopsla25july22}{1}{T}{T}{ls}{F} &      \Unk  &      102.69 \\
& \drift &      \humanCfgtrans{oopsla25july22}{0}{F}{F}{pg}{F} &      \Unk  &      36.36 \\
& \drift &      \humanCfgtrans{oopsla25july22}{1}{F}{F}{pg}{F} &      \Unk  &      60.87 \\
& \drift &      \humanCfgtrans{oopsla25july22}{0}{F}{T}{ls}{F} &      \Unk  &      7.46 \\
& \drift &      \humanCfgtrans{oopsla25july22}{1}{F}{T}{ls}{F} &      \Unk  &      11.04 \\
& \drift &      \humanCfgtrans{oopsla25july22}{1}{T}{F}{pg}{F} &      \Chk  &      137.05 \\
& \evdrift &      \humanCfgdirect{oopsla25july22}{1}{T}{T}{ls}{F} &      \Unk  &      3.76 \\
& \evdrift &      \humanCfgdirect{oopsla25july22}{1}{F}{F}{pg}{F} &      \Chk  &      5.63 \\
& \evdrift &      \humanCfgdirect{oopsla25july22}{0}{F}{T}{ls}{T} &      \Unk  &      1.74 \\
& \evdrift &      \humanCfgdirect{oopsla25july22}{1}{T}{F}{pg}{F} &      \Chk  &      10.27 \\
& \evdrift &      \humanCfgdirect{oopsla25july22}{0}{F}{T}{ls}{F} &      \Unk  &      1.78 \\
& \evdrift &      \humanCfgdirect{oopsla25july22}{1}{F}{T}{ls}{F} &      \Unk  &      2.31 \\
& \evdrift & \hl  \humanCfgdirect{oopsla25july22}{0}{F}{F}{pg}{F} & \hl  \Chk  & \hl  5.33 \\
& \evdrift &      \humanCfgdirect{oopsla25july22}{1}{T}{T}{ls}{T} &      \Unk  &      3.91 \\
& \evdrift &      \humanCfgdirect{oopsla25july22}{1}{F}{T}{ls}{T} &      \Unk  &      2.32 \\
& \evdrift &      \humanCfgdirect{oopsla25july22}{1}{T}{F}{pg}{T} &      \Chk  &      10.34 \\
& \evdrift &      \humanCfgdirect{oopsla25july22}{1}{F}{F}{pg}{T} &      \Chk  &      6.57 \\
& \evdrift &      \humanCfgdirect{oopsla25july22}{0}{F}{F}{pg}{T} &      \Chk  &      6.18 \\
\hline
8. \texttt{\scriptsize disj-gte} \\
& \drift &      \humanCfgtrans{oopsla25july22}{1}{F}{T}{ls}{F} &      \Unk  &      13.60 \\
& \drift &      \humanCfgtrans{oopsla25july22}{0}{F}{T}{ls}{F} &      \Unk  &      8.07 \\
& \drift &      \humanCfgtrans{oopsla25july22}{1}{T}{F}{pg}{F} &      \Chk  &      129.83 \\
& \drift &      \humanCfgtrans{oopsla25july22}{1}{F}{F}{pg}{F} &      \Unk  &      94.92 \\
& \drift &      \humanCfgtrans{oopsla25july22}{0}{F}{F}{pg}{F} &      \Unk  &      34.29 \\
& \drift &      \humanCfgtrans{oopsla25july22}{1}{T}{T}{ls}{F} &      \Chk  &      25.67 \\
& \evdrift &      \humanCfgdirect{oopsla25july22}{1}{F}{F}{pg}{T} &      \Chk  &      8.37 \\
& \evdrift &      \humanCfgdirect{oopsla25july22}{0}{F}{F}{pg}{T} &      \Chk  &      9.06 \\
& \evdrift &      \humanCfgdirect{oopsla25july22}{1}{F}{T}{ls}{T} &      \Chk  &      2.76 \\
& \evdrift &      \humanCfgdirect{oopsla25july22}{1}{T}{F}{pg}{T} &      \Chk  &      7.76 \\
& \evdrift &      \humanCfgdirect{oopsla25july22}{1}{T}{T}{ls}{T} &      \Chk  &      2.68 \\
& \evdrift &      \humanCfgdirect{oopsla25july22}{0}{F}{F}{pg}{F} &      \Chk  &      8.93 \\
& \evdrift & \hl  \humanCfgdirect{oopsla25july22}{0}{F}{T}{ls}{F} & \hl  \Chk  & \hl  1.97 \\
& \evdrift &      \humanCfgdirect{oopsla25july22}{1}{T}{F}{pg}{F} &      \Chk  &      8.54 \\
& \evdrift &      \humanCfgdirect{oopsla25july22}{1}{F}{T}{ls}{F} &      \Chk  &      2.76 \\
& \evdrift &      \humanCfgdirect{oopsla25july22}{0}{F}{T}{ls}{T} &      \Chk  &      2.02 \\
& \evdrift &      \humanCfgdirect{oopsla25july22}{1}{F}{F}{pg}{F} &      \Chk  &      8.36 \\
& \evdrift &      \humanCfgdirect{oopsla25july22}{1}{T}{T}{ls}{F} &      \Chk  &      2.69 \\
\hline
9. \texttt{\scriptsize disj-nondet} \\
& \drift &      \humanCfgtrans{oopsla25july22}{1}{F}{F}{pg}{F} &      \Unk  &      101.79 \\
& \drift &      \humanCfgtrans{oopsla25july22}{1}{T}{F}{pg}{F} &      \Chk  &      252.05 \\
& \drift &      \humanCfgtrans{oopsla25july22}{0}{F}{T}{ls}{F} &      \Unk  &      10.24 \\
& \drift &      \humanCfgtrans{oopsla25july22}{1}{F}{T}{ls}{F} &      \Unk  &      12.33 \\
& \drift &      \humanCfgtrans{oopsla25july22}{1}{T}{T}{ls}{F} &      \Chk  &      35.00 \\
& \drift &      \humanCfgtrans{oopsla25july22}{0}{F}{F}{pg}{F} &      \Unk  &      66.89 \\
& \evdrift &      \humanCfgdirect{oopsla25july22}{0}{F}{F}{pg}{F} &      \Chk  &      15.59 \\
& \evdrift &      \humanCfgdirect{oopsla25july22}{1}{T}{T}{ls}{T} &      \Chk  &      4.39 \\
& \evdrift &      \humanCfgdirect{oopsla25july22}{1}{T}{F}{pg}{T} &      \Chk  &      20.61 \\
& \evdrift &      \humanCfgdirect{oopsla25july22}{1}{F}{T}{ls}{T} &      \Chk  &      2.90 \\
& \evdrift &      \humanCfgdirect{oopsla25july22}{0}{F}{F}{pg}{T} &      \Chk  &      15.69 \\
& \evdrift &      \humanCfgdirect{oopsla25july22}{1}{F}{F}{pg}{T} &      \Chk  &      14.44 \\
& \evdrift &      \humanCfgdirect{oopsla25july22}{1}{T}{T}{ls}{F} &      \Chk  &      4.43 \\
& \evdrift &      \humanCfgdirect{oopsla25july22}{1}{F}{F}{pg}{F} &      \Chk  &      14.38 \\
& \evdrift &      \humanCfgdirect{oopsla25july22}{0}{F}{T}{ls}{T} &      \Chk  &      2.35 \\
& \evdrift &      \humanCfgdirect{oopsla25july22}{1}{F}{T}{ls}{F} &      \Chk  &      2.90 \\
& \evdrift &      \humanCfgdirect{oopsla25july22}{1}{T}{F}{pg}{F} &      \Chk  &      20.45 \\
& \evdrift & \hl  \humanCfgdirect{oopsla25july22}{0}{F}{T}{ls}{F} & \hl  \Chk  & \hl  2.35 \\
\hline
10. \texttt{\scriptsize higher-order} \\
& \drift &      \humanCfgtrans{oopsla25july22}{1}{T}{F}{pg}{F} &      \Chk  &      23.05 \\
& \drift &      \humanCfgtrans{oopsla25july22}{0}{F}{T}{ls}{F} &      \Chk  &      1.88 \\
& \drift &      \humanCfgtrans{oopsla25july22}{1}{F}{T}{ls}{F} &      \Chk  &      5.12 \\
& \drift &      \humanCfgtrans{oopsla25july22}{1}{F}{F}{pg}{F} &      \Unk  &      18.46 \\
& \drift &      \humanCfgtrans{oopsla25july22}{1}{T}{T}{ls}{F} &      \Chk  &      9.88 \\
& \drift &      \humanCfgtrans{oopsla25july22}{0}{F}{F}{pg}{F} &      \Unk  &      9.06 \\
& \evdrift &      \humanCfgdirect{oopsla25july22}{0}{F}{F}{pg}{F} &      \Chk  &      1.09 \\
& \evdrift &      \humanCfgdirect{oopsla25july22}{1}{T}{T}{ls}{T} &      \Chk  &      0.52 \\
& \evdrift &      \humanCfgdirect{oopsla25july22}{0}{F}{F}{pg}{T} &      \Chk  &      1.06 \\
& \evdrift &      \humanCfgdirect{oopsla25july22}{1}{F}{F}{pg}{T} &      \Chk  &      2.36 \\
& \evdrift &      \humanCfgdirect{oopsla25july22}{1}{T}{F}{pg}{T} &      \Chk  &      2.35 \\
& \evdrift &      \humanCfgdirect{oopsla25july22}{1}{F}{T}{ls}{T} &      \Chk  &      0.49 \\
& \evdrift &      \humanCfgdirect{oopsla25july22}{1}{T}{T}{ls}{F} &      \Chk  &      0.42 \\
& \evdrift &      \humanCfgdirect{oopsla25july22}{1}{F}{F}{pg}{F} &      \Chk  &      2.37 \\
& \evdrift &      \humanCfgdirect{oopsla25july22}{1}{F}{T}{ls}{F} &      \Chk  &      0.41 \\
& \evdrift &      \humanCfgdirect{oopsla25july22}{1}{T}{F}{pg}{F} &      \Chk  &      2.39 \\
& \evdrift &      \humanCfgdirect{oopsla25july22}{0}{F}{T}{ls}{F} &      \Chk  &      0.38 \\
& \evdrift & \hl  \humanCfgdirect{oopsla25july22}{0}{F}{T}{ls}{T} & \hl  \Chk  & \hl  0.35 \\
\hline
11. \texttt{\scriptsize intro-ord3} \\
& \drift &      \humanCfgtrans{oopsla25july22}{0}{F}{F}{pg}{F} &      \Unk  &      70.86 \\
& \drift &      \humanCfgtrans{oopsla25july22}{1}{T}{T}{ls}{F} &      \Chk  &      32.26 \\
& \drift &      \humanCfgtrans{oopsla25july22}{1}{F}{F}{pg}{F} &      \Chk  &      152.38 \\
& \drift &      \humanCfgtrans{oopsla25july22}{1}{T}{F}{pg}{F} &      \Chk  &      181.92 \\
& \drift &      \humanCfgtrans{oopsla25july22}{0}{F}{T}{ls}{F} &      \Unk  &      8.37 \\
& \drift &      \humanCfgtrans{oopsla25july22}{1}{F}{T}{ls}{F} &      \Chk  &      24.14 \\
& \evdrift & \hl  \humanCfgdirect{oopsla25july22}{0}{F}{T}{ls}{T} & \hl  \Chk  & \hl  2.59 \\
& \evdrift &      \humanCfgdirect{oopsla25july22}{1}{F}{T}{ls}{F} &      \Chk  &      6.02 \\
& \evdrift &      \humanCfgdirect{oopsla25july22}{1}{T}{F}{pg}{F} &      \Chk  &      40.09 \\
& \evdrift &      \humanCfgdirect{oopsla25july22}{0}{F}{T}{ls}{F} &      \Chk  &      2.60 \\
& \evdrift &      \humanCfgdirect{oopsla25july22}{1}{F}{F}{pg}{F} &      \Chk  &      39.83 \\
& \evdrift &      \humanCfgdirect{oopsla25july22}{1}{T}{T}{ls}{F} &      \Chk  &      4.97 \\
& \evdrift &      \humanCfgdirect{oopsla25july22}{1}{T}{F}{pg}{T} &      \Chk  &      40.40 \\
& \evdrift &      \humanCfgdirect{oopsla25july22}{1}{F}{T}{ls}{T} &      \Chk  &      6.03 \\
& \evdrift &      \humanCfgdirect{oopsla25july22}{0}{F}{F}{pg}{T} &      \Chk  &      13.46 \\
& \evdrift &      \humanCfgdirect{oopsla25july22}{1}{F}{F}{pg}{T} &      \Chk  &      40.38 \\
& \evdrift &      \humanCfgdirect{oopsla25july22}{0}{F}{F}{pg}{F} &      \Chk  &      13.34 \\
& \evdrift &      \humanCfgdirect{oopsla25july22}{1}{T}{T}{ls}{T} &      \Chk  &      6.17 \\
\hline
12. \texttt{\scriptsize last-ev-even} \\
& \drift &      \humanCfgtrans{oopsla25july22}{1}{F}{F}{pg}{F} &      \Unk  &      17.43 \\
& \drift &      \humanCfgtrans{oopsla25july22}{0}{F}{T}{ls}{F} &      \Unk  &      1.48 \\
& \drift &      \humanCfgtrans{oopsla25july22}{1}{F}{T}{ls}{F} &      \Unk  &      6.21 \\
& \drift &      \humanCfgtrans{oopsla25july22}{1}{T}{F}{pg}{F} &      \Chk  &      39.86 \\
& \drift &      \humanCfgtrans{oopsla25july22}{1}{T}{T}{ls}{F} &      \Unk  &      12.52 \\
& \drift &      \humanCfgtrans{oopsla25july22}{0}{F}{F}{pg}{F} &      \Unk  &      7.16 \\
& \evdrift &      \humanCfgdirect{oopsla25july22}{0}{F}{F}{pg}{F} &      \Unk  &      1.05 \\
& \evdrift &      \humanCfgdirect{oopsla25july22}{1}{T}{T}{ls}{T} &      \Unk  &      2.19 \\
& \evdrift &      \humanCfgdirect{oopsla25july22}{1}{F}{T}{ls}{T} &      \Unk  &      0.97 \\
& \evdrift & \hl  \humanCfgdirect{oopsla25july22}{1}{T}{F}{pg}{T} & \hl  \Chk  & \hl  5.52 \\
& \evdrift &      \humanCfgdirect{oopsla25july22}{1}{F}{F}{pg}{T} &      \Unk  &      2.55 \\
& \evdrift &      \humanCfgdirect{oopsla25july22}{0}{F}{F}{pg}{T} &      \Unk  &      1.00 \\
& \evdrift &      \humanCfgdirect{oopsla25july22}{1}{T}{T}{ls}{F} &      \Unk  &      2.07 \\
& \evdrift &      \humanCfgdirect{oopsla25july22}{1}{F}{F}{pg}{F} &      \Unk  &      2.59 \\
& \evdrift &      \humanCfgdirect{oopsla25july22}{0}{F}{T}{ls}{T} &      \Unk  &      0.40 \\
& \evdrift &      \humanCfgdirect{oopsla25july22}{1}{T}{F}{pg}{F} &      \Chk  &      5.64 \\
& \evdrift &      \humanCfgdirect{oopsla25july22}{0}{F}{T}{ls}{F} &      \Unk  &      0.44 \\
& \evdrift &      \humanCfgdirect{oopsla25july22}{1}{F}{T}{ls}{F} &      \Unk  &      1.01 \\
\hline
13. \texttt{\scriptsize lics18-amortized} \\
& \drift &      \humanCfgtrans{oopsla25july22}{1}{F}{F}{pg}{F} &      \TO   &      901.01 \\
& \drift &      \humanCfgtrans{oopsla25july22}{0}{F}{T}{ls}{F} &      \TO   &      900.43 \\
& \drift &      \humanCfgtrans{oopsla25july22}{1}{F}{T}{ls}{F} &      \Unk  &      288.04 \\
& \drift &      \humanCfgtrans{oopsla25july22}{1}{T}{F}{pg}{F} &      \TO   &      901.02 \\
& \drift &      \humanCfgtrans{oopsla25july22}{0}{F}{F}{pg}{F} &      \Unk  &      276.11 \\
& \drift &      \humanCfgtrans{oopsla25july22}{1}{T}{T}{ls}{F} &      \Unk  &      587.41 \\
& \evdrift &      \humanCfgdirect{oopsla25july22}{1}{F}{T}{ls}{T} &      \Chk  &      15.81 \\
& \evdrift &      \humanCfgdirect{oopsla25july22}{1}{T}{F}{pg}{T} &      \Unk  &      69.32 \\
& \evdrift &      \humanCfgdirect{oopsla25july22}{1}{F}{F}{pg}{T} &      \Unk  &      69.41 \\
& \evdrift &      \humanCfgdirect{oopsla25july22}{0}{F}{F}{pg}{T} &      \Unk  &      20.37 \\
& \evdrift &      \humanCfgdirect{oopsla25july22}{0}{F}{F}{pg}{F} &      \Unk  &      20.23 \\
& \evdrift &      \humanCfgdirect{oopsla25july22}{1}{T}{T}{ls}{T} &      \Chk  &      16.15 \\
& \evdrift &      \humanCfgdirect{oopsla25july22}{0}{F}{T}{ls}{T} &      \Chk  &      6.47 \\
& \evdrift &      \humanCfgdirect{oopsla25july22}{1}{T}{F}{pg}{F} &      \Unk  &      68.85 \\
& \evdrift & \hl  \humanCfgdirect{oopsla25july22}{0}{F}{T}{ls}{F} & \hl  \Chk  & \hl  6.32 \\
& \evdrift &      \humanCfgdirect{oopsla25july22}{1}{F}{T}{ls}{F} &      \Chk  &      15.87 \\
& \evdrift &      \humanCfgdirect{oopsla25july22}{1}{F}{F}{pg}{F} &      \Unk  &      68.66 \\
& \evdrift &      \humanCfgdirect{oopsla25july22}{1}{T}{T}{ls}{F} &      \Chk  &      15.77 \\
\hline
14. \texttt{\scriptsize lics18-hoshrink} \\
& \drift &      \humanCfgtrans{oopsla25july22}{1}{T}{T}{ls}{F} &      \Unk  &      3.36 \\
& \drift &      \humanCfgtrans{oopsla25july22}{0}{F}{F}{pg}{F} &      \Unk  &      9.73 \\
& \drift &      \humanCfgtrans{oopsla25july22}{1}{F}{F}{pg}{F} &      \Unk  &      19.35 \\
& \drift &      \humanCfgtrans{oopsla25july22}{1}{F}{T}{ls}{F} &      \Unk  &      3.42 \\
& \drift &      \humanCfgtrans{oopsla25july22}{0}{F}{T}{ls}{F} &      \Unk  &      1.22 \\
& \drift &      \humanCfgtrans{oopsla25july22}{1}{T}{F}{pg}{F} &      \Unk  &      20.18 \\
& \evdrift &      \humanCfgdirect{oopsla25july22}{1}{T}{T}{ls}{F} &      \Unk  &      0.92 \\
& \evdrift & \hl  \humanCfgdirect{oopsla25july22}{1}{F}{F}{pg}{F} & \hl  \Unk  & \hl  6.90 \\
& \evdrift &      \humanCfgdirect{oopsla25july22}{0}{F}{T}{ls}{T} &      \Unk  &      0.38 \\
& \evdrift &      \humanCfgdirect{oopsla25july22}{0}{F}{T}{ls}{F} &      \Unk  &      0.42 \\
& \evdrift &      \humanCfgdirect{oopsla25july22}{1}{T}{F}{pg}{F} &      \Unk  &      7.11 \\
& \evdrift &      \humanCfgdirect{oopsla25july22}{1}{F}{T}{ls}{F} &      \Unk  &      1.04 \\
& \evdrift &      \humanCfgdirect{oopsla25july22}{1}{T}{T}{ls}{T} &      \Unk  &      1.16 \\
& \evdrift &      \humanCfgdirect{oopsla25july22}{0}{F}{F}{pg}{F} &      \Unk  &      2.99 \\
& \evdrift &      \humanCfgdirect{oopsla25july22}{1}{F}{T}{ls}{T} &      \Unk  &      1.08 \\
& \evdrift &      \humanCfgdirect{oopsla25july22}{1}{T}{F}{pg}{T} &      \Unk  &      7.01 \\
& \evdrift &      \humanCfgdirect{oopsla25july22}{1}{F}{F}{pg}{T} &      \Unk  &      6.41 \\
& \evdrift &      \humanCfgdirect{oopsla25july22}{0}{F}{F}{pg}{T} &      \Unk  &      2.93 \\
\hline
15. \texttt{\scriptsize lics18-web} \\
& \drift &      \humanCfgtrans{oopsla25july22}{1}{F}{F}{pg}{F} &      \Unk  &      241.51 \\
& \drift &      \humanCfgtrans{oopsla25july22}{0}{F}{T}{ls}{F} &      \Unk  &      53.00 \\
& \drift &      \humanCfgtrans{oopsla25july22}{1}{F}{T}{ls}{F} &      \Unk  &      189.32 \\
& \drift &      \humanCfgtrans{oopsla25july22}{1}{T}{F}{pg}{F} &      \TO   &      900.27 \\
& \drift &      \humanCfgtrans{oopsla25july22}{1}{T}{T}{ls}{F} &      \TO   &      901.02 \\
& \drift &      \humanCfgtrans{oopsla25july22}{0}{F}{F}{pg}{F} &      \Unk  &      91.42 \\
& \evdrift &      \humanCfgdirect{oopsla25july22}{0}{F}{F}{pg}{F} &      \Unk  &      10.26 \\
& \evdrift &      \humanCfgdirect{oopsla25july22}{1}{T}{T}{ls}{T} &      \Chk  &      23.83 \\
& \evdrift &      \humanCfgdirect{oopsla25july22}{1}{F}{T}{ls}{T} &      \Chk  &      23.89 \\
& \evdrift &      \humanCfgdirect{oopsla25july22}{1}{T}{F}{pg}{T} &      \Unk  &      59.56 \\
& \evdrift &      \humanCfgdirect{oopsla25july22}{1}{F}{F}{pg}{T} &      \Unk  &      59.68 \\
& \evdrift &      \humanCfgdirect{oopsla25july22}{0}{F}{F}{pg}{T} &      \Unk  &      10.57 \\
& \evdrift &      \humanCfgdirect{oopsla25july22}{1}{T}{T}{ls}{F} &      \Chk  &      23.56 \\
& \evdrift &      \humanCfgdirect{oopsla25july22}{1}{F}{F}{pg}{F} &      \Unk  &      58.95 \\
& \evdrift &      \humanCfgdirect{oopsla25july22}{0}{F}{T}{ls}{T} &      \Chk  &      7.45 \\
& \evdrift &      \humanCfgdirect{oopsla25july22}{1}{T}{F}{pg}{F} &      \Unk  &      59.19 \\
& \evdrift & \hl  \humanCfgdirect{oopsla25july22}{0}{F}{T}{ls}{F} & \hl  \Chk  & \hl  7.21 \\
& \evdrift &      \humanCfgdirect{oopsla25july22}{1}{F}{T}{ls}{F} &      \Chk  &      23.78 \\
\hline
16. \texttt{\scriptsize market} \\
& \drift &      \humanCfgtrans{oopsla25july22}{0}{F}{F}{pg}{F} &      \Unk  &      127.80 \\
& \drift &      \humanCfgtrans{oopsla25july22}{1}{T}{T}{ls}{F} &      \Unk  &      641.30 \\
& \drift &      \humanCfgtrans{oopsla25july22}{1}{F}{T}{ls}{F} &      \Unk  &      277.58 \\
& \drift &      \humanCfgtrans{oopsla25july22}{0}{F}{T}{ls}{F} &      \Unk  &      33.12 \\
& \drift &      \humanCfgtrans{oopsla25july22}{1}{T}{F}{pg}{F} &      \TO   &      900.85 \\
& \drift &      \humanCfgtrans{oopsla25july22}{1}{F}{F}{pg}{F} &      \TO   &      901.01 \\
& \evdrift &      \humanCfgdirect{oopsla25july22}{0}{F}{T}{ls}{F} &      \Unk  &      5.47 \\
& \evdrift &      \humanCfgdirect{oopsla25july22}{1}{T}{F}{pg}{F} &      \Unk  &      36.76 \\
& \evdrift &      \humanCfgdirect{oopsla25july22}{1}{F}{T}{ls}{F} &      \Unk  &      14.43 \\
& \evdrift &      \humanCfgdirect{oopsla25july22}{0}{F}{T}{ls}{T} &      \Unk  &      5.43 \\
& \evdrift &      \humanCfgdirect{oopsla25july22}{1}{T}{T}{ls}{F} &      \Unk  &      14.88 \\
& \evdrift &      \humanCfgdirect{oopsla25july22}{1}{F}{F}{pg}{F} &      \Unk  &      35.36 \\
& \evdrift &      \humanCfgdirect{oopsla25july22}{1}{F}{F}{pg}{T} &      \Unk  &      34.96 \\
& \evdrift & \hl  \humanCfgdirect{oopsla25july22}{0}{F}{F}{pg}{T} & \hl  \Unk  & \hl  14.78 \\
& \evdrift &      \humanCfgdirect{oopsla25july22}{1}{F}{T}{ls}{T} &      \Unk  &      14.45 \\
& \evdrift &      \humanCfgdirect{oopsla25july22}{1}{T}{F}{pg}{T} &      \Unk  &      37.05 \\
& \evdrift &      \humanCfgdirect{oopsla25july22}{1}{T}{T}{ls}{T} &      \Unk  &      15.03 \\
& \evdrift &      \humanCfgdirect{oopsla25july22}{0}{F}{F}{pg}{F} &      \Unk  &      14.93 \\
\hline
17. \texttt{\scriptsize max-min} \\
& \drift &      \humanCfgtrans{oopsla25july22}{1}{T}{T}{ls}{F} &      \Unk  &      233.49 \\
& \drift &      \humanCfgtrans{oopsla25july22}{0}{F}{F}{pg}{F} &      \Unk  &      34.04 \\
& \drift &      \humanCfgtrans{oopsla25july22}{1}{F}{F}{pg}{F} &      \Unk  &      145.14 \\
& \drift &      \humanCfgtrans{oopsla25july22}{1}{T}{F}{pg}{F} &      \TO   &      900.58 \\
& \drift &      \humanCfgtrans{oopsla25july22}{1}{F}{T}{ls}{F} &      \Unk  &      24.44 \\
& \drift &      \humanCfgtrans{oopsla25july22}{0}{F}{T}{ls}{F} &      \Unk  &      9.19 \\
& \evdrift &      \humanCfgdirect{oopsla25july22}{1}{T}{T}{ls}{F} &      \Chk  &      30.38 \\
& \evdrift &      \humanCfgdirect{oopsla25july22}{1}{F}{F}{pg}{F} &      \Unk  &      28.20 \\
& \evdrift &      \humanCfgdirect{oopsla25july22}{0}{F}{T}{ls}{T} &      \Unk  &      3.77 \\
& \evdrift &      \humanCfgdirect{oopsla25july22}{1}{F}{T}{ls}{F} &      \Unk  &      6.31 \\
& \evdrift &      \humanCfgdirect{oopsla25july22}{0}{F}{T}{ls}{F} &      \Unk  &      3.73 \\
& \evdrift &      \humanCfgdirect{oopsla25july22}{1}{T}{F}{pg}{F} &      \Chk  &      48.57 \\
& \evdrift & \hl  \humanCfgdirect{oopsla25july22}{1}{T}{T}{ls}{T} & \hl  \Chk  & \hl  30.10 \\
& \evdrift &      \humanCfgdirect{oopsla25july22}{0}{F}{F}{pg}{F} &      \Unk  &      24.39 \\
& \evdrift &      \humanCfgdirect{oopsla25july22}{1}{T}{F}{pg}{T} &      \Chk  &      48.28 \\
& \evdrift &      \humanCfgdirect{oopsla25july22}{1}{F}{T}{ls}{T} &      \Unk  &      6.29 \\
& \evdrift &      \humanCfgdirect{oopsla25july22}{0}{F}{F}{pg}{T} &      \Unk  &      24.35 \\
& \evdrift &      \humanCfgdirect{oopsla25july22}{1}{F}{F}{pg}{T} &      \Unk  &      28.45 \\
\hline
18. \texttt{\scriptsize monotonic} \\
& \drift &      \humanCfgtrans{oopsla25july22}{0}{F}{F}{pg}{F} &      \Chk  &      5.33 \\
& \drift &      \humanCfgtrans{oopsla25july22}{1}{T}{T}{ls}{F} &      \Chk  &      6.28 \\
& \drift &      \humanCfgtrans{oopsla25july22}{0}{F}{T}{ls}{F} &      \Chk  &      2.38 \\
& \drift &      \humanCfgtrans{oopsla25july22}{1}{F}{T}{ls}{F} &      \Chk  &      3.60 \\
& \drift &      \humanCfgtrans{oopsla25july22}{1}{T}{F}{pg}{F} &      \Chk  &      31.05 \\
& \drift &      \humanCfgtrans{oopsla25july22}{1}{F}{F}{pg}{F} &      \Unk  &      19.91 \\
& \evdrift &      \humanCfgdirect{oopsla25july22}{1}{T}{F}{pg}{F} &      \Unk  &      2.44 \\
& \evdrift &      \humanCfgdirect{oopsla25july22}{0}{F}{T}{ls}{F} &      \Chk  &      0.36 \\
& \evdrift &      \humanCfgdirect{oopsla25july22}{1}{F}{T}{ls}{F} &      \Chk  &      0.59 \\
& \evdrift & \hl  \humanCfgdirect{oopsla25july22}{0}{F}{T}{ls}{T} & \hl  \Chk  & \hl  0.35 \\
& \evdrift &      \humanCfgdirect{oopsla25july22}{1}{F}{F}{pg}{F} &      \Unk  &      2.26 \\
& \evdrift &      \humanCfgdirect{oopsla25july22}{1}{T}{T}{ls}{F} &      \Chk  &      0.49 \\
& \evdrift &      \humanCfgdirect{oopsla25july22}{1}{F}{F}{pg}{T} &      \Unk  &      2.16 \\
& \evdrift &      \humanCfgdirect{oopsla25july22}{0}{F}{F}{pg}{T} &      \Chk  &      0.96 \\
& \evdrift &      \humanCfgdirect{oopsla25july22}{1}{F}{T}{ls}{T} &      \Chk  &      0.54 \\
& \evdrift &      \humanCfgdirect{oopsla25july22}{1}{T}{F}{pg}{T} &      \Unk  &      2.32 \\
& \evdrift &      \humanCfgdirect{oopsla25july22}{0}{F}{F}{pg}{F} &      \Chk  &      0.92 \\
& \evdrift &      \humanCfgdirect{oopsla25july22}{1}{T}{T}{ls}{T} &      \Chk  &      0.49 \\
\hline
19. \texttt{\scriptsize nondet\_max} \\
& \drift &      \humanCfgtrans{oopsla25july22}{1}{F}{F}{pg}{F} &      \Chk  &      17.40 \\
& \drift &      \humanCfgtrans{oopsla25july22}{1}{F}{T}{ls}{F} &      \Chk  &      5.14 \\
& \drift &      \humanCfgtrans{oopsla25july22}{0}{F}{T}{ls}{F} &      \Chk  &      2.24 \\
& \drift &      \humanCfgtrans{oopsla25july22}{1}{T}{F}{pg}{F} &      \Chk  &      17.25 \\
& \drift &      \humanCfgtrans{oopsla25july22}{0}{F}{F}{pg}{F} &      \Unk  &      6.39 \\
& \drift &      \humanCfgtrans{oopsla25july22}{1}{T}{T}{ls}{F} &      \Chk  &      3.80 \\
& \evdrift &      \humanCfgdirect{oopsla25july22}{1}{F}{T}{ls}{T} &      \Chk  &      1.04 \\
& \evdrift &      \humanCfgdirect{oopsla25july22}{1}{T}{F}{pg}{T} &      \Unk  &      5.55 \\
& \evdrift &      \humanCfgdirect{oopsla25july22}{1}{F}{F}{pg}{T} &      \Unk  &      5.52 \\
& \evdrift &      \humanCfgdirect{oopsla25july22}{0}{F}{F}{pg}{T} &      \Unk  &      1.90 \\
& \evdrift &      \humanCfgdirect{oopsla25july22}{1}{T}{T}{ls}{T} &      \Chk  &      0.99 \\
& \evdrift &      \humanCfgdirect{oopsla25july22}{0}{F}{F}{pg}{F} &      \Unk  &      1.68 \\
& \evdrift &      \humanCfgdirect{oopsla25july22}{0}{F}{T}{ls}{T} &      \Chk  &      0.61 \\
& \evdrift & \hl  \humanCfgdirect{oopsla25july22}{0}{F}{T}{ls}{F} & \hl  \Chk  & \hl  0.59 \\
& \evdrift &      \humanCfgdirect{oopsla25july22}{1}{T}{F}{pg}{F} &      \Unk  &      5.59 \\
& \evdrift &      \humanCfgdirect{oopsla25july22}{1}{F}{T}{ls}{F} &      \Chk  &      1.07 \\
& \evdrift &      \humanCfgdirect{oopsla25july22}{1}{T}{T}{ls}{F} &      \Chk  &      1.03 \\
& \evdrift &      \humanCfgdirect{oopsla25july22}{1}{F}{F}{pg}{F} &      \Unk  &      4.71 \\
\hline
20. \texttt{\scriptsize num\_evens} \\
& \drift &      \humanCfgtrans{oopsla25july22}{1}{T}{T}{ls}{F} &      \Chk  &      21.41 \\
& \drift &      \humanCfgtrans{oopsla25july22}{0}{F}{F}{pg}{F} &      \Unk  &      16.61 \\
& \drift &      \humanCfgtrans{oopsla25july22}{1}{F}{F}{pg}{F} &      \Chk  &      30.58 \\
& \drift &      \humanCfgtrans{oopsla25july22}{0}{F}{T}{ls}{F} &      \Chk  &      8.88 \\
& \drift &      \humanCfgtrans{oopsla25july22}{1}{F}{T}{ls}{F} &      \Chk  &      13.58 \\
& \drift &      \humanCfgtrans{oopsla25july22}{1}{T}{F}{pg}{F} &      \Chk  &      61.34 \\
& \evdrift &      \humanCfgdirect{oopsla25july22}{1}{T}{T}{ls}{F} &      \Chk  &      4.50 \\
& \evdrift &      \humanCfgdirect{oopsla25july22}{1}{F}{F}{pg}{F} &      \Chk  &      4.06 \\
& \evdrift &      \humanCfgdirect{oopsla25july22}{0}{F}{T}{ls}{T} &      \Unk  &      0.90 \\
& \evdrift &      \humanCfgdirect{oopsla25july22}{1}{T}{F}{pg}{F} &      \Chk  &      7.46 \\
& \evdrift &      \humanCfgdirect{oopsla25july22}{0}{F}{T}{ls}{F} &      \Unk  &      0.93 \\
& \evdrift &      \humanCfgdirect{oopsla25july22}{1}{F}{T}{ls}{F} &      \Chk  &      2.46 \\
& \evdrift &      \humanCfgdirect{oopsla25july22}{0}{F}{F}{pg}{F} &      \Unk  &      1.80 \\
& \evdrift &      \humanCfgdirect{oopsla25july22}{1}{T}{T}{ls}{T} &      \Chk  &      4.53 \\
& \evdrift & \hl  \humanCfgdirect{oopsla25july22}{1}{F}{T}{ls}{T} & \hl  \Chk  & \hl  2.34 \\
& \evdrift &      \humanCfgdirect{oopsla25july22}{1}{T}{F}{pg}{T} &      \Chk  &      7.24 \\
& \evdrift &      \humanCfgdirect{oopsla25july22}{1}{F}{F}{pg}{T} &      \Chk  &      4.20 \\
& \evdrift &      \humanCfgdirect{oopsla25july22}{0}{F}{F}{pg}{T} &      \Unk  &      1.91 \\
\hline
21. \texttt{\scriptsize order-irrel} \\
& \drift &      \humanCfgtrans{oopsla25july22}{0}{F}{F}{pg}{F} &      \Unk  &      9.99 \\
& \drift &      \humanCfgtrans{oopsla25july22}{1}{T}{T}{ls}{F} &      \Unk  &      15.91 \\
& \drift &      \humanCfgtrans{oopsla25july22}{0}{F}{T}{ls}{F} &      \Unk  &      1.85 \\
& \drift &      \humanCfgtrans{oopsla25july22}{1}{F}{T}{ls}{F} &      \Unk  &      2.97 \\
& \drift &      \humanCfgtrans{oopsla25july22}{1}{T}{F}{pg}{F} &      \Unk  &      36.43 \\
& \drift &      \humanCfgtrans{oopsla25july22}{1}{F}{F}{pg}{F} &      \Unk  &      13.22 \\
& \evdrift &      \humanCfgdirect{oopsla25july22}{1}{T}{F}{pg}{F} &      \Chk  &      3.80 \\
& \evdrift &      \humanCfgdirect{oopsla25july22}{0}{F}{T}{ls}{F} &      \Unk  &      0.52 \\
& \evdrift &      \humanCfgdirect{oopsla25july22}{1}{F}{T}{ls}{F} &      \Unk  &      0.89 \\
& \evdrift &      \humanCfgdirect{oopsla25july22}{0}{F}{T}{ls}{T} &      \Unk  &      0.50 \\
& \evdrift &      \humanCfgdirect{oopsla25july22}{1}{T}{T}{ls}{F} &      \Unk  &      1.74 \\
& \evdrift &      \humanCfgdirect{oopsla25july22}{1}{F}{F}{pg}{F} &      \Unk  &      3.38 \\
& \evdrift &      \humanCfgdirect{oopsla25july22}{1}{F}{F}{pg}{T} &      \Unk  &      3.61 \\
& \evdrift &      \humanCfgdirect{oopsla25july22}{0}{F}{F}{pg}{T} &      \Unk  &      1.67 \\
& \evdrift &      \humanCfgdirect{oopsla25july22}{1}{F}{T}{ls}{T} &      \Unk  &      0.82 \\
& \evdrift & \hl  \humanCfgdirect{oopsla25july22}{1}{T}{F}{pg}{T} & \hl  \Chk  & \hl  3.54 \\
& \evdrift &      \humanCfgdirect{oopsla25july22}{0}{F}{F}{pg}{F} &      \Unk  &      1.63 \\
& \evdrift &      \humanCfgdirect{oopsla25july22}{1}{T}{T}{ls}{T} &      \Unk  &      1.77 \\
\hline
22. \texttt{\scriptsize order-irrel-nondet} \\
& \drift &      \humanCfgtrans{oopsla25july22}{0}{F}{F}{pg}{F} &      \Unk  &      13.01 \\
& \drift &      \humanCfgtrans{oopsla25july22}{1}{T}{T}{ls}{F} &      \Unk  &      14.84 \\
& \drift &      \humanCfgtrans{oopsla25july22}{1}{F}{F}{pg}{F} &      \Unk  &      27.81 \\
& \drift &      \humanCfgtrans{oopsla25july22}{0}{F}{T}{ls}{F} &      \Unk  &      2.63 \\
& \drift &      \humanCfgtrans{oopsla25july22}{1}{F}{T}{ls}{F} &      \Unk  &      6.72 \\
& \drift &      \humanCfgtrans{oopsla25july22}{1}{T}{F}{pg}{F} &      \Unk  &      75.50 \\
& \evdrift &      \humanCfgdirect{oopsla25july22}{0}{F}{T}{ls}{T} &      \Unk  &      1.13 \\
& \evdrift &      \humanCfgdirect{oopsla25july22}{1}{T}{F}{pg}{F} &      \Chk  &      8.38 \\
& \evdrift &      \humanCfgdirect{oopsla25july22}{0}{F}{T}{ls}{F} &      \Unk  &      1.14 \\
& \evdrift &      \humanCfgdirect{oopsla25july22}{1}{F}{T}{ls}{F} &      \Unk  &      2.01 \\
& \evdrift &      \humanCfgdirect{oopsla25july22}{1}{T}{T}{ls}{F} &      \Chk  &      2.19 \\
& \evdrift &      \humanCfgdirect{oopsla25july22}{1}{F}{F}{pg}{F} &      \Unk  &      10.27 \\
& \evdrift &      \humanCfgdirect{oopsla25july22}{1}{F}{T}{ls}{T} &      \Unk  &      1.99 \\
& \evdrift &      \humanCfgdirect{oopsla25july22}{1}{T}{F}{pg}{T} &      \Chk  &      8.32 \\
& \evdrift &      \humanCfgdirect{oopsla25july22}{1}{F}{F}{pg}{T} &      \Unk  &      10.39 \\
& \evdrift &      \humanCfgdirect{oopsla25july22}{0}{F}{F}{pg}{T} &      \Unk  &      6.26 \\
& \evdrift &      \humanCfgdirect{oopsla25july22}{0}{F}{F}{pg}{F} &      \Unk  &      6.19 \\
& \evdrift & \hl  \humanCfgdirect{oopsla25july22}{1}{T}{T}{ls}{T} & \hl  \Chk  & \hl  1.84 \\
\hline
23. \texttt{\scriptsize overview1} \\
& \drift &      \humanCfgtrans{oopsla25july22}{1}{F}{F}{pg}{F} &      \Chk  &      5.02 \\
& \drift &      \humanCfgtrans{oopsla25july22}{1}{T}{F}{pg}{F} &      \Chk  &      5.39 \\
& \drift &      \humanCfgtrans{oopsla25july22}{1}{F}{T}{ls}{F} &      \Chk  &      1.95 \\
& \drift &      \humanCfgtrans{oopsla25july22}{0}{F}{T}{ls}{F} &      \Unk  &      0.76 \\
& \drift &      \humanCfgtrans{oopsla25july22}{1}{T}{T}{ls}{F} &      \Chk  &      2.35 \\
& \drift &      \humanCfgtrans{oopsla25july22}{0}{F}{F}{pg}{F} &      \Unk  &      4.15 \\
& \evdrift &      \humanCfgdirect{oopsla25july22}{1}{T}{T}{ls}{T} &      \Chk  &      0.49 \\
& \evdrift &      \humanCfgdirect{oopsla25july22}{0}{F}{F}{pg}{F} &      \Chk  &      0.70 \\
& \evdrift &      \humanCfgdirect{oopsla25july22}{1}{T}{F}{pg}{T} &      \Chk  &      1.24 \\
& \evdrift &      \humanCfgdirect{oopsla25july22}{1}{F}{T}{ls}{T} &      \Chk  &      0.41 \\
& \evdrift &      \humanCfgdirect{oopsla25july22}{0}{F}{F}{pg}{T} &      \Chk  &      0.66 \\
& \evdrift &      \humanCfgdirect{oopsla25july22}{1}{F}{F}{pg}{T} &      \Chk  &      1.28 \\
& \evdrift &      \humanCfgdirect{oopsla25july22}{1}{F}{F}{pg}{F} &      \Chk  &      1.23 \\
& \evdrift &      \humanCfgdirect{oopsla25july22}{1}{T}{T}{ls}{F} &      \Chk  &      0.52 \\
& \evdrift & \hl  \humanCfgdirect{oopsla25july22}{0}{F}{T}{ls}{T} & \hl  \Chk  & \hl  0.29 \\
& \evdrift &      \humanCfgdirect{oopsla25july22}{1}{F}{T}{ls}{F} &      \Chk  &      0.52 \\
& \evdrift &      \humanCfgdirect{oopsla25july22}{0}{F}{T}{ls}{F} &      \Chk  &      0.30 \\
& \evdrift &      \humanCfgdirect{oopsla25july22}{1}{T}{F}{pg}{F} &      \Chk  &      1.24 \\
\hline
24. \texttt{\scriptsize reentr} \\
& \drift &      \humanCfgtrans{oopsla25july22}{1}{F}{F}{pg}{F} &      \Chk  &      14.98 \\
& \drift &      \humanCfgtrans{oopsla25july22}{1}{T}{F}{pg}{F} &      \Chk  &      19.02 \\
& \drift &      \humanCfgtrans{oopsla25july22}{0}{F}{T}{ls}{F} &      \Chk  &      3.66 \\
& \drift &      \humanCfgtrans{oopsla25july22}{1}{F}{T}{ls}{F} &      \Chk  &      6.84 \\
& \drift &      \humanCfgtrans{oopsla25july22}{1}{T}{T}{ls}{F} &      \Chk  &      8.84 \\
& \drift &      \humanCfgtrans{oopsla25july22}{0}{F}{F}{pg}{F} &      \Chk  &      6.57 \\
& \evdrift &      \humanCfgdirect{oopsla25july22}{0}{F}{F}{pg}{F} &      \Chk  &      0.18 \\
& \evdrift &      \humanCfgdirect{oopsla25july22}{1}{T}{T}{ls}{T} &      \Chk  &      0.18 \\
& \evdrift &      \humanCfgdirect{oopsla25july22}{1}{T}{F}{pg}{T} &      \Chk  &      0.25 \\
& \evdrift &      \humanCfgdirect{oopsla25july22}{1}{F}{T}{ls}{T} &      \Chk  &      0.17 \\
& \evdrift &      \humanCfgdirect{oopsla25july22}{0}{F}{F}{pg}{T} &      \Chk  &      0.17 \\
& \evdrift &      \humanCfgdirect{oopsla25july22}{1}{F}{F}{pg}{T} &      \Chk  &      0.26 \\
& \evdrift &      \humanCfgdirect{oopsla25july22}{1}{T}{T}{ls}{F} &      \Chk  &      0.17 \\
& \evdrift &      \humanCfgdirect{oopsla25july22}{1}{F}{F}{pg}{F} &      \Chk  &      0.25 \\
& \evdrift & \hl  \humanCfgdirect{oopsla25july22}{0}{F}{T}{ls}{T} & \hl  \Chk  & \hl  0.14 \\
& \evdrift &      \humanCfgdirect{oopsla25july22}{1}{F}{T}{ls}{F} &      \Chk  &      0.18 \\
& \evdrift &      \humanCfgdirect{oopsla25july22}{1}{T}{F}{pg}{F} &      \Chk  &      0.27 \\
& \evdrift &      \humanCfgdirect{oopsla25july22}{0}{F}{T}{ls}{F} &      \Chk  &      0.15 \\
\hline
25. \texttt{\scriptsize resource-analysis} \\
& \drift &      \humanCfgtrans{oopsla25july22}{0}{F}{F}{pg}{F} &      \Unk  &      4.67 \\
& \drift &      \humanCfgtrans{oopsla25july22}{1}{T}{T}{ls}{F} &      \Chk  &      3.10 \\
& \drift &      \humanCfgtrans{oopsla25july22}{1}{F}{F}{pg}{F} &      \Chk  &      9.36 \\
& \drift &      \humanCfgtrans{oopsla25july22}{1}{T}{F}{pg}{F} &      \Chk  &      9.56 \\
& \drift &      \humanCfgtrans{oopsla25july22}{1}{F}{T}{ls}{F} &      \Chk  &      3.07 \\
& \drift &      \humanCfgtrans{oopsla25july22}{0}{F}{T}{ls}{F} &      \Chk  &      2.44 \\
& \evdrift &      \humanCfgdirect{oopsla25july22}{0}{F}{T}{ls}{T} &      \Chk  &      0.18 \\
& \evdrift &      \humanCfgdirect{oopsla25july22}{1}{F}{T}{ls}{F} &      \Chk  &      0.26 \\
& \evdrift & \hl  \humanCfgdirect{oopsla25july22}{0}{F}{T}{ls}{F} & \hl  \Chk  & \hl  0.17 \\
& \evdrift &      \humanCfgdirect{oopsla25july22}{1}{T}{F}{pg}{F} &      \Chk  &      0.48 \\
& \evdrift &      \humanCfgdirect{oopsla25july22}{1}{F}{F}{pg}{F} &      \Chk  &      0.44 \\
& \evdrift &      \humanCfgdirect{oopsla25july22}{1}{T}{T}{ls}{F} &      \Chk  &      0.26 \\
& \evdrift &      \humanCfgdirect{oopsla25july22}{1}{T}{F}{pg}{T} &      \Chk  &      0.44 \\
& \evdrift &      \humanCfgdirect{oopsla25july22}{1}{F}{T}{ls}{T} &      \Chk  &      0.26 \\
& \evdrift &      \humanCfgdirect{oopsla25july22}{0}{F}{F}{pg}{T} &      \Chk  &      0.26 \\
& \evdrift &      \humanCfgdirect{oopsla25july22}{1}{F}{F}{pg}{T} &      \Chk  &      0.46 \\
& \evdrift &      \humanCfgdirect{oopsla25july22}{1}{T}{T}{ls}{T} &      \Chk  &      0.27 \\
& \evdrift &      \humanCfgdirect{oopsla25july22}{0}{F}{F}{pg}{F} &      \Chk  &      0.26 \\
\hline
26. \texttt{\scriptsize sum-appendix} \\
& \drift &      \humanCfgtrans{oopsla25july22}{1}{F}{F}{pg}{F} &      \Chk  &      5.98 \\
& \drift &      \humanCfgtrans{oopsla25july22}{0}{F}{T}{ls}{F} &      \Chk  &      1.18 \\
& \drift &      \humanCfgtrans{oopsla25july22}{1}{F}{T}{ls}{F} &      \Chk  &      1.63 \\
& \drift &      \humanCfgtrans{oopsla25july22}{1}{T}{F}{pg}{F} &      \Chk  &      5.86 \\
& \drift &      \humanCfgtrans{oopsla25july22}{0}{F}{F}{pg}{F} &      \Chk  &      4.38 \\
& \drift &      \humanCfgtrans{oopsla25july22}{1}{T}{T}{ls}{F} &      \Chk  &      1.63 \\
& \evdrift &      \humanCfgdirect{oopsla25july22}{1}{F}{T}{ls}{T} &      \Chk  &      0.03 \\
& \evdrift &      \humanCfgdirect{oopsla25july22}{1}{T}{F}{pg}{T} &      \Chk  &      0.06 \\
& \evdrift &      \humanCfgdirect{oopsla25july22}{1}{F}{F}{pg}{T} &      \Chk  &      0.06 \\
& \evdrift &      \humanCfgdirect{oopsla25july22}{0}{F}{F}{pg}{T} &      \Chk  &      0.06 \\
& \evdrift &      \humanCfgdirect{oopsla25july22}{0}{F}{F}{pg}{F} &      \Chk  &      0.06 \\
& \evdrift &      \humanCfgdirect{oopsla25july22}{1}{T}{T}{ls}{T} &      \Chk  &      0.03 \\
& \evdrift & \hl  \humanCfgdirect{oopsla25july22}{0}{F}{T}{ls}{T} & \hl  \Chk  & \hl  0.02 \\
& \evdrift &      \humanCfgdirect{oopsla25july22}{1}{T}{F}{pg}{F} &      \Chk  &      0.06 \\
& \evdrift &      \humanCfgdirect{oopsla25july22}{0}{F}{T}{ls}{F} &      \Chk  &      0.02 \\
& \evdrift &      \humanCfgdirect{oopsla25july22}{1}{F}{T}{ls}{F} &      \Chk  &      0.03 \\
& \evdrift &      \humanCfgdirect{oopsla25july22}{1}{F}{F}{pg}{F} &      \Chk  &      0.06 \\
& \evdrift &      \humanCfgdirect{oopsla25july22}{1}{T}{T}{ls}{F} &      \Chk  &      0.03 \\
\hline
27. \texttt{\scriptsize sum-of-ev-even} \\
& \drift &      \humanCfgtrans{oopsla25july22}{0}{F}{F}{pg}{F} &      \Chk  &      2.11 \\
& \drift &      \humanCfgtrans{oopsla25july22}{1}{T}{T}{ls}{F} &      \Unk  &      1.77 \\
& \drift &      \humanCfgtrans{oopsla25july22}{0}{F}{T}{ls}{F} &      \Unk  &      0.56 \\
& \drift &      \humanCfgtrans{oopsla25july22}{1}{F}{T}{ls}{F} &      \Unk  &      1.71 \\
& \drift &      \humanCfgtrans{oopsla25july22}{1}{T}{F}{pg}{F} &      \Chk  &      6.19 \\
& \drift &      \humanCfgtrans{oopsla25july22}{1}{F}{F}{pg}{F} &      \Chk  &      5.89 \\
& \evdrift &      \humanCfgdirect{oopsla25july22}{1}{T}{F}{pg}{F} &      \Chk  &      0.69 \\
& \evdrift &      \humanCfgdirect{oopsla25july22}{0}{F}{T}{ls}{F} &      \Unk  &      0.13 \\
& \evdrift &      \humanCfgdirect{oopsla25july22}{1}{F}{T}{ls}{F} &      \Unk  &      0.29 \\
& \evdrift &      \humanCfgdirect{oopsla25july22}{0}{F}{T}{ls}{T} &      \Unk  &      0.13 \\
& \evdrift &      \humanCfgdirect{oopsla25july22}{1}{F}{F}{pg}{F} &      \Chk  &      0.66 \\
& \evdrift &      \humanCfgdirect{oopsla25july22}{1}{T}{T}{ls}{F} &      \Unk  &      0.30 \\
& \evdrift &      \humanCfgdirect{oopsla25july22}{1}{F}{F}{pg}{T} &      \Chk  &      0.68 \\
& \evdrift &      \humanCfgdirect{oopsla25july22}{0}{F}{F}{pg}{T} &      \Chk  &      0.32 \\
& \evdrift &      \humanCfgdirect{oopsla25july22}{1}{F}{T}{ls}{T} &      \Unk  &      0.27 \\
& \evdrift &      \humanCfgdirect{oopsla25july22}{1}{T}{F}{pg}{T} &      \Chk  &      0.67 \\
& \evdrift & \hl  \humanCfgdirect{oopsla25july22}{0}{F}{F}{pg}{F} & \hl  \Chk  & \hl  0.32 \\
& \evdrift &      \humanCfgdirect{oopsla25july22}{1}{T}{T}{ls}{T} &      \Unk  &      0.30 \\
\hline
28. \texttt{\scriptsize temperature} \\
& \drift &      \humanCfgtrans{oopsla25july22}{1}{F}{F}{pg}{F} &      \Chk  &      63.71 \\
& \drift &      \humanCfgtrans{oopsla25july22}{1}{T}{F}{pg}{F} &      \Chk  &      322.50 \\
& \drift &      \humanCfgtrans{oopsla25july22}{0}{F}{T}{ls}{F} &      \Unk  &      13.34 \\
& \drift &      \humanCfgtrans{oopsla25july22}{1}{F}{T}{ls}{F} &      \Unk  &      15.75 \\
& \drift &      \humanCfgtrans{oopsla25july22}{0}{F}{F}{pg}{F} &      \Unk  &      40.36 \\
& \drift &      \humanCfgtrans{oopsla25july22}{1}{T}{T}{ls}{F} &      \Unk  &      133.53 \\
& \evdrift &      \humanCfgdirect{oopsla25july22}{1}{T}{F}{pg}{T} &      \Chk  &      17.56 \\
& \evdrift &      \humanCfgdirect{oopsla25july22}{1}{F}{T}{ls}{T} &      \Unk  &      3.14 \\
& \evdrift &      \humanCfgdirect{oopsla25july22}{0}{F}{F}{pg}{T} &      \Chk  &      6.08 \\
& \evdrift &      \humanCfgdirect{oopsla25july22}{1}{F}{F}{pg}{T} &      \Chk  &      12.12 \\
& \evdrift & \hl  \humanCfgdirect{oopsla25july22}{0}{F}{F}{pg}{F} & \hl  \Chk  & \hl  6.06 \\
& \evdrift &      \humanCfgdirect{oopsla25july22}{1}{T}{T}{ls}{T} &      \Unk  &      5.17 \\
& \evdrift &      \humanCfgdirect{oopsla25july22}{0}{F}{T}{ls}{T} &      \Unk  &      2.34 \\
& \evdrift &      \humanCfgdirect{oopsla25july22}{1}{F}{T}{ls}{F} &      \Unk  &      3.12 \\
& \evdrift &      \humanCfgdirect{oopsla25july22}{1}{T}{F}{pg}{F} &      \Chk  &      17.54 \\
& \evdrift &      \humanCfgdirect{oopsla25july22}{0}{F}{T}{ls}{F} &      \Unk  &      2.36 \\
& \evdrift &      \humanCfgdirect{oopsla25july22}{1}{F}{F}{pg}{F} &      \Chk  &      12.02 \\
& \evdrift &      \humanCfgdirect{oopsla25july22}{1}{T}{T}{ls}{F} &      \Unk  &      5.17 \\
\hline

\end{longtable}
}
\fi

\end{document}